\documentclass[a4paper,leqno]{amsart}

\usepackage{amsfonts}
\usepackage{mathrsfs}

\usepackage{amsmath}

\usepackage{latexsym}
\usepackage{amsthm}
\usepackage{enumitem}

\usepackage{wasysym}

\usepackage{hyperref}
\hypersetup{%
    pdfborder = {0 0 0},
    colorlinks,
    citecolor=red!50!black,
    filecolor=Darkgreen,
    linkcolor=blue!40!black,
    urlcolor=cyan!50!black!90
}

\usepackage{tikz}
\newcommand{\btp}{\begin{tikzpicture}[baseline=-.25em,scale=0.25,line width=0.7pt]}
\newcommand{\etp}{\end{tikzpicture}}

\allowdisplaybreaks[1]

\numberwithin{equation}{section}

\newtheorem{thrm}{Theorem}[section]

\newtheorem{crl}[thrm]{Corollary}
\newtheorem{lemma}[thrm]{Lemma}

\theoremstyle{definition}
\newtheorem{defn}[thrm]{Definition}

\theoremstyle{remark}
\newtheorem{rmk}[thrm]{Remark}

\newcommand{\rmkend}{\ensuremath{\diameter}}

\newcommand{\defnend}{\ensuremath{\diameter}}

\setlist[itemize,1]{leftmargin=.4in}
\setlist[enumerate,1]{leftmargin=.4in,label=(\roman*)}
\setlist[description,1]{leftmargin=.4in,font=\normalfont\itshape}

\newcommand{\nc}{\newcommand}
\newcommand{\rnc}{\renewcommand}

\nc{\al}{\alpha}
\nc{\eps}{\epsilon}
\nc{\veps}{\varepsilon}
\nc{\ga}{\gamma}
\nc{\ka}{\kappa}
\nc{\la}{\lambda}
\nc{\del}{\delta}
\nc{\si}{\sigma}

\nc{\id}{\mathrm{id}}
\nc{\Id}{\mathrm{Id}}
\nc{\gr}{\mathrm{gr}}
\rnc{\t}{\mathrm{t}}
\nc{\rk}{\mathrm{rank}}

\nc{\flL}{\phi_1}
\nc{\flR}{\phi_2}
\nc{\flLR}{\phi_{12}}

\nc{\tl}{\tilde}

\nc{\mfsl}{\mf{sl}}

\nc{\Aut}{{\rm Aut}}
\nc{\Out}{{\rm Out}}

\nc{\End}{\mathrm{End}}
\nc{\Ext}{\mathrm{Ext}}
\nc{\GL}{\mathrm{GL}}

\nc{\mf}{\mathfrak}
\nc{\mc}{\mathcal}
\nc{\ms}{\mathsf}
\nc{\bb}{\mathbb}

\nc{\wh}{\widehat}
\nc{\wt}{\widetilde}
\nc{\lra}{\longrightarrow}
\nc{\ra}{\rightarrow}
\nc{\into}{\hookrightarrow}
\nc{\onto}{\twoheadrightarrow}

\nc{\N}{\mathbb{N}}
\nc{\Z}{\mathbb{Z}}
\nc{\Q}{\mathbb{Q}}
\nc{\R}{\mathbb{R}}
\nc{\C}{\mathbb{C}}
\nc{\K}{\mathbb{K}}
\nc{\A}{\mathbb{A}}

\nc{\ot}{\otimes}
\nc{\op}{\oplus}

\nc{\equ}[1]{\begin{equation}#1\end{equation}}
\nc{\eqa}[1]{\begin{equation}\begin{alignedat}{50}#1\end{alignedat}\end{equation}}
\nc{\eqn}[1]{\begin{equation*}\begin{alignedat}{50}#1\end{alignedat}\end{equation*}}
\nc{\eqg}[1]{\begin{equation}\begin{gathered}#1\end{gathered}\end{equation}}

\nc{\ali}[1]{\begin{alignat}{50}#1\end{alignat}}
\nc{\als}[1]{\begin{subequations}\begin{alignat}{50}#1\end{alignat}\end{subequations}}
\nc{\aln}[1]{\begin{alignat*}{50}#1\end{alignat*}}

\nc{\gat}[1]{\begin{gather}#1\end{gather}}
\nc{\gas}[1]{\begin{subequations}\begin{gather}#1\end{gather}\end{subequations}}
\nc{\gan}[1]{\begin{gather*}#1\end{gather*}}

\nc{\tx}[1]{\qu\text{#1}\qu}

\nc{\eqrefs}[2]{\text{(\ref{#1}-\ref{#2})}}

\nc\el{\nonumber\\}
\nc\nn{\nonumber}

\nc{\qu}{\quad}
\nc{\qq}{\qquad}

\usepackage[scr=boondoxo]{mathalfa}
\nc{\key}{{\mathscr{k}}}
\nc{\ley}{{\mathscr{l}}}

\nc{\bi}{\bar\imath}
\nc{\bj}{\bar\jmath}
\nc{\bk}{\bar k}
\nc{\bl}{\bar l}
\nc{\tq}{\tl q}
\nc{\Rat}{{\rm Rat}}
\nc{\Arg}{{\rm Arg}}

\nc{\RC}{R^{\vee}}

\nc{\tsp}{\hspace{.05em}}

\nc{\cid}{{{c.i.d.\ }}}

\renewcommand{\,}{\kern 0.1em} 

\rnc\appendixname{}

\begin{document}

\title{Solutions of the $U_q(\wh\mfsl_N)$ reflection equations}

\begin{abstract} \small 
We find the complete set of invertible solutions of the untwisted and twisted reflection equations for the Bazhanov-Jimbo $R$-matrix of type ${\rm A}^{(1)}_{N-1}$. 
We also show that all invertible solutions can be obtained by an appropriate affinization procedure from solutions of the constant untwisted and twisted reflection equations.
\end{abstract}

\author{Vidas Regelskis}
\address{Department of Mathematics, University of York, York, YO10 5DD, UK and \newline \mbox{\hspace{.35cm}} Institute of Theoretical Physics and Astronomy, Vilnius University, Saul\.etekio av.~3, Vilnius 10257, Lithuania}
\email{vidas.regelskis@gmail.com}

\author{Bart Vlaar}
\address{
Department of Mathematics, University of York, York, YO10 5DD, UK and \newline \mbox{\hspace{.31cm}} Department of Mathematics, Heriot-Watt University, Edinburgh, EH14 4AS, UK}
\email{b.vlaar@hw.ac.uk}

\maketitle

\setlength{\parskip}{1ex}



\section{Introduction} \label{Sec:intro}

Let $V$ be a complex vector space and $R$ an $\End(V^{\ot2})$-valued meromorphic function of one variable, called the \emph{spectral parameter}, satisfying the parameter-dependent Yang-Baxter equation
\equ{ \label{YBE}
R_{12}(\tfrac uv)\,R_{13}(\tfrac uw)\,R_{23}(\tfrac vw) = R_{23}(\tfrac vw) \,R_{13}(\tfrac uw)\,R_{12}(\tfrac uv)
}
which is to be understood as an identity for $\End(V^{\ot3})$-valued meromorphic functions of one variable. 
Here the lower indices specify the embedding $\End(V^{\ot2}) \into \End(V^{\ot3})$. 
Equation \eqref{YBE} can be seen as a condition on the interaction of certain particles with internal state space $V$, see e.g. \cite{ZaZa}. 
It expresses that two possible factorizations of the interaction of three such particles into simple interactions (i.e., interactions of just two particles) are equivalent, which underpins quantum integrability in many physical models.

Given a solution $R$ to \eqref{YBE}, it is an interesting question to classify all meromorphic $\End(V)$-valued functions of one variable $K, \wt K$ satisfying
\gat{ 
R_{21}(\tfrac{u}{v})\, K_1(u)\, R(uv)\, K_2(v) = K_2(v)\, R_{21}(uv)\, K_1(u)\, R(\tfrac{u}{v}) , \label{RE}  \\
R(\tfrac{u}{v})\, \wt{K}_1(u)\, R^{{\t}_1}(\tfrac{1}{uv})\, \wt{K}_2(v) = \wt{K}_2(v)\, R^{{\t}_1}(\tfrac{1}{uv})\, \wt{K}_1(u)\, R(\tfrac{u}{v}) , \label{tRE} 
}
where $\rm t_1$ denotes the usual transposition of matrices in the first tensor leg.
These equations are known as the \emph{untwisted} and \emph{twisted parameter-dependent reflection equation}, respectively, and arise naturally as conditions on the interaction of the aforementioned particles with a boundary, see \cite{Ch,Sk,KuSk}. 
In this case the two possible factorizations of the interaction of two particles with one boundary into simple interactions (i.e., particle-particle and particle-boundary) are required to be equivalent.

In the present paper $R$ will denote the Bazhanov-Jimbo $R$-matrix of type ${\rm A}^{(1)}_{N-1}$ \cite{Ba,Ji1}, for which $V = \C^N$ and the dependence on $u$ is rational.
For the case $N\ge3$, we will use the method of separation of variables to find all invertible solutions $K, \wt K \in \End(\C^N)$ of \eqrefs{RE}{tRE}.

The $N=2$ case is well-known. 
In this case there is a natural one-to-one correspondence between solutions of \eqref{RE} and solutions of \eqref{tRE}; moreover, there exists a general solution \cite{dVGR}, which can be specialized to any invertible solution, for example the general diagonal solution of \eqref{RE} obtained in \cite{Ch,Sk}. 
The story is very different if $N\ge 3$. 
The untwisted and twisted reflection equations are not equivalent and in either case there is not a single general solution which can be specialized to obtain an arbitrary invertible solution. 

To our knowledge, the classification of the invertible solutions of \eqrefs{RE}{tRE} is not known and it is our main goal to obtain this. 
Many classes of particular solutions are known in the literature.
For example, in \cite{MLS} a wide class of symmetric solutions of \eqref{RE} was found for all Bazhanov-Jimbo $R$-matrices of classical affine Lie type. 
In the case ${\rm A}^{(1)}_{N-1}$ it generalizes the solution found in \cite{AbRi} but does not encompass all the solutions. 
In \cite{CGM} an affinization procedure was used to construct a class of solutions of \eqref{RE} from the invertible solutions of the constant reflection equation found by \cite{DNS}. 
Regarding the twisted reflection equation \eqref{tRE}, in \cite{Gan} a solution of \eqref{tRE} with all entries nonzero was found.
In \cite{MRS} it was shown that solutions of type AI and AII of the constant twisted reflection equation, found in \cite{NoSu}, are also solutions of \eqref{tRE}. 

In the present paper we show that the set of solutions is naturally partitioned into equivalence classes identified in Lemma \ref{L:K-symm}.
This extends, explains and makes precise the observation in \cite[eqns.~(39-40)]{MLS} that there exists a transformation related to the cyclic group of order $N$ acting on the set of solutions of \eqref{RE} found in {\it ibid}.
Our main results are Theorems \ref{T:K} and \ref{T:CK}. 
The first theorem states that any invertible solution of the untwisted reflection equation \eqref{RE} is equivalent to a non-diagonal symmetric matrix given by the formula \eqref{KS} or a ``triangular'' matrix given by formula \eqref{KT}.
We remark that formula \eqref{KS}, in a slightly different formulation, was first reported in \cite{RV}. 
Solutions found in \cite{MLS} and \cite{CGM} are special cases of \eqref{KS} up to equivalence. 
To our best knowledge, the family of solutions given by formula \eqref{KT} was not known before, with the exception of some low rank cases.

For example, when $N=4$, a choice of inequivalent invertible non-diagonal symmetric solutions of \eqref{RE} is 
\[
\setlength{\arraycolsep}{3pt}
\renewcommand{\arraystretch}{0.8}
\begin{pmatrix} a \\ & a \\ && b& d \\ && d & c \end{pmatrix},\qu
\begin{pmatrix} a \\ & b && d \\ && 1 \\ &d&& c \end{pmatrix},\qu
\begin{pmatrix} b &&& d \\ & 1 \\ && 1 \\ d &&& c \end{pmatrix},\qu
\begin{pmatrix} b &&& d \\ & b & d \\ & d& c \\ d &&& c \end{pmatrix},
\]
where 
\[
a= 1+ \theta, \qu 
b= 1+\la \theta \chi, \qu 
c= 1+\la^{-1} \theta \chi, \qu 
d= - \theta \chi, \qu
\theta = \frac{u-u^{-1}}{\la^{-1}\mu^{-1}+u^{-1}} , \qu
\chi = \frac{1}{\la-\mu u} ,
\]
with free parameters $\la,\mu\in\C^\times$ satisfying $0 \le \Arg(\la), \Arg(\mu) < \pi$. 
A choice of inequivalent triangular solutions is
\gan{
\setlength{\arraycolsep}{3pt}
\renewcommand{\arraystretch}{0.8}
\begin{pmatrix} 1&&& \\ & 1&& \\ && 1& \\ &&& 1 \end{pmatrix},\qu
\begin{pmatrix} 1&&& \\ & 1&& \\ && 1& \\ &&& a \end{pmatrix},\qu
\begin{pmatrix} 1&&& \\ & 1&& \\ && a& \\ &&& a \end{pmatrix},\qu
\begin{pmatrix} 1&&&b \\ & 1&& \\ && 1& \\ &&& a \end{pmatrix},\qu
\begin{pmatrix} 1 \\ & 1&&b \\ &&1 \\ &&& a \end{pmatrix},\qu
\begin{pmatrix} 1 \\ & 1&& \\ && 1&b \\ &&& a \end{pmatrix},
\\[1em]
\setlength{\arraycolsep}{3pt}
\renewcommand{\arraystretch}{0.8}
\begin{pmatrix} 1&&b& \\ & 1&& \\ && a& \\ &&& a \end{pmatrix},\qu
\begin{pmatrix} 1&&&b \\ & 1&& \\ && a& \\ &&& a \end{pmatrix},\qu
\begin{pmatrix} 1&&& \\ &1&&b \\ && a& \\ &&& a \end{pmatrix},\qu
\begin{pmatrix} 1&&&b \\ & 1&b& \\ && a& \\ &&& a \end{pmatrix}, \qu
\begin{pmatrix} 1&&&b \\ & 1&& \\ &b & a& \\ &&& a \end{pmatrix},
}
where $a= \dfrac{\al-u^{-1}}{\al-u}$, $b = \dfrac{u-u^{-1}}{\al-u}$ with a free parameter $\al\in\C$ satisfying $0\le \Arg(\al) < \pi$. 
In the formulas above $u \in \C^\times$ is the spectral parameter.

Theorem \ref{T:CK} states that any invertible solution of the twisted reflection equation \eqref{tRE} (in a ``charge-conjugated'' form explained in Lemma \ref{L:CRE}) is equivalent to the dense matrix \eqref{CK:qOns} or one of the generalized permutation matrices \eqref{CK:A1}, \eqref{CK:A2} or \eqref{CK:A4}. 
The latter two are only defined for even $N$. 
Although the first three solutions are well-known \cite{Gan,MRS} and the solution \eqref{CK:A4} was found recently in \cite{RV}, heretofore the statement that these are all solutions to \eqref{tRE} was unproven. 

Again for $N=4$ a choice of inequivalent invertible solutions to \eqref{tRE} is as follows.
The solution \eqref{CK:qOns} is
\[
\setlength{\arraycolsep}{3pt}
\renewcommand{\arraystretch}{0.8}
\begin{pmatrix}
 a q u & -a \sqrt{-q} u & -a u & 1 \\
 -a \sqrt{-q} u & -a u & 1 & a q \\
 -a u & 1 & a q & -a \sqrt{-q} \\
 1 & a q & -a \sqrt{-q} & -a 
\end{pmatrix} 
\]
where $ a=\dfrac{q+1}{\sqrt{-q}(q+u)}$. The solutions \eqref{CK:A1}, \eqref{CK:A2} and \eqref{CK:A4} are, respectively,
\[
\setlength{\arraycolsep}{3pt}
\renewcommand{\arraystretch}{0.8}
\begin{pmatrix} &&&1 \\ &&1 \\ &1 \\ 1 \end{pmatrix},\qu
\begin{pmatrix} &&1 \\ &&&1 \\ 1 \\ &1 \end{pmatrix},\qu
\begin{pmatrix} &u \\ u \\ &&&1 \\ &&1 \end{pmatrix}.
\]
Again, $u \in \C^\times$ is the spectral parameter (varying in the equation) and $q\in\C^\times$ is the deformation parameter occurring in the $R$-matrix $R(u)$.

It is important to note that all the invertible solutions of the untwisted reflection equation as given by the formulas \eqref{KS} and \eqref{KT} can be obtained by an affinization of (possibly non-invertible) solutions of the constant reflection equation. 
This is stated precisely in Corollary~\ref{C:Aff}. 
A classification of these solutions was obtained by Mudrov \cite{Mu} (all invertible ones were found before in \cite{DNS}). 
A similar statement for the twisted reflection equation is discussed in Remark \ref{R:Aff}.
A classification of (possibly non-invertible) solutions of the constant twisted reflection equation will be presented elsewhere.

For small values of $N$, in \cite{RV} the symmetric and diagonal solutions of \eqref{RE} and all solutions of \eqref{tRE} were derived as intertwiners of representations of quantum symmetric Kac-Moody pair algebras \cite{Ko}. 
In particular, the matrix solutions provide one-dimensional representations (characters) of these algebras \cite{KoSt}. 
This raises the natural question: what are the underlying algebras for which the triangular solutions \eqref{KT} provide characters?
For the $N=2$ case this was addressed in \cite{BsBe}.

The paper is organized as follows. 
In Section \ref{sec:2} we find the group of symmetries of the Bazhanov-Jimbo $R$-matrix. 
We believe this must be known, however we were unable to locate a proof of this in the literature. 
Another reason to include this is that its proof is similar to, but easier than, the proofs given in the next section.  
Section \ref{sec:3} contains the main results of the paper, Theorems~\ref{T:K} and~\ref{T:CK}. 
They provide a classification of invertible solutions to the untwisted and twisted reflection equations. 
We note that \eqref{tRE} has an equivalent formulation \eqref{CtRE} which is more in line with \eqref{RE}. 
Our classification is written for this modified form.

{\it Acknowledgements.} The authors thank the referees for some helpful suggestions and Ana Ros Camacho for her academic tweets. V.R. was supported in part by the Engineering and Physical Sciences Research Council (EPSRC) of the United Kingdom, grant number EP/K031805/1 and by the European Social Fund, grant number 09.3.3-LMT-K-712-02-0017. B.V. was supported by EPSRC, grant numbers EP/N023919/1 and EP/R009465/1.
The authors gratefully acknowledge the financial support.


\section{The Bazhanov-Jimbo $R$-matrix and its symmetries} \label{sec:2}


\subsection{Definitions and Preliminaries}

Choose $N\in \Z_{\ge2}$. If~not stated otherwise, for any indices $i,j,k,l,\dots$ we will always assume that $1\le i,j,k,l\le N$. 
Fix a basis $\{ e_k \}_{k=1}^N$ of $\C^N$ and denote by $\{ e^*_k \}_{k=1}^N$ its dual: $e^*_i(e_j)=\del_{ij}$. 
Let $E_{ij}\in \End(\C^N) = \End_\C(\C^N)$ denote the standard unit matrices, so that $E_{ij} e_k = \del_{jk} e_i$. 
We will denote the usual transposition of matrices by $\t: E_{ij} \mapsto E_{ji}$ and the transposition with respect to the main antidiagonal by $w : E_{ij} \mapsto E_{\bj\tsp\bi}$, where $\bi := N-i+1$ for any $i \in \{ 1, \ldots, N \}$. 
Let $\mf{S}_N$ denote the symmetric group.
Recall that $M \in \GL(\C^N)$ is called a \emph{generalized permutation matrix} if $M \in \sum_{i=1}^N \C E_{i,\si(i)}$ for some $\si\in\mf{S}_N$.

Let $\Rat$ denote the algebra of rational functions in one complex variable. 
Their arguments will always be written in terms of complex numbers $u$ and $v$, known in the literature as \emph{spectral parameters}, or combinations of them such as $u/v$ and $uv$.
We write $\Rat^\times$ for the group $\Rat \backslash \{ 0 \}$. 
Given a complex vector space $V$, let $\Rat(V) := \Rat \ot_\C \End(V)$ denote the algebra of rational $\End(V)$-valued functions of one complex variable where $\End(V) = \End_\C(V)$.
We will consider the multiplicative group
\[
\Rat(V)^\times :=  \{ M \in \Rat(V) \, : \, M(u) \text{ is invertible for generic values of } u \in \C \}.
\]

Fix $q \in\C^\times$ not a root of unity. 
The Bazhanov-Jimbo $R$-matrix of type ${\rm A}^{(1)}_{N-1}$ \cite{Ba,Ji1} is the element $R$ of $\Rat(\C^N\ot\C^N)^\times$ defined by
\equ{
R(u) = f_q(u)\,R_q + f_{q^{-1}}(u^{-1}) P R_{q^{-1}} P, \qu\text{where}\qu f_q(u) = \frac{1}{q-q^{-1} u},
\label{R(u)}
}
$P$ is the permutation operator and $R_q$ is a constant $R$-matrix given by
\ali{
P = \sum_{1\le i,j\le N} E_{ij} \ot E_{ji} , \qq
R_q = \sum_{1\le i,j \le N} \left( q^{\delta_{ij}}E_{ii}\ot E_{jj} + \delta_{i<j} (q-q^{-1}) E_{ij}\ot E_{ji} \right) . \label{P-Rq}
}
The matrix-valued function $R$ defined above is a solution to the \emph{parameter-dependent quantum Yang-Baxter equation} \eqref{YBE} on $(\C^N)^{\ot 3}$, see e.g.~\cite{Ji1}. We will need a few properties of $R$ later on:

\begin{lemma} 
Let $J$ be any invertible antidiagonal matrix.
For generic $u \in \C$ we have:
\gat{  
\label{R21} J_1 J_2 R(u) J_1^{-1} J_2^{-1} = R_{21}(u) = R^\t(u) \\
\label{Rbar} R(u^{-1})|_{q \to q^{-1}} = R_{21}(u).
}
\end{lemma}

\begin{proof}
Using \eqref{R(u)} we obtain \eqref{R21} as an immediate consequence of $J_1 J_2 R_q J_1^{-1} J_2^{-1} = PR_qP = R_q^\t$.
This in turn follows by a direct computation: we have
\begin{align*}
J_1 J_2 R_q J_1^{-1} J_2^{-1} = R_q^\t &=  \sum_{1\le i,j \le N} \left( q^{\delta_{ij}}E_{\bar \imath \bar \imath}\ot E_{\bar \jmath \bar \jmath} + \delta_{i<j} (q-q^{-1}) E_{\bar \imath \bar \jmath}\ot E_{\bar \jmath \bar \imath} \right),  \\
PR_qP &= \sum_{1\le i,j \le N} \left( q^{\delta_{ij}}E_{jj}\ot E_{ii} + \delta_{i<j} (q-q^{-1}) E_{ji}\ot E_{ij} \right),
\end{align*}
which are seen to be equal if we make the substitution $(i,j) \mapsto (\bar \jmath, \bar \imath)$.
Finally, \eqref{Rbar} follows immediately from \eqref{R(u)}.
\end{proof}


\subsection{Symmetries}

We say that $Z \in \Rat(\C^N)^\times$ is a \emph{symmetry of $R$} if 
\equ{
[ R(\tfrac uv), Z(u)\ot Z(v) ] = 0 \label{RZZ}
}
for generic $u,v$.
The set of such $Z$ is a subgroup of $\Rat(\C^N)^\times$. 

\begin{lemma} \label{L:R-symm}
The group of symmetries of $R$ is generated by rational functions times the identity matrix, constant invertible diagonal (c.i.d.) matrices and $Z^\rho \in \Rat(\C^N)^\times$ defined by
\equ{
Z^\rho(u) = \sum_{1\le i < N} E_{i,i+1} + u\tsp E_{N1}.  \label{Z}
}
\end{lemma}

\begin{proof}
%
%
($\Longrightarrow$) We begin by showing that the symmetry relation \eqref{RZZ} holds if $Z$ is one of the three indicated generators.
Clearly \eqref{RZZ} is satisfied if $Z$ is a rational function times the identity matrix. 
Set $\check R(u) := P R(u)$. Then \eqref{RZZ} is equivalent to
$
\check{R}(\tfrac uv)\,Z(u)\ot Z(v) - Z(v)\ot Z(u)\,\check R(\tfrac uv) = 0.
$
Observe that
\equ{
\check{R}(\tfrac uv) (e_k \ot e_l) = a_{kl}(\tfrac uv)\, e_l \ot e_k + b_{kl}(\tfrac uv)\, e_k \ot e_l , \label{Z0}
}
where 
\eqg{
a_{ij}(\tfrac uv) = q^{\del_{ij}} f_q(\tfrac uv) + q^{-\del_{ij}} f_{q^{-1}}(\tfrac vu), 
\\
b_{ij}(\tfrac uv) = (q-q^{-1})\big( \del_{i>j}\, f_q(\tfrac uv) - \del_{i<j}\,f_{q^{-1}}(\tfrac vu) \big) . \label{ab}
}
In particular, $a_{ii}(\tfrac uv)=1$, $b_{ii}(\tfrac uv)=0$ and $b_{ij}(\tfrac uv)/b_{ji}(\tfrac uv)= v/u$ if $i>j$. 

Let $Z$ be a \cid matrix, i.e., $Z(u)=\sum_{1\le i \le N}\,d_i E_{ii}$ with $d_i\in\C^\times$. Then 
\[ 
Z(u)\ot Z(v) (e_k\ot e_l)  = d_k d_l (e_k \ot e_l).
\]
It is clear from \eqref{Z0} and the equality above that \eqref{RZZ} is indeed true in this case. 
Now let $Z = Z^\rho$. 
Then
\[ 
Z(u)\ot Z(v) (e_k\ot e_l)  =  (\del_{k>1}\,e_{k-1} + \del_{k1}\,u\, e_N) \ot (\del_{l>1}\,e_{l-1} + \del_{l1}\,v\, e_N).
\]
It is clear that, if $l,k>1$ or $k=l=1$, 
\[
\big(\check{R}(\tfrac uv)\,Z(u)\ot Z(v) - Z(v)\ot Z(u) \check{R}(\tfrac uv)\big)(e_k\ot e_l) = 0 .
\]
Let $k>1$ and $l=1$. Then
\aln{
& \big(\check{R}(\tfrac uv)\,Z(u)\ot Z(v) - Z(v)\ot Z(u) \check{R}(\tfrac uv)\big)(e_k\ot e_l) \\
& \qq = v\,\big(a_{k-1,N}(\tfrac uv)-a_{k1}(\tfrac uv)\big)\, e_N \ot e_{k-1} + \big(v\,b_{k-1,N}(\tfrac uv) - u\,b_{k1}(\tfrac uv)\big)\, e_{k-1} \ot e_N = 0 ,
}
since $k-1<N$ and $k>1$, so that $a_{k-1,N}(\tfrac uv)-a_{k1}(\tfrac uv)=0$ and $b_{k-1,N}(\tfrac uv)/b_{k1}(\tfrac uv) = u/v$. The case when $k=1$ and $l>1$ is checked in a similar way. Hence $Z^\rho$ is indeed a symmetry of $R$.


($\Longleftarrow$) Conversely, suppose that $Z(u) = \sum_{1\le i,j\le N} z_{ij}(u) E_{ij}$ with $z_{ij} \in \Rat$, so that $e^*_i\,Z(u)\,e_j = z_{ij}(u)$. Also observe that
\eqn{
(e^*_i \ot e^*_j)\, \check{R}(\tfrac uv) &= a_{ij}(\tfrac uv)\, e^*_j \ot e^*_i + b_{ij}(\tfrac uv)\, e^*_i \ot e^*_j .
}
Consequently,
\eqa{
&(e^*_i \ot e^*_j)\,\big(\check{R}(\tfrac uv)\, Z(u)\ot Z(v) - Z(v)\ot Z(u)\, \check{R}(\tfrac uv)\big)(e_k \ot e_l) \\
& \qq = \big( a_{ij}(\tfrac uv) - a_{kl}(\tfrac uv) \big)\, z_{il}(v)\, z_{jk}(u) + b_{ij}(\tfrac uv)\, z_{ik}(u)\, z_{jl}(v) - b_{kl}(\tfrac uv)\, z_{ik}(v)\, z_{jl}(u). \label{Z1}
}
Next, in two steps, we prove that if $Z$ is a symmetry of $R$ then $Z(u)$ is a product of a rational function, a \cid matrix and $(Z^\rho(u))^{m}$ for some $0\le m < N$. 
We do not need to consider greater powers because $(Z^\rho(u))^N=u I$. 


{\noindent\it Step 1: $Z(u)$ is a generalized permutation matrix.}

Let $i=k\ne j=l$. Then \eqref{Z1} gives $z_{ii}(u)\,z_{jj}(v) - z_{ii}(v)\,z_{jj}(u) = 0$ implying  $z_{jj}(u) = c_j z_{11}(u)$ with $c_j \in\C$ for all $1 < j \le N$. 
We focus on the $N\ge3$ case first. 
Let $i\ne j\ne k \ne i$ and $l=j$. 
Then \eqref{Z1} gives
\[
c_j \,(\del_{i<j}\,u + \del_{i>j}\,v)\,z_{ik}(u) - (\del_{k<j}\,u + \del_{k>j}\,v)\,z_{ik}(v) = 0.
\]
Hence $c_j=0$ or $z_{ik}(u)=0$. 
It follows that 
\[
Z(u) \in z_{11}(u) \sum_{1\le i\le N} \C E_{ii} \qu\text{or}\qu Z(u) \in \sum_{1\le i\ne j\le N} z_{ij}(u)E_{ij} .  
\]
The first case yields a rational function times a \cid matrix and hence a generalized permutation matrix. 
The second case requires a little bit more consideration. 
Let $i=j\ne k \ne l \ne j$ or $l = i\ne j\ne k \ne i$. 
Then \eqref{Z1} gives
\gan{
(q v+u)\,z_{ik}(u)\,z_{il}(v)- (\del_{k<l}\,u + \del_{k>l}\,v)\,(q+1)\,z_{ik}(v)\,z_{il}(u)  = 0,
\\
(q v+u)\,z_{ik}(v)\,z_{jk}(u) - (\del_{i<j}\,u + \del_{i>j}\,v)\,(q+1)\,z_{ik}(u)\,z_{jk}(v)  = 0 .
}
Setting $u=-qv$ this yields that $z_{ik}(v)\, z_{il}(-qv) = 0$ and $z_{ik}(-qv)\, z_{jk}(v) = 0$ for generic values of $v$. 
Hence there can be at most one nonzero off-diagonal element in each row and each column of $Z(u)$. 
Together with the requirement that $Z \in \Rat(\C^N)^\times$, this implies that $Z(u)$ is a generalized permutation matrix for generic $u$ and hence for all $u$ in the domain of $Z$.  
It remains to consider the $N=2$ case. Let $i\ne j=k=l$ and $u=-qv$. Then \eqref{Z1} gives $(\del_{i<j}\,q-\del_{i>j})\,c_j\,z_{ij}(-q v)=0$ implying $Z(u)=c_1 E_{11} + c_2 E_{22}$ or $Z(u)=z_{12}(u) E_{12} + z_{21}(u) E_{21}$, as required.


{\noindent\it Step 2: For generic values of $u$, $Z(u)$ is a product of a rational function, a \cid matrix and $(Z^\rho(u))^{m}$ for some $0\le m < N$.}

Taking Step 1 into account, it suffices to show that if $Z(u)$ is not a rational function times a \cid matrix, then it is of the form indicated with $1 \le m < N$.
We have $Z(u) = \sum_{1\le i \le N} z_{i,\si(i)}(u)\, E_{i,\si(i)}$ for some $\si \in \mf{S}_N$. 
Assuming that $i\ne j$ and $k\ne l$ \eqref{Z1} gives
\[
(\del_{i<j}\,u + \del_{i>j}\,v)\, z_{ik}(u)\, z_{jl}(v) - (\del_{k<l}\,u + \del_{k>l}\,v)\,z_{ik}(v)\,z_{jl}(u) = 0 .
\]
Let $z_{ij}(u)$ with $i\ne j$ be nonzero and let $k,l$ be arbitrary. 
Then the equality above is equivalent~to 
\equ{
\frac{z_{i\pm k,j\pm l}(u)}{z_{ij}(u)}\in\C, \qq \frac{z_{i\pm k,j\mp l}(u)}{z_{ij}(u)}\in\C u^{\pm1} . \label{Z2}
}
For $N=2$ the expressions above imply that $Z(u)\in \C(u)(E_{12} + \C u E_{21})$ and the proof is complete.
Otherwise, if $N\ge 3$, assume that coefficients $z_{pk}(u)$ and $z_{rl}(u)$ with $i<p<r\le N$ are also nonzero. 
By \eqref{Z2} we have 
\[
\frac{z_{pk}(u)}{z_{ij}(u)} \in \C (\del_{k<j}\,u + \del_{k>j}) , \qu\;
\frac{z_{rl}(u)}{z_{ij}(u)} \in \C (\del_{l<j}\,u + \del_{l>j}) , \qu\;
\frac{z_{rl}(u)}{z_{pk}(u)} \in \C (\del_{l<k}\,u + \del_{l>k}) .
\]
The conditions above are satisfied if and only if $j<k<l$ or $l<j<k$ or $k<l<j$. For $N=3$ only the latter two cases are possible yielding, up to multiplication by a \cid matrix and a rational function in $u$, $Z(u) = E_{12}+E_{23}+ u E_{31} = Z^\rho(u)$ or $Z(u) = E_{21}+ u E_{21}+ u E_{32} = (Z^\rho(u))^2$, which completes the proof. For $N\ge 4$ assume that coefficients $z_{ij}(u)$ and $z_{i+1,k}(u)$ are nonzero. If $k>j+1$ there must exist a nonzero coefficient $z_{ab}(u)$ with $1\le a<i$ or $i+1<a\le N$ and $j<b<k$ yielding a contradiction. If $1<k<j$ there must exist a nonzero coefficient $z_{ab}(u)$ with $1\le a<i$ or $i+1<a\le N$ and $b=1$ yielding a contradiction once again. Hence, if $z_{ij}(u)$ is nonzero, $z_{i+1,k}(u)$ can only be nonzero if $k=j+1$ or $k=1$. Clearly, if $j=1$ then the only option is $k=2$. Similarly, if $j=N$, then $k=1$. Let $1<j<N$ and $k=1$. But then there must exist a nonzero coefficient $z_{ab}(u)$ with $1\le a < i$ or $i+1<a\le N$ and $b>j$, which yields a contradiction. Therefore, if $z_{ij}(u)$ is nonzero, then $z_{i+1,k}(u)$ is also nonzero only if $k=j+1\le N$ or $k=1$ and $j=N$. It follows that a generically non-diagonal $Z$ is a symmetry of $R$ only if, up to multiplication by a \cid matrix and a rational function in $u$, 
\[
Z(u) = \sum_{m< j\le N} E_{j-m,j} + \sum_{1\le j \le m} u E_{N+j-m,j} = (Z^{\rho}(u))^m
\]
for some $1\le m < N$. This completes the proof.   
\end{proof}

\begin{rmk} \label{R:Z}
The symmetry $Z^\rho$ corresponds to the Hopf algebra automorphism of the quantum affine algebra $U_q(\wh\mfsl_N)$ given by an elementary rotation of the Dynkin diagram. 
The symmetries given by diagonal matrices (with determinant 1) correspond to the Hopf algebra automorphisms that multiply positive non-affine simple root vectors by nonzero complex numbers; the negative simple root vectors are multiplied by the corresponding inverses.  \hfill \rmkend
\end{rmk}


\section{Solutions of the untwisted and twisted reflection equations} \label{sec:3}


\subsection{Equivalence of solutions}

Our goal is find all $K, \wt{K} \in \Rat(\C^N)^\times$ satisfying \eqref{RE} and \eqref{tRE}, respectively, which we recall here:%
\aln{
R_{21}(\tfrac{u}{v})\, K_1(u)\, R(uv)\, K_2(v) &= K_2(v)\, R_{21}(uv)\, K_1(u)\, R(\tfrac{u}{v}),
\\
R(\tfrac{u}{v})\, \wt{K}_1(u)\, R^{{\t}_1}(\tfrac{1}{uv})\, \wt{K}_2(v) &= \wt{K}_2(v)\, R^{{\t}_1}(\tfrac{1}{uv})\, \wt{K}_1(u)\, R(\tfrac{u}{v}).   
}
It will be convenient to rewrite the twisted reflection equation \eqref{tRE} in a \emph{cross-conjugated} form similar to \eqref{RE}. (This form has a natural braided counterpart, which will be used to prove the classification theorem.) 
We introduce the antidiagonal matrix%
\footnote{In the physics literature it is called the cross or charge conjugation matrix, see e.g., \cite{Ba}.
}
\[
C = \sum_{1\le i \le N} (-q)^{i} E_{\bi i} 
\]
and the $C$-conjugated $R$-matrix
\equ{
\RC(u) = C_2^{-1} R^{t_1}( \tq^{\,2} u^{-1})\,C_2 \qq\text{where}\qq \tq = (-q)^{N/2} . \label{RC(u)}
}
In particular, cf.~\eqref{R(u)},
\equ{
\RC(u) = f_q( \tq^{\,2} u^{-1})\, \RC_{q} + f_{q^{-1}}( \tq^{\,-2} u)\,P\, \RC_{q^{-1}} P , \label{RC(u)-aff}
}
where
\equ{
\RC_q = C_2^{-1} R^{t_1}_q\,C_2 = \sum_{1\le i,j \le N} \left( q^{\del_{ij}}E_{ii}\ot E_{\bj\tsp\bj} + \delta_{i<j} (-q)^{j-i}(q-q^{-1}) E_{ji}\ot E_{\bj\tsp\bi} \right) . \label{RCq}
}
\begin{rmk} \label{R:C} 
The matrix $C$ corresponds to the Hopf algebra automorphism of $U_q(\wh\mfsl_N)$ given by the reflection of the Dynkin diagram which fixes the affine node; see also Remark \ref{R:Z}. \hfill \rmkend
\end{rmk}


\begin{lemma} \label{L:CRE}
Define $K \in \Rat(\C^N)$ by
\equ{
K(u) = C^{-1} \wt{K}(\tq^{\,-1} u). \label{CK}
}
Then the twisted reflection equation \eqref{tRE} is equivalent to 
\equ{
R_{21}(\tfrac{u}{v})\, K_1(u)\, \RC(uv) \,K_2(v) = K_2(v)\, \RC_{21}(uv)\, K_1(u)\, R(\tfrac{u}{v}). \label{CtRE}
}
\end{lemma}

\begin{proof}
Upon substituting $u\to \tq^{\,-1} u$ and $v\to \tq^{\,-1} v$ and using the notation \eqref{RC(u)}, we obtain that \eqref{tRE} is equivalent to
\[
R(\tfrac{u}{v}) \, C_1 \, C_2 \, K_1(u) \, \RC(u v) \, K_2(v) = C_1 \, C_2 \, K_2(v) \, \RC_{21}(uv) \, K_1(u) \, R(\tfrac{u}{v}).
\]
Now \eqref{R21} gives the equivalence with \eqref{CtRE}.
\end{proof}


There are straightforward transformations on the set of solutions of \eqref{RE} and \eqref{CtRE}.

\begin{lemma} \label{L:K-symm} 
Suppose $K \in \Rat(\C^N)$ is a solution of \eqref{RE} or \eqref{CtRE}. 
Then the element of $\Rat(\C^N)$ defined by any of the following is also a solution of the same reflection equation:
\begin{enumerate}
\item $u \mapsto K(-u)$,
\item $u \mapsto g(u)\,K(u)$ for any $g \in \Rat^\times$, 
\item $u \mapsto \phi(Z(\tfrac{\eta}{u}))K(u)Z(\eta u)$, where $\phi$ is the inverse map in case of \eqref{RE} and the transposition $w$ in case of \eqref{CtRE}, for any $\eta \in \C^\times$ and any symmetry $Z$ of $R$, 
\item $u \mapsto \psi(K(u))$, where $\psi$ is the usual transposition in case of \eqref{RE} and the composition of the three operations, transposition $w$, $u \mapsto u^{-1}$ and $(-q)^{1/2} \mapsto (-q)^{-1/2}$, in case of \eqref{CtRE}.
\end{enumerate}
\end{lemma}

\begin{proof} 
Statements {(i)} and {(ii)} are obvious.
For \eqref{RE} statement {(iii)} is a special case of \cite[Prop.~2]{Sk}; for \eqref{CtRE} we need to use
\gan{
Z^w_2(\tfrac{\eta}{v}) \, \RC(uv)\,Z_1(\eta u) = Z_1(\eta u)\,\RC(uv)\,Z^w_2(\tfrac{\eta}{v}) , \\
Z^w_1(\tfrac{\eta}{v}) \, \RC_{21}(uv)\,Z_2(\eta u) = Z_2(\eta u)\, \RC_{21}(uv)\,Z^w_1(\tfrac{\eta}{v}) ,
}
which follow from \eqref{RZZ} and \eqref{RC(u)}.
Statement {(iv)} in the case of \eqref{RE} follows immediately by transposing the untwisted reflection equation \eqref{RE} and using \eqref{R21}.
In case of \eqref{CtRE} we need a little bit more work. 
Denote $J:= \sum_{1\le i\le N} E_{i\tsp\bi}$.
Then from $M^w = J M^\t J$ for any $M \in \End(\C^N)$ and \eqref{R21} we obtain
\equ{ \label{wR}
(w \otimes w)(R(u^{-1}))|_{q \to q^{-1}} = R_{21}(u)
}
Consequently, combining \eqref{Rbar} with $JCJ=(-q)^{N+1}C|_{q \to q^{-1}}$ and \eqref{R21} we also obtain
\equ{ \label{wRC}
(w \otimes w)(\RC_{21}(u^{-1}))|_{q \to q^{-1}} = \RC(u).
}
Now applying the antiautomorphism $w \otimes w$ to \eqref{RE} followed by the operations $u \mapsto u^{-1}$, $v \mapsto v^{-1}$ and $(-q)^{1/2} \mapsto (-q)^{-1/2}$ (the latter of which induces $q \mapsto q^{-1}$) and using \eqrefs{wR}{wRC} we obtain statement {(iv)} in the case of \eqref{CtRE}.
\end{proof}

Note that statements (i) and (iii) of the above Lemma correspond to Hopf algebra automorphisms of $U_q(\wh\mfsl_N)$. Statement (i) corresponds to the automorphism which multiplies the simple affine root vectors by $-1$. 
For (iii) see Remark \ref{R:Z}.

\begin{defn} 
We say that two solutions to the same reflection equation are \emph{equivalent} if they are related by a composition of the transformations defined in Lemma \ref{L:K-symm}. (It is straightforward to verify that this indeed defines an equivalence relation.) 
\end{defn}

In the remaining parts of this section we will provide a classification of solutions of \eqref{RE} and \eqref{CtRE} up to equivalence.


\subsection{Untwisted reflection equation}

In this section we provide a classification of invertible solutions of the reflection equation \eqref{RE}.

\begin{defn} \label{D:GIM-GCM}
Let $M \in \GL(\C^N)$.
We call $M$ a \emph{generalized involution matrix} or a \emph{generalized cross matrix} if
\equ{ \label{gim-gcm}
M \in \sum_{1\le i\le N} \C E_{i \sigma(i)} \qu\text{ or }\qu  M \in \sum_{1\le i\le N} (\C E_{ii} + \C E_{i \sigma(i)}) ,
}
respectively, for some involution $\sigma \in \mf{S}_N$. \hfill \defnend
\end{defn}

If the off-diagonal coefficients in \eqref{gim-gcm} are all nonzero, then the matrix $M$ is diagonally similar to a symmetric matrix. 
We shall see that solutions of \eqref{RE1} evaluate to generalized cross matrices, and more precisely to matrices which are either diagonally similar to a symmetric matrix or, up to a permutation of $\{ 1,\ldots,N \}$, a direct sum of a lower- and an upper-triangular matrix.

The solutions will be labelled by the sets 
\gan{
\Sigma^{\rm sym}_N = \big\{ (\ell,r,\si) \,:\, 0\le \ell < r \le \tfrac12(N+\ell) \text{ and } \si(i)=N+\ell-i+1 \text{ for } \ell<i\le r  \big\} , 
\\[.2em]
\Sigma^{\rm tri}_N = \big\{ (m,\si) \,:\, \tfrac12N\le m\le N \;\text{and}\; 0<\si(j)\le \si(i)\le m\; \text{for}\; m<i\le j\le N \;\text{s.t.}\; \si(i)\ne i, \si(j)\ne j 
\big\} ,
}
where it is understood that $\si\in\mf{S}_N$ is an involution satisfying $\si(i)=i$ unless otherwise specified in the definition of the sets above.

Let $N$ be even and suppose $(\tfrac12N,\si_1), (\tfrac12N,\si_2)\in \Sigma^{\rm tri}_N$. We say that involutions $\si_1$ and $\si_2$ are \emph{related}, if $\tilde \rho \si_1 = \si_2 \tilde \rho$ where $\tilde \rho \in \mf{S}_N$ is the involution determined by $\tilde \rho(i)=i+N/2$ if $i \le N/2$.
For example, when $N=4$, the involutions $\si_1=(14)$ and $\si_2=(23)$ are related. We will denote by $\Sigma^{\rm tri,\sim}_N$ a maximal subset of $\Sigma^{\rm tri}_N$ which has no related involutions. (Note that $\Sigma^{\rm tri,\sim}_N=\Sigma^{\rm tri}_N$ if $N$ is odd.)

\begin{thrm} \label{T:K} 
Let $N\ge 3$. 
Then $K \in\Rat(\C^N)^\times$ is a solution of the reflection equation \eqref{RE} if and only if it is equivalent to the solution given by
\equ{ 
I + \frac{u-u^{-1}}{\la^{-1} \mu^{-1} + u^{-1} } \,\Bigg( \sum_{1\le i \le \ell} E_{ii} + \frac{1}{\la - \mu u} \sum_{\ell<i\le r} \Big( \la E_{ii} + \la^{-1} E_{\si(i)\si(i)} - E_{i\si(i)} - E_{\si(i)i} \Big)\!\Bigg)  \label{KS}
}
with $(\ell,r,\si)\in \Sigma^{\rm sym}_{N}$ and $\la,\mu\in\C^\times$ with $ 0 \le  \Arg(\la) , \Arg(\mu) < \pi$, or to the solution given by
\eqa{
& I + \frac{u-u^{-1}}{\al-u}\sum_{m< i \le N} \Big( E_{ii} + \veps_{i\si(i)} E_{i\si(i)} + \veps_{\si(i)i} E_{\si(i)i} \Big) \label{KT}
}
with $(m,\si)\in \Sigma^{\rm tri,\sim}_{N}$ and $\al\in\C$ with $0\le \Arg(\al) < \pi$ and $\veps_{i\si(i)},\veps_{\si(i)i} \in \{0,1\}$ satisfying $\veps_{ii}=0$ and $\veps_{i\si(i)}+\veps_{\si(i)i}=1$ if $i\ne \si(i)$. Moreover, when $\si\ne\id$, $\veps_{j\si(j)}=1$ and $\veps_{\si(j)j}=0$, where $j = \min \{ 1\le i \le N :\, i<\si(i)\}$.
\end{thrm}

We call matrix \eqref{KS} the {\it canonical symmetric solution} of \eqref{RE}. We call \eqref{KT} the {\it canonical triangular solution}. More particularly, we call \eqref{KT} the {\it canonical diagonal solution} if $\si=\id$.

\begin{rmk} [{{\it Specializations of the canonical symmetric solution\,}}] \label{R:spec} 
In \eqref{KS} one can make various choices for and take limits of the parameters:

\vspace{-.75em} 
 
\begin{itemize}
\item Allowing $\ell=r$ one recovers (up to an overall scalar) the canonical diagonal solution \eqref{KT} with $m=\ell$ and $\al = \la^{-1} \mu^{-1}$.

\item Setting $\la=\mu=1$ yields the generalized involution matrix
\equ{
\sum_{1\le i \le \ell}u\,E_{ii} + \sum_{\ell< i\le r} \Big(E_{i\si(i)} + E_{\si(i)i}\Big) + \sum_{r<i\le N+\ell-r} E_{ii} .   \label{KP}
}

\item The $\mu\to0$ limit gives the identity matrix. The $\la\to 0$ limit gives the canonical diagonal solution \eqref{KT} with $m=N+\ell-r$ and $\al=0$. 

\item The $N=2$, $\ell=0$, $r=1$ case, upon conjugating with a constant diagonal matrix, is the general solution found in \cite{dVGR}. 
The type I solutions found in \cite{MLS} are equivalent to the $r-\ell=1$ cases; their type II ($N$ odd) solutions are equivalent to the $r-\ell=\tfrac12(N-1)$ cases; their type II ($N$ even) solution are equivalent to the $r-\ell=\tfrac12N$ and $r-\ell=\tfrac12N-1$ cases. 
The solution given by Lemma 6.2 in \cite{CGM} corresponds to the $\ell=0$ case.
\hfill \rmkend 
\end{itemize}
\end{rmk}

Below we state two technical lemmas that will assist us in proving Theorem \ref{T:K}.
The first lemma will need the following sets:
\aln{
\Sigma^{\rm sym'}_N &= \big\{ (\ell,r,t,\si) \;\;\;:\, 0\le \ell < r < \tfrac12(\ell+t),\; 2\le t\le N \text{ and } \si(i)=t+\ell-i+1 \text{ for } \ell<i\le r  \big\} , 
\\[.2em]
\Sigma^{\rm sym''}_N &= \left\{ (\ell,m,r,\si) \,:\,\begin{aligned} 
& 2\le 2\ell < m < r \le \tfrac12(m+N) \text{ and } \si(i)=m-i+1 \text{ for } 0<i\le m \\[-.2em] 
& \text{and } \si(i)=N+m-i+1 \text{ for } m<i\le N \end{aligned}\; \right\} , 
\\[.2em]
\Sigma^{\rm tri'}_N &= \left\{ (\ell,m,r,\si) \,:\, \begin{aligned} 
& 0\le \ell \le m \le r\le N \text{ and } 0<\si(j)\le\si(i)\le \ell \;\text{for}\; \ell<i\le j\le m \\[-.2em] 
& \text{and } r<\si(j)\le\si(i)\le N \text{ for } m<i\le j\le r \text{ s.t.\ } \si(i)\ne i,\,\si(j)\ne j \end{aligned} \right\} .
}
Again, it is understood that $\si\in\mf{S}_N$ is an involution satisfying $\si(i)=i$ unless otherwise specified.

\begin{lemma} \label{L:[K]} 
(i) If $K$ is given by \eqref{KS}, then the solutions of \eqref{RE} equivalent to $K$ are given by 
\ali{
g(u) \Bigg( I + \frac{u-u^{-1}}{\la^{-1}\mu^{-1}+u^{-1}} &\Bigg( \sum_{1\le i \le \ell} E_{ii} - \frac{1}{\la\mu u} \sum_{t< i \le N} E_{ii} \el & \qu + \sum_{\ell<i\le r} \frac{1}{\la-\mu u}\Big( \la E_{ii} + \la^{-1} E_{\si(i)\si(i)} - c_i E_{i\si(i)} - c_i^{-1}E_{\si(i)i}  \Big) \Bigg) \Bigg) \label{KS1}
}
with $(\ell,r,t,\si)\in\Sigma^{\rm sym'}_N$ or 
\ali{
g(u) \Bigg( I + \frac{u-u^{-1}}{\la^{-1}\mu^{-1}+u^{-1}} &\Bigg(  \sum_{1\le i \le \ell} E_{ii}  + \sum_{\ell<i\le m-\ell} E_{ii} \el & \qu + \frac{1}{\la-\mu u} \sum_{1\le i \le \ell}\Big( \mu^{-1} u E_{ii} + \la E_{\si(i)\si(i)} - c_i u E_{i\si(i)} - c_i^{-1} u E_{\si(i)i} \Big) \el & \qu +  \frac{1}{\la-\mu u} \sum_{m< i\le r} \Big( \la E_{ii} + \la^{-1} E_{\si(i)\si(i)} - c_{i} E_{i\si(i)} - c_{i}^{-1} E_{i\si(i)}\Big)\Bigg) \label{KS2}
} 
with $(\ell,m,r,\si)\in\Sigma^{\rm sym''}_N$ and $\la,\mu,c_i\in\C^\times$.

\noindent (ii) If $K$ is given by \eqref{KT}, then the solutions of \eqref{RE} equivalent to $K$ are given by 
\ali{
g(u) &\Bigg( \sum_{1\le i \le \ell} u\,E_{ii} + \sum_{\ell< i\le r} \frac{\al-u}{\al\,u-1}\, E_{ii} + \sum_{r< i\le N} u^{-1} E_{ii} \el
& \qu\; + \sum_{\ell < i \le r} (\del_{i\le m}\tsp u + \del_{i>m}) \frac{u-u^{-1}}{\al\,u-1}\,\Big(c_{i\si(i)} E_{i\si(i)} + c_{\si(i)i} E_{\si(i)i} \Big) \Bigg) \label{KT1}
}
with $(\ell,m,r,\si)\in\Sigma^{\rm tri'}_N$ and $\al,c_{i\si(i)},c_{\si(i) i} \in\C$ such that either $c_{i \si(i)}$ or $c_{\si(i) i}$ is 0 for each $\ell < i \le r$.

\noindent Everywhere above $g \in \Rat^\times$ is arbitrary.
\end{lemma}

\begin{proof}
This follows by a tedious computation using Lemma \ref{L:K-symm}. Hence we provide the idea behind the proof only. Recall from Lemma \ref{L:R-symm} that the group of symmetries of $R$ is generated by \cid matrices and $Z^\rho$ defined in \eqref{Z} satisfying $(Z^\rho(u))^N= u\,I$. Set $D=\sum_{1\le i \le N} d_i E_{ii}$ with $d_i \in \C^\times$. Denote the matrix \eqref{KS} by $K_{\ell,r}(u)$, \eqref{KS1} by $K'_{\ell,r,t}(u)$ and \eqref{KS2} by $K''_{\ell,m,r}(u)$. Then, applying parts (ii) and (iii) of Lemma \ref{L:K-symm}, we compute
\aln{
g(u)\,D^{-1} & (Z^\rho(\tfrac{\eta}{u}))^{-k} K_{\ell,r}(u) (Z^\rho(\eta u))^{k} D \\[.25em] 
& = 
\begin{cases}
K'_{\ell+k,r+k,N+k}(u) &\text{if }\; -\ell \le k\le 0, \\[.25em]
K''_{k,\ell+2k,r+k}(u) &\text{if }\;  0< k \le r-\ell, \\[.25em]
\frac{1+\la\mu u}{1+\la\mu u^{-1}} \,\big(K'_{\ell-r+k,k,r+k}(u)\big|_{\la\leftrightarrow \mu^{-1}}\big) &\text{if }\; r-\ell < k\le N- r ,\\[.25em]
\frac{1+\la\mu u}{1+\la\mu u^{-1}} \,\big(K''_{r-N+k,\ell-N+2k,k}(u)\big|_{\la\leftrightarrow \mu^{-1}}\big) &\text{if }\; N-r < k< N-\ell ,
\end{cases}
\intertext{with $c_i = \eta^{\veps}\, d_i^{-1} d_{\si(i)}$ and $\veps$ defined by}
\veps &= -1 \;\text{ if } \begin{cases}0<k\le r-\ell \text{ and } 1\le i\le k , \\ 
r-\ell<k\le N-r \text{ and } \ell-r+k<i\le k , \\ 
N-r<k< N-\ell \text{ and } \ell-N+2k<i\le k
\end{cases}\\
\veps &= 1 \qu\text{ if } 0<k\le r-\ell \text{ and } \ell+2k<i\le r+k , \\
\veps &= 0 \qu\text{ otherwise}.
}
Applying part (i) of Lemma \ref{L:K-symm}, we find
\[
K_{\ell,r}(-u) = K_{\ell,r}(u)\big|_{\mu\to - \mu} , \qq
D^{-1}K_{\ell,r}(-u) D = K_{\ell,r}(u)\big|_{\la\to - \la} 
\]
with $d_i=-1$ if $i\le r$ and $d_i=1$ if $i>r$, and so $\la,\mu\in\C^\times$ for \eqref{KS1} and \eqref{KS2}. Finally, note that \eqref{KS} is invariant under the usual transposition; for \eqref{KS1}, \eqref{KS2} it is equivalent to a conjugation with an appropriate \cid matrix $D$. 

In case of \eqref{KT} the arguments are analogous with the following two exceptions. 
Let $N$ be even and $(\tfrac12N,\si_1),(\tfrac12N,\si_2)\in\Sigma^{\rm tri}_N$ be such that $\si_1$ and $\si_2$ are related. 
Let $\pi(\si_1)$, $\pi(\si_2) \in \GL(\C^N)$ be the corresponding permutation matrices. 
Then $(Z^\rho(u^{-1}))^{-N/2}\, \pi(\si_1)\,(Z^\rho(u))^{N/2} = \pi(\si_2)$ by virtue of Lemma~\ref{L:K-symm}~(iii). 
Consequently, there is an equivalence between the elements in the set of solutions of \eqref{CtRE} associated to $(\tfrac12N,\si_1)$ and the one associated to $(\tfrac12N,\si_2)$. 
Finally, Lemma \ref{L:K-symm} (iv) allows us to fix the position (i.e., in the upper or in the lower triangular part) of one of the off-diagonal entries of $K(u)$; a choice is given by the last constraint in Theorem \ref{T:K}.
\end{proof}

The Lemma below will assist us in proving that the matrix \eqref{KS} is indeed a solution of \eqref{RE}.

\begin{lemma} [{See \cite[Prop.~2.2]{DNS} and \cite[Thm.~2]{Mu}}] \label{L:KS-const}
Define $G, Q \in \End(\C^N)$ by 
\equ{
G = \sum_{r<i<\si(r)} \la E_{ii} + \sum_{\ell< i \le r} \Big( E_{i\si(i)} + E_{\si(i)i} + (\la-\la^{-1}) E_{\si(i)\si(i)} \Big) , \qq Q=\sum_{1\le i \le \ell} E_{ii} \label{KQ}
}
with $(\ell,r,\si)\in\Sigma^{\rm sym}_N$ and $\la\in\C^\times$. Then
\equ{
(G-\la I)(G+\la^{-1}I) = 0 \label{KK} 
}
if $\ell=0$, and
\gat{
(G-\la I)(G+\la^{-1}I)G = 0, \label{KKK1} \\
Q^2=Q,\qu GQ=GQ=0,\qu Q=I + (\la-\la^2)G - G^2  \label{KKK2}
}
if $\ell>0$. Moreover, setting $\check{R}_q = P R_q$,
\gat{
\check{R}_q G_2 \check{R}_q G_2 = G_2 \check{R}_q G_2 \check{R}_q , \label{RE:C1} \\
\check{R}_q G_2 \check{R}_q Q_2 = Q_2 \check{R}_q G_2 \check{R}_q , \label{RE:C2} \\
\check{R}_q Q_2 \check{R}_q Q_2 - Q_2 \check{R}_q Q_2 \check{R}_q = (q-q^{-1})(\check{R}_q Q_2 - Q_2 \check{R}_q), \label{RE:C3} \\
\check{R}_q Q_2 \check{R}_q G_2 - G_2 \check{R}_q Q_2 \check{R}_q = (q-q^{-1})(Q_2 \check{R}_q G_2 - G_2 \check{R}_q Q_2) . \label{RE:C4}
}
\end{lemma}

\begin{proof}
The equalities \eqrefs{KK}{KKK2} and \eqrefs{RE:C1}{RE:C4} are verified by applying each of them to the vectors $e_i$ and $e_i\ot e_j$, respectively. 
For example, for \eqref{KK} and \eqref{KKK1} we write $G\,e_i = g_{ii}\,e_i + g_{\si(i)i}\,e_{\si(i)}$, where $g_{ii} = \del_{r<i<\si(r)} \la + \del_{\si(r)\le i\le N} (\la-\la^{-1})$ and $g_{\si(i)i} = \del_{\ell<i\le r} + \del_{\si(r)\le i\le N} $.  
Then, using properties of delta functions, we compute
\aln{
(G &-\la I)(G+\la^{-1}I)\,e_i \\ 
& = \big((g_{ii}-\la)(g_{ii}+\la^{-1})+g_{i\si(i)}g_{\si(i)i}\big)\,e_i + \big(g_{\si(i)i}\,(g_{ii}+\la^{-1}) + (g_{\si(i)\si(i)}-\la)\,g_{\si(i)i}\big)\,e_{\si(i)} 
\\
& = \big(g_{ii}^2 - g_{ii}(\la-\la^{-1}) + g_{\si(i)i} -1 \big)\,e_i  + \big(g_{\si(i)i}\,(g_{ii}+g_{\si(i)\si(i)} - \la +\la^{-1}) \big)\,e_{\si(i)}  \\
& = (\del_{\ell<i\le r} + \del_{r<i<\si(r)} + \del_{\si(r)\le i\le N} - 1 )\,e_i  = -\del_{0<i\le \ell} \,e_i ,
}
which equals zero if $\ell=0$ thus proving \eqref{KK}. 
For \eqref{KKK1} we need to multiply from left with $G$ yielding zero as required. The remaining identities are verified in a similar way. 
\end{proof}

\begin{proof} [Proof of Theorem \ref{T:K}] 


($\Longrightarrow$) We begin by showing that matrices \eqref{KS} and \eqref{KT} are solutions of the reflection equation \eqref{RE}. It will be convenient to consider \eqref{RE} in the braided form,
\equ{ 
\check{R}(\tfrac{u}{v})\, K_2(u)\, \check{R}(uv)\, K_2(v) = K_2(v)\, \check{R}(uv)\, K_2(u)\, \check{R}(\tfrac{u}{v}) \label{RE1} .
}
Using the characteristic identity $(\check{R}_q - q\,I)(\check{R}_q + q^{-1} I) = 0$ we rewrite the $R$-matrix $\check{R}(u)$ as
\equ{
\check{R}(u) = f_q(u) \Big( (1-u) \check{R} + (q-q^{-1})\,u\,I\Big) \label{R(u):check}
}
and, using \eqref{KQ}, we rewrite the matrix \eqref{KS} as 
\equ{
K(u) = I + \frac{u-u^{-1}}{(\la^{-1}\mu^{-1}+u^{-1})(\la-\mu u)} \Big(\la I - \mu u\,Q - G\Big). \label{KS:aff}
}
Recall that $Q^2=Q$ and $GQ=QG=0$. In particular, $Q=0$ if $\ell=0$. Using these identities together with \eqref{R(u):check} and \eqref{KS:aff} we find that \eqref{RE1} is equivalent to the equations \eqrefs{RE:C1}{RE:C4} and 
\gan{
\check{R}_q \check{R}_q G_2 - G_2 \check{R}_q \check{R}_q = (q-q^{-1}) (\check{R}_q G_2 - G_2 \check{R}_q), \\
\check{R}_q \check{R}_q Q_2 - Q_2 \check{R}_q \check{R}_q = (q-q^{-1}) (\check{R}_q Q_2 - Q_2 \check{R}_q), \\
\check{R}_q G_2 G_2 - G_2 G_2 \check{R}_q = (\la-\la^{-1}) (\check{R}_q G_2 - G_2 \check{R}_q) - (\check{R}_q Q - Q \check{R}_q) .
}
(We substituted \eqref{R(u):check} and \eqref{KS:aff} to \eqref{RE1}, multiplied by the common denominator, and took coefficients at different powers of $u$ and $v$.) 
The first two equalities hold because of the characteristic identity of $\check{R}_q$. The third equality is true because of \eqref{KKK2}.

We use the same approach to show that the matrix \eqref{KT} is also a solution to \eqref{RE}. We rewrite \eqref{KT} as
\equ{
K(u) = I + \frac{u-u^{-1}}{\al-u}\,Q, \label{KT:aff}
}
where $Q$ denotes the sum in \eqref{KT}. Observe that $Q^2=Q$. The braided reflection equation \eqref{RE1} is then equivalent to the following two equations 
\gan{
\check{R}_q Q_2 \check{R}_q Q_2 = Q_2 \check{R}_q Q_2 \check{R}_q , \\
\check{R}_q \check{R}_q Q_2 - Q_2 \check{R}_q \check{R}_q = (q-q^{-1}) (\check{R}_q Q_2 - Q_2 \check{R}_q) .
}
The first equality is the braided constant reflection equation for the matrix $Q$; it is a direct computation to verify that it is indeed is true, i.e., as in the proof of Lemma \ref{L:KS-const} (this also follows from \cite[Thm.~2]{Mu}). The second equality holds by the same arguments as before.


($\Longleftarrow$) 
Conversely, suppose that $K(u) = \sum_{ij} \key_{ij}(u)\,e_{ij}$ with $\key_{ij}\in\Rat$. Then 
\aln{
& (e^*_i \ot e^*_j)\, \check{R}(\tfrac{u}{v})\, K_2(u)\, \check{R}(uv)\, K_2(v) \,(e_k \ot e_l) 
\el
& \qq = \sum_{1\le r,s\le N} \Big( a_{ij}(\tfrac uv)\,\key_{ir}(u)\, e^*_j \ot e^*_r + b_{ij}(\tfrac uv)\, \key_{jr}(u)\, e^*_i \ot e^*_r \Big) \el[-.75em] & \hspace{2.75cm} \times \Big( a_{ks}(uv)\, \key_{sl}(v)\, e_s \ot e_k + b_{ks}(uv)\, \key_{sl}(v)\, e_k \ot e_s \Big) \\
& \qq = \sum_{1\le r \le N} \Big( \del_{jk}\, a_{ij}(\tfrac uv)\,b_{jr}(uv)\,\key_{ir}(u)\,\key_{rl}(v) + \del_{ik}\, b_{ij}(\tfrac uv)\, b_{ir}(uv)\, \key_{jr}(u)\, \key_{rl}(v) \Big) \el
& \hspace{2.75cm} + a_{ij}(\tfrac uv)\,a_{kj}(uv)\, \key_{ik}(u)\,\key_{jl}(v) + b_{ij}(\tfrac uv)\,a_{ki}(uv)\, \key_{jk}(u)\,  \key_{il}(v) 
\intertext{and}
& (e^*_i \ot e^*_j)\, K_2(v)\, \check{R}(uv)\, K_2(u)\, \check{R}(\tfrac{u}{v}) \,(e_k \ot e_l) 
\el
& \qq= \sum_{1\le r,s\le N} \Big( \key_{jr}(v)\, a_{ir}(uv)\, e^*_r \ot e^*_i + \key_{jr}(v)\, b_{ir}(uv)\, e^*_i \ot e^*_r \Big) \el[-.75em] & \hspace{2.75cm} \times \Big( \key_{sk}(u)\, a_{kl}(\tfrac uv)\, e_l \ot e_s + \key_{sl}(u)\,b_{kl}(\tfrac uv)\, e_k \ot e_s \Big)
\\
& \qq = \sum_{1\le r \le N} \Big( \del_{il}\, b_{ir}(uv)\,a_{ki}(\tfrac uv)\,\key_{jr}(v)\,\key_{rk}(u) +  \del_{ik}\, b_{ir}(uv)\,b_{il}(\tfrac uv)\,\key_{jr}(v)\, \key_{rl}(u) \Big) \el
& \hspace{2.75cm} +  a_{il}(uv)\,a_{kl}(\tfrac uv)\,\key_{jl}(v)\,\key_{ik}(u) + a_{ik}(uv)\,b_{kl}(\tfrac uv)\,\key_{jk}(v)\,\key_{il}(u). 
}
Taking the difference of the expressions above we obtain
\eqa{
\label{K1}
& \sum_{1\le r \le N}  \Big( \del_{ik}\,  b_{ir}(uv)\, \big( b_{ij}(\tfrac uv)\,\key_{jr}(u)\, \key_{rl}(v) - b_{il}(\tfrac uv)\,\key_{jr}(v)\, \key_{rl}(u) \big) \\[-1em]
& \qq\qq + \del_{jk}\, a_{ij}(\tfrac uv)\,b_{jr}(uv)\,\key_{ir}(u)\,\key_{rl}(v) - \del_{il}\, a_{ik}(\tfrac uv)\,b_{ir}(uv)\,\key_{jr}(v)\,\key_{rk}(u)  \Big) 
\\
& \qu + \big( a_{ij}(\tfrac uv)\,a_{kj}(uv) - a_{il}(uv)\,a_{kl}(\tfrac uv)\big)\,\key_{jl}(v)\,\key_{ik}(u) \\ & \qu + a_{ik}(uv)\, \big( b_{ij}(\tfrac uv)\,\key_{jk}(u)\,  \key_{il}(v) - b_{kl}(\tfrac uv)\,\key_{jk}(v)\,\key_{il}(u) \big) = 0. 
}

Next, in five steps, we prove that any solution $K$ of \eqref{RE1} is equivalent to the solution given by \eqref{KS} or \eqref{KT}. 
Recall that $a_{ij}(u)$ and $b_{ij}(u)$ were defined in \eqref{ab}.
For any $i\ne j$ we will write $a_{ij}(u)=a(u)$. We will use the following terminology. 
Let $\si\in\mf{S}_N$ be an involution. Let $i,j$ be such that $i<\si(i)$, $j<\si(j)$ and $i<j$. 
We say that the pairs $(i,\si(i))$ and $(j,\si(j))$ are \emph{non-crossing} if $i<\si(i)<j<\si(j)$ or $i<j<\si(j)<\si(i)$. 
We call the first case \emph{non-crossing split}, we call the second case \emph{non-crossing non-split}. 


{\noindent \it Step 1: Any solution of \eqref{RE1} is a generalized cross matrix, i.e.
\equ{
K(u) = \sum_{1\le i\le N} \Big( \key_{ii}(u) + (1-\del_{i\si(i)})\,\key_{i\si(i)}(u) \Big)\, E_{i\si(i)} \label{K:S1:A}
}
with $\si\in\mf{S}_N$ satisfying $\si^2=\id$. Moreover, without loss of generality, for all $i$ such that $i\ne \si(i)$,}
\equ{
\key_{i\si(i)}(u)=\key_{\si(i)i}(u)\qq\text{or}\qq\key_{i\si(i)}(u)\,\key_{\si(i)i}(u)=0. \label{K:S1:B}
}


Let $i=j\ne k\ne l$ and $i\ne l$ or $i\ne j\ne k= l$ and $i\ne k$. Then \eqref{K1} becomes, respectively, 
\gan{
\big(a(\tfrac{u}{v})-1\big)\, \key_{ik}(u)\, \key_{il}(v) + b_{kl}(\tfrac{u}{v}) \, \key_{ik}(v)\, \key_{il}(u) = 0 , 
\\
\big(a(\tfrac{u}{v})-1\big)\, \key_{il}(v)\, \key_{jl}(u) + b_{ij}(\tfrac{u}{v}) \, \key_{il}(u)\, \key_{jl}(v) = 0 . 
}
Setting $u=-q v$, i.e., $a(\tfrac{u}{v})=1$,  these yield, respectively, $\key_{ik}(v)\,k_{il}(-qv)=0$ and $\key_{il}(v)\,k_{jl}(-qv)=0$ for any value of $v$. It~follows that there can be at most one nonzero off-diagonal entry in each row and each column of $K(u)$. 


Now let $i=j=l\ne k$. This time \eqref{K1} gives
\aln{
& \sum_{1\le r \le N} a(\tfrac{u}{v})\, b_{ir}(u v)\,\key_{ir}(v)\,\key_{rk}(u) + a(\tfrac{u}{v})\,\key_{ii}(v)\,\key_{ik}(u) \\[-.5em] & \qq\qq - a(uv)\,\big( \key_{ii}(v) \,\key_{ik}(u) - b_{ki}(\tfrac{u}{v})\,\key_{ii}(u) \,\key_{ik}(v)\big) = 0.
}
Since each row and each column of $K(u)$ can have at most one nonzero off-diagonal entry, we can assume that $\key_{ir}(u)=0$ for all $r$ except when $r=i$ and $r=s\ne i$. 
Choose $k$ such that $k\ne s$. 
Since $b_{ii}(uv)=0$, the equality above gives $\key_{is}(v)\,\key_{sk}(u) = 0$. 
It follows that, if $\key_{is}(u)$ with $i\ne s$ is nonzero, then $\key_{sk}(u)=0$ for any $k$ satisfying $k\ne s$ and $k\ne i$. 
In particular, $K(u)$ is a generalized cross matrix.


Let $i=k\ne j=l$ and $i=\si(j)$. Then \eqref{K1} becomes
\[
b_{i\si(i)}(u v) \,\big(\key_{i\si(i)}(u)\, \key_{\si(i)i}(v)-\key_{\si(i)i}(u)\, \key_{i\si(i)}(v)\big) = 0.
\]
Hence $\key_{i\si(i)}(u)\,\key_{\si(i)i}(u)=0$ or $\key_{i\si(i)}(u) = c_i\,\key_{\si(i)i}(u)$ for some $c_i \in \C^\times$. In the latter case, by Lemma~\ref{L:K-symm} (iii), we may assume that $c_i=1$, i.e., $\key_{i\si(i)}(u)=\key_{\si(i)i}(u)$. This implies \eqref{K:S1:B}.


{\noindent \it Step 2: For any $i$ such that $i<\si(i)$ both $\key_{ii}$ and $\key_{\si(i)\si(i)}$ are either zero or lie in $\Rat^\times$. 
In the first case $K(u)$ is a generalized permutation matrix such that, for any $j$ satisfying $i<\si(i)<j<\si(j)$ or $i<j<\si(j)<\si(i)$ and any $k,l$ satisfying $k=\si(k)$ and $l=\si(l)$, we have, for $\al_{kl}\in\C$,
\gat{
\frac{\key_{kk}(u)}{\key_{ll}(u)} = \frac{\al_{kl}-u}{\al_{kl}-u^{-1}} ,  
\qq
\frac{\key_{i\si(i)}^2(u)}{\key_{j\si(j)}^2(u)} = \del_{i<j<\si(j)<\si(i)} + u^2\, \del_{i<\si(i)<j<\si(j)}, \label{K:S2:A1}
\\
\frac{\key_{kk}^2(u)}{\key_{i\si(i)}^2(u)} = \del_{k<i}\,u + \del_{i<k<\si(i)} + \del_{\si(i)<k}\,u^{-1}. \label{K:S2:A2}
}
Otherwise, i.e., if $\key_{ii}, \key_{\si(i)\si(i)} \in \Rat^\times$, then
\equ{
\frac{\key_{i\si(i)}(u)}{\key_{\si(i)\si(i)}(u)} = \frac{u-u^{-1}}{\al_{i\si(i)}-u^{-1}}\,\beta_{i\si(i)} , \qq \frac{\key_{\si(i)i}(u)}{\key_{\si(i)\si(i)}(u)} = \frac{u-u^{-1}}{\al_{i\si(i)}-u^{-1}}\,\beta_{\si(i)i}  \label{K:S2:B1}
}
with $\beta_{i\si(i)},\beta_{\si(i)i} \in \C$ such that $\beta_{i\si(i)}=\beta_{\si(i)i}\ne 0$ or $\beta_{i\si(i)}\,\beta_{\si(i)i}=0$. 
Moreover, for any $i,j$ such that $i<j$ and $\key_{i\si(i)}(u)\,\key_{\si(i)i}(u)=0$ if $i\ne \si(i)$ and $\key_{j\si(j)}(u)\,\key_{\si(j)j}(u)=0$ if $j\ne \si(j)$ we have that
\equ{
\frac{\key_{ii}(u)}{\key_{jj}(u)} = \frac{\al_{ij}-u}{\al_{ij}-u^{-1}}  \label{K:S2:B2}
}
for some $\al_{ij}\in\C$. Furthermore, for any $i,k$ such that $k\ne i<\si(i)\ne k$,
\equ{
\key_{kk}(u) = \bigg( 1 - (u-u^{-1})\,\frac{\del_{k<i} - (\del_{k<i} u + \del_{i<k<\si(i)} + \del_{\si(i)<k} u^{-1}) \ga_{i\si(i)k}}{\al_{i\si(i)}-u^{-1}}\bigg)\, \key_{\si(i)\si(i)}(u) . \label{K:S2:C}
}
for some $\ga_{i\si(i)k}\in\C$, and so $\key_{kk}$ is nonzero for all $k$. \smallskip 
}


Let $i=l\ne j=k$. Then \eqref{K1} gives
\eqa{
& \sum_{1\le r \le N} a(\tfrac{u}{v}) \Big( b_{jr}(u v)\,\key_{ir}(u)\,\key_{ri}(v) - b_{ir}(u v)\,\key_{jr}(u)\,\key_{rj}(v) \Big) \\[-.5em]
& \qq\qq + a(u v) \Big(b_{ij}(\tfrac uv)\,\key_{ii}(v)\,\key_{jj}(u) - b_{ji}(\tfrac uv)\,\key_{ii}(u)\,\key_{jj}(v) \Big) = 0 . \label{K:22}
}
Assume that $i<j=\si(i)$. Then  
\gan{
b_{ji}(u) = \frac{b_{ij}(u)}{u}, \qq 
 a(\tfrac uv)\, b_{ij}(uv) = \frac{(q^3-q)\,uv\,(v-u)}{(u-q^2 v)(uv-q^2)} , \qq
a(uv)\, b_{ij}(\tfrac uv) = \frac{(q^3-q)\,u\,(1-u v)}{(u-q^2 v)(uv-q^2)} 
}
and the equality above becomes 
\eqa{
& \frac{q^3-q}{(u-q^2 v)(uv-q^2)} \bigg(\key_{ii}(u) \Big((u-v)\,\key_{ii}(v)+v\,(1-u v)\,\key_{jj}(v) \Big) 
\\ & \hspace{3.8cm}- u\,\key_{jj}(u) \Big(v\,(u-v)\, \key_{jj}(v) + (1-u v)\,\key_{ii}(v) \Big)\bigg)  = 0 . \label{K:22-diag}
}
(In the computations below we will often omit writing the common factors which have no zeros or poles for generic values of $u$ and $v$.)
Hence $\key_{ii}(u)=\key_{jj}(u)=0$ or both are nonzero. Note that the first case can is only allowed if both $\key_{i\si(i)}(u)$ and $\key_{\si(i)i}(u)$ are nonzero, by the invertibility of $K(u)$.
In the second case assume that 
\equ{
\frac{\key_{ii}(u)}{\key_{jj}(u)} = \frac{\al_{ij}(u)-u}{\al_{ij}(u)-u^{-1}} \label{K:ii-jj-u}
}
with $\al_{ij}(u)\in\Rat$. This allows us to separate variables yielding $\al_{ij}(u)=\al_{ij}(v)$. In particular, $\al_{ij}(u) = \al_{ij} \in\C$. The case $i>j=\si(i)$ yields an analogous result.  

Let us return to \eqref{K:22} and assume that $i,j$ are such that $i<j$ and $\key_{i\si(i)}(u)\,\key_{\si(i)i}(u)=0$ if $i\ne \si(i)$ and the same condition for $j$. 
Then \eqref{K:22} simplifies to \eqref{K:22-diag} and $\key_{ii}(u)$, $\key_{jj}(u)$ must be nonzero, by invertibility of $K(u)$.
Repeating the same arguments as before we find the same result, i.e., \eqref{K:ii-jj-u} with $\al_{ij}(u)=\al_{ij}\in\C$. 
This completes the proof of the first relation in \eqref{K:S2:A1} and the relation \eqref{K:S2:B2}.

On the other hand, when $i\ne\si(i)$, $j\ne\si(j)$ and $\key_{ii}(u)=\key_{jj}(u)=0$ but $\key_{i\si(i)}(u)\,\key_{\si(i)i}(u)$ and $\key_{j\si(j)}(u)\,\key_{\si(j)j}(u)$ are nonzero, the equality \eqref{K:22} becomes
\[
a(\tfrac uv) \Big( b_{j\si(i)}(u v)\,\key_{i\si(i)}(u)\,\key_{\si(i)i}(v) - b_{i\si(j)}(u v)\,\key_{j\si(j)}(u)\,\key_{\si(j)j}(v) \Big) = 0 
\]
giving
\[
(\del_{j<\si(i)}\,uv+\del_{\si(i)<j})\,\key_{i\si(i)}(u)\,\key_{\si(i)i}(v) - (\del_{i<\si(j)}\,uv+\del_{\si(j)<i})\,\key_{j\si(j)}(u)\,\key_{\si(j)j}(v) = 0 .
\]
We focus on the cases when $i<\si(i)<j<\si(j)$ and $i<j<\si(j)<\si(i)$. Using \eqref{K:S1:B} and separating variables we obtain the second relation in \eqref{K:S2:A1}.

Finally, assume that $i=\si(i)$, $j<\si(j)$ and $\key_{jj}(u)=0$. The equality \eqref{K:22} becomes
\[
a(\tfrac uv) \Big( b_{ji}(u v)\,\key_{ii}(u)\,\key_{ii}(v) - b_{i\si(j)}(u v)\,\key_{j\si(j)}(u)\,\key_{\si(j)j}(v) \Big) = 0 
\]
giving
\[
(\del_{j<i}\,uv + \del_{i<j})\,\key_{ii}(u)\,\key_{ii}(v) - (\del_{i<\si(j)}\,uv + \del_{\si(j)<i})\, \key_{j\si(j)}(u)\,\key_{\si(j)j}(v) = 0 .
\]
By separating the variables we obtain \eqref{K:S2:A2}.


Let $i=j=l\ne k$ and $i=\si(k)$. Upon renaming $k\to j$ in \eqref{K1} we find
\[
a(u v) \Big(\key_{ii}(v)\, \key_{ij}(u) - b_{ji}(\tfrac{u}{v})\, \key_{ii}(u)\, \key_{ij}(v) \Big) = a(\tfrac{u}{v}) \Big( \key_{ii}(v)\,\key_{ij}(u) + b_{ij}(u v)\,\key_{ij}(v)\,\key_{jj}(u)\Big) .
\]
The trivial solutions are $\key_{ii}(u)=\key_{jj}(u)=0$ and $\key_{ij}(u)=0$. In the non-trivial case, for $i<j=\si(i)$, the equation above gives
\aln{
(u-u^{-1})(\al_{i\si(i)}-v^{-1})\,\key_{i\si(i)}(v)\,\key_{\si(i)\si(i)}(u)   - (v-v^{-1}) (\al_{ij}-u^{-1})\,\key_{i\si(i)}(u)\,\key_{\si(i)\si(i)}(v) = 0.
}
Here we used \eqref{K:S2:B2}, which we have already proved.
Upon separating the variables we obtain the first identity in \eqref{K:S2:B1}. If~$i>j$ instead, we need to interchange $i$ with $j$ and proceed in the same way as before. 
This yields the second identity in \eqref{K:S2:B1}.


Let $i=l\ne j \ne k$, $i\ne k$ and $\si(j)=k$. Then \eqref{K1} gives 
\aln{
& a(\tfrac{u}{v})\Big(b_{ij}(u v)\,\key_{jj}(v)\,\key_{jk}(u) + b_{ik}(u v)\,\key_{jk}(v)\,\key_{kk}(u) \Big) \\
& \qu - a(u v) \Big(b_{ij}(\tfrac{u}{v})\,\key_{ii}(v)\,\key_{jk}(u) - b_{ki}(\tfrac{u}{v})\,\key_{ii}(u)\,\key_{jk}(v) \Big) = 0 .
}
Assume further that $j<k$ and $\key_{jj}(u)$, $\key_{kk}(u)$ and  $\key_{jk}(u)$ are all nonzero. Using \eqref{K:S2:B1} and \eqref{K:S2:B2}, and dividing by the common factor, we rewrite the equality above as
\aln{
& \frac{v-v^{-1}}{\al_{jk}-v^{-1}} \Big((u v-1)\,(\del_{k<i}\,u + \del_{i<k}\,v)\,\key_{ii}(u) + (u-v)\,(\del_{i<k}\,u v + \del_{k<i})\,\key_{kk}(u)  \Big)\key_{kk}(v) \\ 
& \qq- \frac{u-u^{-1}}{\al_{jk}-u^{-1}} \Big((u v-1)\,(\del_{i<j}\,u + \del_{j<i}\,v)\,\key_{ii}(v) - (u-v)\,(\del_{i<j}\,u v + \del_{j<i})\, \key_{jj}(v) \Big) \key_{kk}(u) = 0.
}
Upon separating the variables we find, for some $\ga_{ijk}\in\C$, 
\[
\key_{ii}(u) = \bigg( 1 - (u-u^{-1})\,\frac{\del_{i<j} - (\del_{i<j} u + \del_{j<i<k} + \del_{k<i} u^{-1}) \ga_{ijk}}{\al_{jk}-u^{-1}}\bigg)\,\key_{kk}(u) . 
\]
The case when $k<j$ yields an analogous result. 
It follows that, if $\key_{jj}(u)$, $\key_{\si(j),\si(j)}(u)$ and $\key_{j,\si(j)}(u)$ or $\key_{\si(j),j}(u)$ with $j\ne\si(j)$ are nonzero for some $j$, then $\key_{ii}(u)$ are nonzero for all $i$. 
This completes the proof of the second statement of Step 2.


{\noindent \it Step 3: If $K(u)$ is a generalized involution matrix, then $K$ is equivalent to the solution given by
\equ{
I + \frac{u-u^{-1}}{\al-u} \sum_{m< i\le N} E_{ii} \label{KD0}
}
for some $m$ satisfying $\tfrac12 N \le m \le N$ and $\al\in\C$ satisfying $0\le \Arg(\al) < \pi$, or to the solution given by
\equ{
\sum_{1\le i \le \ell} u E_{ii} + \sum_{\ell<i\le r} \big(E_{i\si(i)} + E_{\si(i)i} \big) + \sum_{r<i\le N+\ell-r} E_{ii} \label{KA3:X0} 
}
for some $(\ell,r,\si)\in\Sigma^{\rm sym}_N$. 
}

Suppose that $K(u)$ is diagonal, i.e., $K(u)=\sum_{1\le i\le N} \key_{ii}(u)\,E_{ii}$. Then the first relation in \eqref{K:S2:A1} implies that there can be at most three different entries, i.e., for any $\ell,r$ satisfying $0\le \ell<r\le N$ we have, for all $i,i',j,j',k,k'$ such that $i,i'\le \ell< j,j'\le r < k,k'$, 
\eqg{
\frac{\key_{ii}(u)}{\key_{jj}(u)}= \frac{1-\al\,u}{1-\al\,u^{-1}} , \qq  \frac{\key_{jj}(u)}{\key_{kk}(u)}= \frac{\al-u}{\al-u^{-1}}, \qq \frac{\key_{ii}(u)}{\key_{kk}(u)}= u^2 , \\
\frac{\key_{ii}(u)}{\key_{i'i'}(u)}=\frac{\key_{jj}(u)}{\key_{j'j'}(u)}=\frac{\key_{kk}(u)}{\key_{k'k'}(u)}=1, \label{K:diags2} 
}
with $\al\in\C$. This yields, up to an overall scalar factor, 
\equ{
K(u) = \sum_{1\le i \le \ell} u\,E_{ii} + \sum_{\ell< i\le r} \frac{\al-u}{\al\,u-1}\, E_{ii} + \sum_{r< i\le N} u^{-1} E_{ii}, \label{KD}
}
which, by Lemma \ref{L:[K]} (\eqref{KD} is a special case of \eqref{KT1} when $\si=\id$), is equivalent to \eqref{KD0}: denoting the matrix \eqref{KD0} by $K_m(u)$ and \eqref{KD} by $K_{\ell,r}(u)$ we have, for $0\le k <N$,
\[
(Z^\rho(\eta^{-1} u))^{-k} K_m(u)\,(Z^\rho(\eta u))^k = \begin{cases}
\frac{\al\,u-1}{\al-u}\, K_{k,k+m}(u) &\text{if } k+m\le N,\\ 
u\, K_{k+m-N,k}(u)\big|_{\al\to\al^{-1}} &\text{if } k+m> N. 
\end{cases}
\]

Next, suppose that $K(u)$ is non-diagonal, i.e., 
\equ{
K(u)=\sum_{1\le i\le N} \key_{i\si(i)}(u)\,E_{i\si(i)} \label{K:GPM-ansatz}
}
with $\si \in \mf{S}_N$ of order 2. 
We will assume that $\key_{i\si(i)}(u)=\key_{\si(i)i}(u)$ for all $i$ satisfying $i<\si(i)$; this follows from Step 1.
Let $i,j,k,l$ be all different or $i,j,k$ different and $l=j$.
Then \eqref{K1} gives
\[
a(uv) \Big( b_{ij}(\tfrac uv)\,\key_{il}(v)\,\key_{jk}(u) - b_{kl}(\tfrac uv)\,\key_{il}(u)\,\key_{jk}(v) \Big) = 0 .
\]
Using \eqref{ab} we obtain
\[
(\del_{i>j}\,v + \del_{i<j}\,u)\,\key_{il}(v)\,\key_{jk}(u) - ( \del_{k>l}\,v + \del_{k<l}\,u )\,\key_{il}(u)\,\key_{jk}(v) = 0 ,
\]
or equivalently ({\it cf.}~\eqref{Z2})
\equ{
\frac{\key_{i\pm k,j\pm l}(u)}{\key_{ij}(u)} \in \C u^{\mp1} , \qq 
\frac{\key_{i\pm k,j\mp l}(u)}{\key_{ij}(u)} \in \C . \label{K:ij-kl}
}

Let $i,j$ be such that $i<\si(i)$, $j<\si(j)$ and $i<j$. The relations \eqref{K:ij-kl} applied to the coefficients $\key_{i,\si(i)}(u)$, $\key_{\si(i),i}(u)$ and $\key_{j,\si(j)}(u)$, $\key_{\si(j),j}(u)$ imply that we must have $\si(i)<j$ or $\si(j)<\si(i)$, i.e., the pairs $(i,\si(i))$ and $(j,\si(j))$ are non-crossing. We represent these configurations diagrammatically~by
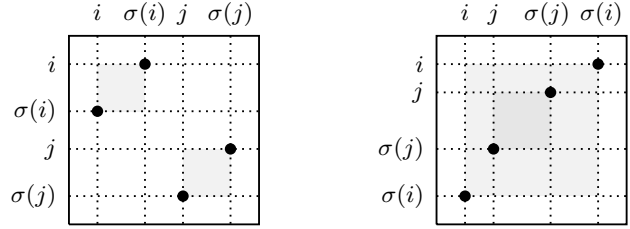
\[
\btp
\fill[gray!10] (1.5,6) rectangle (4,8.5);
\fill[gray!10] (6,4) rectangle (8.5,1.5);
\draw[dotted] (0,8.5) node[left=1pt]{\color{black}\small$i$} -- (10,8.5) ;
\draw[dotted] (1.5,10) node[above=6pt, anchor=base]{\color{black}\small$i$} -- (1.5,0);
\draw[dotted] (0,6) node[left=1pt]{\color{black}\small$\si(i)$} -- (10,6) ;
\draw[dotted] (4,10) node[above=6pt, anchor=base]{\color{black}\small$\si(i)$} -- (4,0);
\draw[dotted] (0,4) node[left=1pt]{\color{black}\small$j$} -- (10,4);
\draw[dotted] (6,10) node[above=6pt, anchor=base]{\color{black}\small$j$} -- (6,0);
\draw[dotted] (0,1.5) node[left=1pt]{\color{black}\small$\si(j)$} -- (10,1.5);
\draw[dotted] (8.5,10) node[above=6pt, anchor=base]{\color{black}\small$\si(j)$} -- (8.5,0);
\draw (0,0) -- (10,0) -- (10,10) -- (0,10) -- (0,0);
\filldraw[fill=black] (4,8.5) circle (.25);
\filldraw[fill=black] (1.5,6) circle (.25);
\filldraw[fill=black] (8.5,4) circle (.25);
\filldraw[fill=black] (6,1.5) circle (.25);
\etp
\qq\qq
\btp
\fill[gray!10] (1.5,8.5) rectangle (8.5,1.5);
\fill[gray!20] (3,4) rectangle (6,7);
\draw[dotted] (0,8.5) node[left=1pt]{\color{black}\small$i$} -- (10,8.5);
\draw[dotted] (1.5,10) node[above=6pt, anchor=base]{\color{black}\small$i$} -- (1.5,0);
\draw[dotted] (0,7) node[left=1pt]{\color{black}\small$j$} -- (10,7) ;
\draw[dotted] (3,10) node[above=6pt, anchor=base]{\color{black}\small$j$} -- (3,0);
\draw[dotted] (0,4) node[left=1pt]{\color{black}\small$\si(j)$} -- (10,4);
\draw[dotted] (6,10) node[above=6pt, anchor=base]{\color{black}\small$\si(j)\;$} -- (6,0);
\draw[dotted] (0,1.5) node[left=1pt]{\color{black}\small$\si(i)$} -- (10,1.5);
\draw[dotted] (8.5,10) node[above=6pt, anchor=base]{\color{black}\small$\;\si(i)$} -- (8.5,0);
\draw (0,0) -- (10,0) -- (10,10) -- (0,10) -- (0,0);
\filldraw[fill=black] (1.5,1.5) circle (.25);
\filldraw[fill=black] (8.5,8.5) circle (.25);
\filldraw[fill=black] (3,4) circle (.25);
\filldraw[fill=black] (6,7) circle (.25);
\etp
\]
where the filled nodes indicate the nonzero entries of $K(u)$ in question. Moreover, since $K(u)$ is a generalized permutation matrix, the filled nodes correspond to the only nonzero entries in each dotted line and each dotted column. The grey shadings are added to emphasize the non-crossing property of the pairs $(i,\si(i))$ and $(j,\si(j))$: they are split in the first case and non-split in the second case.
In particular, when $N=3$ we have $\si\in\{(12),(13),(23)\}$ and when $N=4$ we have $\si\in\{(12),(13),(14),(23),(24),(12)(34),(14)(23)\}$. 
When $N\ge5$ we need additional arguments.

Let us focus on case when the non-crossing pairs $(i,\si(i))$ and $(j,\si(j))$ are split. Let $k$ be such that $k<\si(k)$. By the same arguments as before, the relations \eqref{K:ij-kl} imply that the pair $(k,\si(k))$ must satisfy one the following four conditions:
\aln{
& k<i<\si(i)<\si(k)<j<\si(j) ,\qq 
&& i<k<\si(k)<\si(i)<j<\si(j) ,\\ 
& i<\si(i)<k<j<\si(j)<\si(k) ,\qq 
&& i<\si(i)<j<k<\si(k)<\si(j) .
}
In other words, all the pairs must be simultaneously non-crossing and there can be at most two simultaneously split pairs. We represent these configurations by 
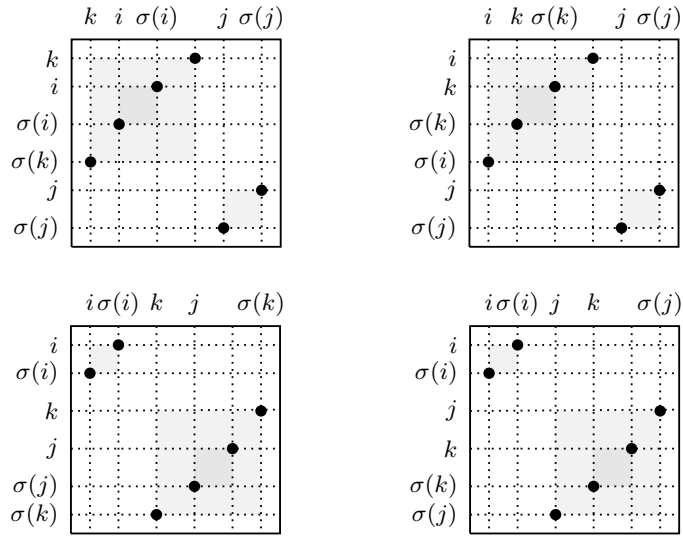
\[
\btp
\fill[gray!10] (1,4.5) rectangle (6.5,10);
\fill[gray!20] (2.5,8.5) rectangle (4.5,6.5);
\fill[gray!10] (8,1) rectangle (10,3);
\draw[dotted] (0,10) node[left=1pt]{\color{black}\small$k$} -- (11,10);
\draw[dotted] (1,11) node[above=6pt, anchor=base]{\color{black}\small$k$} -- (1,0);
\draw[dotted] (0,8.5) node[left=1pt]{\color{black}\small$i$} -- (11,8.5);
\draw[dotted] (2.5,11) node[above=6pt, anchor=base]{\color{black}\small$i$} -- (2.5,0);
\draw[dotted] (0,6.5) node[left=1pt]{\color{black}\small$\si(i)$} -- (11,6.5);
\draw[dotted] (4.5,11) node[above=6pt, anchor=base]{\color{black}\small$\si(i)$} -- (4.5,0);
\draw[dotted] (0,4.5) node[left=1pt]{\color{black}\small$\si(k)$} -- (11,4.5);
\draw[dotted] (6.5,11) -- (6.5,0);
\draw[dotted] (0,3) node[left=1pt]{\color{black}\small$j$} -- (11,3);
\draw[dotted] (8,11) node[above=6pt, anchor=base]{\color{black}\small$j$} -- (8,0);
\draw[dotted] (0,1) node[left=1pt]{\color{black}\small$\si(j)$} -- (11,1);
\draw[dotted] (10,11) node[above=6pt, anchor=base]{\color{black}\small$\si(j)$} -- (10,0);
\draw (0,0) -- (11,0) -- (11,11) -- (0,11) -- (0,0);
\filldraw[fill=black] (1,4.5) circle (.25);
\filldraw[fill=black] (6.5,10) circle (.25);
\filldraw[fill=black] (2.5,6.5) circle (.25);
\filldraw[fill=black] (4.5,8.5) circle (.25);
\filldraw[fill=black] (8,1) circle (.25);
\filldraw[fill=black] (10,3) circle (.25);
\etp
\qq\qq
\btp
\fill[gray!10] (1,4.5) rectangle (6.5,10);
\fill[gray!20] (2.5,8.5) rectangle (4.5,6.5);
\fill[gray!10] (8,1) rectangle (10,3);
\draw[dotted] (0,10) node[left=1pt]{\color{black}\small$i$} -- (11,10);
\draw[dotted] (1,11) node[above=6pt, anchor=base]{\color{black}\small$i$} -- (1,0);
\draw[dotted] (0,8.5) node[left=1pt]{\color{black}\small$k$} -- (11,8.5);
\draw[dotted] (2.5,11) node[above=6pt, anchor=base]{\color{black}\small$k$} -- (2.5,0);
\draw[dotted] (0,6.5) node[left=1pt]{\color{black}\small$\si(k)$} -- (11,6.5);
\draw[dotted] (4.5,11) node[above=6pt, anchor=base]{\color{black}\small$\si(k)$} -- (4.5,0);
\draw[dotted] (0,4.5) node[left=1pt]{\color{black}\small$\si(i)$} -- (11,4.5);
\draw[dotted] (6.5,11) -- (6.5,0);
\draw[dotted] (0,3) node[left=1pt]{\color{black}\small$j$} -- (11,3);
\draw[dotted] (8,11) node[above=6pt, anchor=base]{\color{black}\small$j$} -- (8,0);
\draw[dotted] (0,1) node[left=1pt]{\color{black}\small$\si(j)$} -- (11,1);
\draw[dotted] (10,11) node[above=6pt, anchor=base]{\color{black}\small$\si(j)$} -- (10,0);
\draw (0,0) -- (11,0) -- (11,11) -- (0,11) -- (0,0);
\filldraw[fill=black] (1,4.5) circle (.25);
\filldraw[fill=black] (6.5,10) circle (.25);
\filldraw[fill=black] (2.5,6.5) circle (.25);
\filldraw[fill=black] (4.5,8.5) circle (.25);
\filldraw[fill=black] (8,1) circle (.25);
\filldraw[fill=black] (10,3) circle (.25);
\etp
\]
\[
\btp
\fill[gray!10] (1,8.5) rectangle (2.5,10);
\fill[gray!10] (4.5,6.5) rectangle (10,1);
\fill[gray!20] (6.5,2.5) rectangle (8.5,4.5);
\draw[dotted] (0,10) node[left=1pt]{\color{black}\small$i$} -- (11,10);
\draw[dotted] (1,11) node[above=6pt, anchor=base]{\color{black}\small$i$} -- (1,0);
\draw[dotted] (0,8.5) node[left=1pt]{\color{black}\small$\si(i)$} -- (11,8.5);
\draw[dotted] (2.5,11) node[above=6pt, anchor=base]{\color{black}\small$\si(i)$} -- (2.5,0);
\draw[dotted] (0,6.5) node[left=1pt]{\color{black}\small$k$} -- (11,6.57);
\draw[dotted] (4.5,11) node[above=6pt, anchor=base]{\color{black}\small$k$} -- (4.5,0);
\draw[dotted] (0,4.5) node[left=1pt]{\color{black}\small$j$} -- (11,4.5);
\draw[dotted] (6.5,11) node[above=6pt, anchor=base]{\color{black}\small$j$} -- (6.5,0);
\draw[dotted] (0,2.5) node[left=1pt]{\color{black}\small$\si(j)$} -- (11,2.5);
\draw[dotted] (8.5,11)  -- (8.5,0);
\draw[dotted] (0,1) node[left=1pt]{\color{black}\small$\si(k)$} -- (11,1);
\draw[dotted] (10,11) node[above=6pt, anchor=base]{\color{black}\small$\si(k)$} -- (10,0);
\draw (0,0) -- (11,0) -- (11,11) -- (0,11) -- (0,0);
\filldraw[fill=black] (1,8.5) circle (.25);
\filldraw[fill=black] (2.5,10) circle (.25);
\filldraw[fill=black] (4.5,1) circle (.25);
\filldraw[fill=black] (10,6.5) circle (.25);
\filldraw[fill=black] (6.5,2.5) circle (.25);
\filldraw[fill=black] (8.5,4.5) circle (.25);
\etp
\qq\qq
\btp
\fill[gray!10] (1,8.5) rectangle (2.5,10);
\fill[gray!10] (4.5,6.5) rectangle (10,1);
\fill[gray!20] (6.5,2.5) rectangle (8.5,4.5);
\draw[dotted] (0,10) node[left=1pt]{\color{black}\small$i$} -- (11,10);
\draw[dotted] (1,11) node[above=6pt, anchor=base]{\color{black}\small$i$} -- (1,0);
\draw[dotted] (0,8.5) node[left=1pt]{\color{black}\small$\si(i)$} -- (11,8.5);
\draw[dotted] (2.5,11) node[above=6pt, anchor=base]{\color{black}\small$\si(i)$} -- (2.5,0);
\draw[dotted] (0,6.5) node[left=1pt]{\color{black}\small$j$} -- (11,6.57);
\draw[dotted] (4.5,11) node[above=6pt, anchor=base]{\color{black}\small$j$} -- (4.5,0);
\draw[dotted] (0,4.5) node[left=1pt]{\color{black}\small$k$} -- (11,4.5);
\draw[dotted] (6.5,11) node[above=6pt, anchor=base]{\color{black}\small$k$} -- (6.5,0);
\draw[dotted] (0,2.5) node[left=1pt]{\color{black}\small$\si(k)$} -- (11,2.5);
\draw[dotted] (8.5,11)  -- (8.5,0);
\draw[dotted] (0,1) node[left=1pt]{\color{black}\small$\si(j)$} -- (11,1);
\draw[dotted] (10,11) node[above=6pt, anchor=base]{\color{black}\small$\si(j)$} -- (10,0);
\draw (0,0) -- (11,0) -- (11,11) -- (0,11) -- (0,0);
\filldraw[fill=black] (1,8.5) circle (.25);
\filldraw[fill=black] (2.5,10) circle (.25);
\filldraw[fill=black] (4.5,1) circle (.25);
\filldraw[fill=black] (10,6.5) circle (.25);
\filldraw[fill=black] (6.5,2.5) circle (.25);
\filldraw[fill=black] (8.5,4.5) circle (.25);
\etp
\]
respectively. (Note that the first two diagrams describe equivalent configurations; the same is true for the last two diagrams.)
Now let $k$ be such that $k=\si(k)$ instead. The relations \eqref{K:ij-kl} together with \eqref{K:S2:A2} imply that $k$ must satisfy one of the following two conditions:
\equ{
i<k<\si(i)<j<\si(j) , \qq i<\si(i)<j<k<\si(j). \label{K:XX-iki}
}
We represent these configurations by
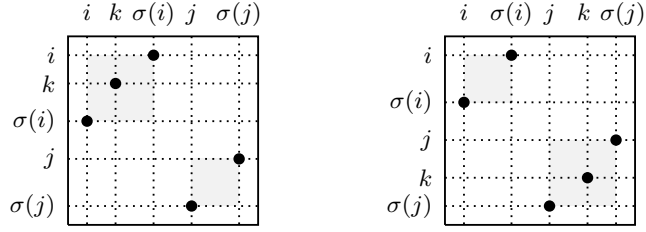
\[
\btp
\fill[gray!10] (1,5.5) rectangle (4.5,9);
\fill[gray!10] (6.5,3.5) rectangle (9,1);
\draw[dotted] (0,9) node[left=1pt]{\color{black}\small$i$} -- (10,9);
\draw[dotted] (1,10) node[above=6pt, anchor=base]{\color{black}\small$i$} -- (1,0);
\draw[dotted] (0,7.5) node[left=1pt]{\color{black}\small$k$} -- (10,7.5);
\draw[dotted] (2.5,10) node[above=6pt, anchor=base]{\color{black}\small$k$} -- (2.5,0);
\draw[dotted] (0,5.5) node[left=1pt]{\color{black}\small$\si(i)$} -- (10,5.5) ;
\draw[dotted] (4.5,10) node[above=6pt, anchor=base]{\color{black}\small$\si(i)$} -- (4.5,0);
\draw[dotted] (0,3.5) node[left=1pt]{\color{black}\small$j$} -- (10,3.5);
\draw[dotted] (6.5,10) node[above=6pt, anchor=base]{\color{black}\small$j$} -- (6.5,0);
\draw[dotted] (0,1) node[left=1pt]{\color{black}\small$\si(j)$} -- (10,1);
\draw[dotted] (9,10) node[above=6pt, anchor=base]{\color{black}\small$\si(j)$} -- (9,0);
\draw (0,0) -- (10,0) -- (10,10) -- (0,10) -- (0,0);
\filldraw[fill=black] (4.5,9) circle (.25);
\filldraw[fill=black] (1,5.5) circle (.25);
\filldraw[fill=black] (9,3.5) circle (.25);
\filldraw[fill=black] (6.5,1) circle (.25);
\filldraw[fill=black] (2.5,7.5) circle (.25);
\etp
\qq\qq
\btp
\fill[gray!10] (1,6.5) rectangle (3.5,9);
\fill[gray!10] (5.5,4.5) rectangle (9,1);
\draw[dotted] (0,9) node[left=1pt]{\color{black}\small$i$} -- (10,9);
\draw[dotted] (1,10) node[above=6pt, anchor=base]{\color{black}\small$i$} -- (1,0);
\draw[dotted] (0,6.5) node[left=1pt]{\color{black}\small$\si(i)$} -- (10,6.5) ;
\draw[dotted] (3.5,10) node[above=6pt, anchor=base]{\color{black}\small$\si(i)$} -- (3.5,0);
\draw[dotted] (0,4.5) node[left=1pt]{\color{black}\small$j$} -- (10,4.5);
\draw[dotted] (5.5,10) node[above=6pt, anchor=base]{\color{black}\small$j$} -- (5.5,0);
\draw[dotted] (0,2.5) node[left=1pt]{\color{black}\small$k$} -- (10,2.5);
\draw[dotted] (7.5,10) node[above=6pt, anchor=base]{\color{black}\small$k\;$} -- (7.5,0);
\draw[dotted] (0,1) node[left=1pt]{\color{black}\small$\si(j)$} -- (10,1);
\draw[dotted] (9,10) node[above=6pt, anchor=base]{\color{black}\small$\;\;\si(j)$} -- (9,0);
\draw (0,0) -- (10,0) -- (10,10) -- (0,10) -- (0,0);
\filldraw[fill=black] (3.5,9) circle (.25);
\filldraw[fill=black] (1,6.5) circle (.25);
\filldraw[fill=black] (9,4.5) circle (.25);
\filldraw[fill=black] (5.5,1) circle (.25);
\filldraw[fill=black] (7.5,2.5) circle (.25);
\etp
\]

Let us now focus on the case when the non-crossing pairs $(i,\si(i))$ and $(j,\si(j))$ are non-split, i.e., when $i<j<\si(j)<\si(i)$. Let $k$ be such that $k=\si(k)$.  (The case when $k\ne\si(k)$ follows from the considerations above.) Then \eqref{K:ij-kl} and \eqref{K:S2:A2} imply that $k$ must satisfy one the following three conditions:
\equ{
k<i<j<\si(j)<\si(i) , \qq i<j<k<\si(j)<\si(i)) , \qq i<j<\si(j)<\si(i)<k . \label{K:X-iki}
}
We represent these configurations by
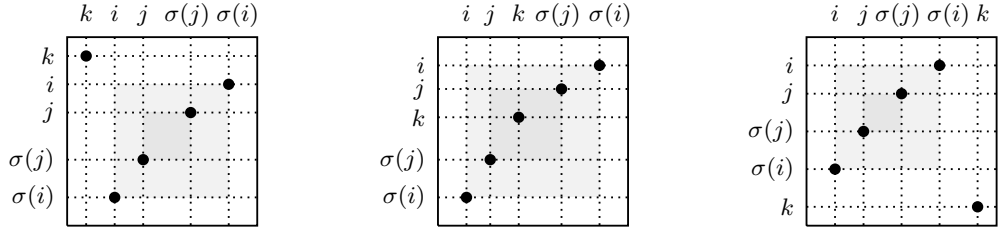
\[
\btp
\fill[gray!10] (2.5,7.5) rectangle (8.5,1.5);
\fill[gray!20] (4,3.5) rectangle (6.5,6);
\draw[dotted] (0,9) node[left=1pt]{\color{black}\small$k$} -- (10,9);
\draw[dotted] (1,10) node[above=6pt, anchor=base]{\color{black}\small$k$} -- (1,0);
\draw[dotted] (0,7.5) node[left=1pt]{\color{black}\small$i$} -- (10,7.5);
\draw[dotted] (2.5,10) node[above=6pt, anchor=base]{\color{black}\small$i$} -- (2.5,0);
\draw[dotted] (0,6) node[left=1pt]{\color{black}\small$j$} -- (10,6) ;
\draw[dotted] (4,10) node[above=6pt, anchor=base]{\color{black}\small$j$} -- (4,0);
\draw[dotted] (0,3.5) node[left=1pt]{\color{black}\small$\si(j)$} -- (10,3.5);
\draw[dotted] (6.5,10) node[above=6pt, anchor=base]{\color{black}\small$\si(j)\;$} -- (6.5,0);
\draw[dotted] (0,1.5) node[left=1pt]{\color{black}\small$\si(i)$} -- (10,1.5);
\draw[dotted] (8.5,10) node[above=6pt, anchor=base]{\color{black}\small$\;\;\si(i)\!$} -- (8.5,0);
\draw (0,0) -- (10,0) -- (10,10) -- (0,10) -- (0,0);
\filldraw[fill=black] (1,9) circle (.25);
\filldraw[fill=black] (2.5,1.5) circle (.25);
\filldraw[fill=black] (8.5,7.5) circle (.25);
\filldraw[fill=black] (4,3.5) circle (.25);
\filldraw[fill=black] (6.5,6) circle (.25);
\etp
\qq\qq
\btp
\fill[gray!10] (1.5,1.5) rectangle (8.5,8.5);
\fill[gray!20] (2.75,3.5) rectangle (6.5,7.25);
\draw[dotted] (0,8.5) node[left=1pt]{\color{black}\small$i$} -- (10,8.5);
\draw[dotted] (1.5,10) node[above=6pt, anchor=base]{\color{black}\small$i$} -- (1.5,0);
\draw[dotted] (0,7.25) node[left=1pt]{\color{black}\small$j$} -- (10,7.25);
\draw[dotted] (2.75,10) node[above=6pt, anchor=base]{\color{black}\small$j$} -- (2.75,0);
\draw[dotted] (0,5.75) node[left=1pt]{\color{black}\small$k$} -- (10,5.75) ;
\draw[dotted] (4.25,10) node[above=6pt, anchor=base]{\color{black}\small$k$} -- (4.25,0);
\draw[dotted] (0,3.5) node[left=1pt]{\color{black}\small$\si(j)$} -- (10,3.5);
\draw[dotted] (6.5,10) node[above=6pt, anchor=base]{\color{black}\small$\si(j)\;$} -- (6.5,0);
\draw[dotted] (0,1.5) node[left=1pt]{\color{black}\small$\si(i)$} -- (10,1.5);
\draw[dotted] (8.5,10) node[above=6pt, anchor=base]{\color{black}\small$\;\;\si(i)\!\!$} -- (8.5,0);
\draw (0,0) -- (10,0) -- (10,10) -- (0,10) -- (0,0);
\filldraw[fill=black] (4.25,5.75) circle (.25);
\filldraw[fill=black] (1.5,1.5) circle (.25);
\filldraw[fill=black] (8.5,8.5) circle (.25);
\filldraw[fill=black] (2.75,3.5) circle (.25);
\filldraw[fill=black] (6.5,7.25) circle (.25);
\etp
\qq\qq
\btp
\fill[gray!10] (1.5,3) rectangle (7,8.5);
\fill[gray!20] (3,5) rectangle (5,7);
\draw[dotted] (0,8.5) node[left=1pt]{\color{black}\small$i$} -- (10,8.5);
\draw[dotted] (1.5,10) node[above=6pt, anchor=base]{\color{black}\small$i$} -- (1.5,0);
\draw[dotted] (0,7) node[left=1pt]{\color{black}\small$j$} -- (10,7);
\draw[dotted] (3,10) node[above=6pt, anchor=base]{\color{black}\small$j$} -- (3,0);
\draw[dotted] (0,5) node[left=1pt]{\color{black}\small$\si(j)$} -- (10,5) ;
\draw[dotted] (5,10) node[above=6pt, anchor=base]{\color{black}\small$\si(j)\;$} -- (5,0);
\draw[dotted] (0,3) node[left=1pt]{\color{black}\small$\si(i)$} -- (10,3);
\draw[dotted] (7,10) node[above=6pt, anchor=base]{\color{black}\small$\;\;\si(i)\!$} -- (7,0);
\draw[dotted] (0,1) node[left=1pt]{\color{black}\small$k$} -- (10,1);
\draw[dotted] (9,10) node[above=6pt, anchor=base]{\color{black}\small$\;k\!$} -- (9,0);
\draw (0,0) -- (10,0) -- (10,10) -- (0,10) -- (0,0);
\filldraw[fill=black] (9,1) circle (.25);
\filldraw[fill=black] (1.5,3) circle (.25);
\filldraw[fill=black] (7,8.5) circle (.25);
\filldraw[fill=black] (3,5) circle (.25);
\filldraw[fill=black] (5,7) circle (.25);
\etp
\]

Since $K(u)$ is generically invertible, the considerations above imply that all simultaneously non-crossing non-split pairs $(i_r,\si(i_r))$ satisfying $i_r<\si(i_r)$ for $r=1,\dots,n$, where $n$ is the number of such pairs, form a sequence 
\equ{
(i,\si(i)),\; (i+1,\si(i)-1),\; (i+2,\si(i)-2),\;  \dots\;,\; (i+n-1,\si(i)-n+1), \label{K:seq-ii}
}
where $i= \min\{i_1,i_2,\dots,i_n\}$. It follows that 
\[
\si = 
\begin{pmatrix} \ell+1 & \ell+2 & \dots & r & t-r & t-r+1 & \dots & t-\ell+1 \\ t-\ell+1 & t-\ell & \dots & t-r & r & r-1 & \dots & \ell+1 \end{pmatrix} 
\]
with $\ell,r,t$ satisfying $0\le \ell < r \le \tfrac12(\ell+t)$ and $2 \le t \le N$, or
\aln{
\si &= \begin{pmatrix} 1 & 2 & \dots & \ell & m-\ell+1 & m-\ell+2 & \dots & m \\ m & m-1 & \dots & m-\ell+1 & \ell & \ell-1 & \dots & 1 \end{pmatrix} 
\el & \qq\times
\begin{pmatrix} m+1 & m+2 & \dots & r & N+m-r+1 & N+m-r+1 & \dots & N \\ N & N-1 & \dots & N+m-r+1 & r & r-1 & \dots & m+1 \end{pmatrix} 
}
with $\ell,m,r$ satisfying $2\le 2\ell < m < r \le \tfrac12 (N+m)$. In other words, we have that $(\ell,r,t,\si)\in\Sigma^{\rm sym'}_N$ and $(\ell,m,r,\si)\in\Sigma^{\rm sym''}_N$, respectively. Typical configurations for the first and the second case are shown below (the grey lines denote locations of the only nonzero matrix entries): 
\[
\btp
\fill[gray!10] (2,12) rectangle (11,3);
\fill[white] (4.5,9.5) rectangle (8.5,5.5);
\draw[dotted] (-0.5,12) node[left=1pt]{\color{black}\small$\ell\!+\!1$} -- (14.5,12);
\draw[dotted] (2,14.5) node[above=6pt, anchor=base]{\color{black}\small$\ell\!+\!1$} -- (2,-0.5);
\draw[dotted] (-0.5,9.5) node[left=1pt]{\color{black}\small$r$} -- (14.5,9.5);
\draw[dotted] (4.5,14.5) node[above=6pt, anchor=base]{\color{black}\small$r$} -- (4.5,-0.5);
\draw[dotted] (-0.5,5.5) node[left=1pt]{\color{black}\small$t\!-\!r$} -- (14.5,5.5);
\draw[dotted] (8.5,14.5) node[above=6pt, anchor=base]{\color{black}\small$\!t\!-\!r\;$} -- (8.5,0);
\draw[dotted] (-0.5,3) node[left=1pt]{\color{black}\small$t\!-\!\ell\!+\!1$} -- (14.5,3);
\draw[dotted] (11,14.5) node[above=6pt, anchor=base]{\color{black}\small$\;\;\;t\!-\!\ell\!+\!1\!\!\!\!\!$} -- (11,-0.5);
\draw (-0.5,-0.5) -- (14.5,-0.5) -- (14.5,14.5) -- (-0.5,14.5) -- (-0.5,-0.5);
\draw[gray!50,line width=4.5pt,line cap=round] (0,14) -- (2-.5,12+.5); 
\draw[gray!50,line width=4.5pt,line cap=round] (5,9) -- (8,6);
\draw[gray!50,line width=4.5pt,line cap=round] (11.5,2.5) -- (14,0);
\draw[gray!50,line width=4.5pt,line cap=round] (2,3) -- (4.5,5.5);
\draw[gray!50,line width=4.5pt,line cap=round] (8.5,9.5) -- (11,12);
\filldraw[fill=black] 
(0,14) circle (.25) (2-.5,12+.5) circle (.25)
(5,9) circle (.25) (8,6) circle (.25) (11.5,2.5) circle (.25) (14,0) circle (.25)
(2,3) circle (.25) (11,12) circle (.25) 
(4.5,5.5) circle (.25) (8.5,9.5) circle (.25);
\etp
\qq\qu
\btp
\fill[gray!10] (0,14) rectangle (8,6);
\fill[white] (2.5,11.5) rectangle (5.5,8.5);
\fill[gray!10] (14,0) rectangle (8.5,5.5);
\fill[white] (12.5,1.5) rectangle (10,4);
\draw[dotted] (-0.5,14) node[left=1pt]{\color{black}\small$1$} -- (14.5,14);
\draw[dotted] (0,14.5) node[above=6pt, anchor=base]{\color{black}\small$1$} -- (0,-0.5);
\draw[dotted] (-0.5,14-2.5) node[left=1pt]{\color{black}\small$\ell$} -- (14.5,14-2.5);
\draw[dotted] (2.5,14.5) node[above=6pt, anchor=base]{\color{black}\small$\ell$} -- (2.5,-0.5);
\draw[dotted] (-0.5,14-5.5) node[left=1pt]{\color{black}\small$m\!-\!\ell\!+\!1$} -- (14.5,14-5.5);
\draw[dotted] (5.5,14.5) 
-- (5.5,-0.5);
\draw (-0.5,14-7.65) node[left=1pt]{\color{black}\small$m$};
\draw[dotted] (-0.5,14-8) -- (14.5,14-8);
\draw[dotted] (8,14.5) node[above=6pt, anchor=base]{\color{black}\small$m$} -- (8,-0.5);
\draw[dotted] (-0.5,14-8.75) node[left=1pt]{\color{black}\small$m\!+\!1$};
\draw[dotted] (-0.5,14-8.5) -- (14.5,14-8.5);
\draw[dotted] (8.5,14.5) 
-- (8.5,-0.5);
\draw[dotted] (-0.5,14-10) node[left=1pt]{\color{black}\small$r$} -- (14.5,14-10);
\draw[dotted] (10,14.5) node[above=6pt, anchor=base]{\color{black}\small$r$} -- (10,-0.5);
\draw[dotted] (-0.5,14-12.5) node[left=1pt]{\color{black}\small$N\!+\!m\!-\!r\!+\!1$} -- (14.5,14-12.5);
\draw[dotted] (12.5,14.5) 
-- (12.5,-0.5);
\draw[dotted] (-0.5,0) node[left=1pt]{\color{black}\small$N$} -- (14.5,0);
\draw[dotted] (14,14.5) node[above=6pt, anchor=base]{\color{black}\small$N$} -- (14,-0.5);
\draw (-0.5,-0.5) -- (14.5,-0.5) -- (14.5,14.5) -- (-0.5,14.5) -- (-0.5,-0.5);
\draw[gray!50,line width=4.5pt,line cap=round] (0,6) -- (2.5,8.5);
\draw[gray!50,line width=4.5pt,line cap=round] (5.5,11.5) -- (8,14) ;
\draw[gray!50,line width=4.5pt,line cap=round] (3,11) -- (5,9);
\draw[gray!50,line width=4.5pt,line cap=round] (8.5,0) -- (10,1.5);
\draw[gray!50,line width=4.5pt,line cap=round] (12.5,4) -- (14,5.5) ;
\draw[gray!50,line width=4.5pt,line cap=round] (10.5,3.5) -- (12,2);
\filldraw[fill=black] 
(0,6) circle (.25) (2.5,8.5) circle (.25)
(5.5,11.5) circle (.25) (8,14) circle (.25)
(3,11) circle (.25) (5,9) circle (.25)
(8.5,0) circle (.25) (10,1.5) circle (.25)
(12.5,4) circle (.25) (14,5.5) circle (.25)
(10.5,3.5) circle (.25) (12,2) circle (.25);
\etp
\]

In the first case \eqref{K:S2:A1} and \eqref{K:S2:A2} imply that, up to an overall scalar factor,
\aln{
K(u) &= \sum_{1\le i \le \ell} u\, E_{ii} + \sum_{\ell< i\le r} \veps_{i}\,( E_{i,t-i+1} + E_{t-i+1,i}) + \sum_{r<i\le t-r} \chi\,E_{ii} + \sum_{t-\ell<i\le N} u^{-1} E_{ii},  
\intertext{with $\veps_{i},\chi=\pm1$. In the second case, by the same arguments as before,}
K(u) &= \sum_{1\le i \le \ell} u\,\veps_{i}\,(E_{i,m-i+1} + E_{m-i+1,i}) + \sum_{\ell<i\le m-\ell} u\,E_{ii} \el & \hspace{2.6cm} + \sum_{c< i\le r} \veps_{i}\,(E_{i,N+m-i+1} +E_{N+m-i+1,i}) + \sum_{r<i< N+m-r} \chi\,E_{ii} , 
}
with $\veps_{i},\chi=\pm1$. These are $\la=\mu=1$ and $c_i = \pm1$ specializations of \eqref{KS1} and \eqref{KS2}, respectively. By Lemma \ref{L:[K]} and Remark \ref{R:spec} they are equivalent to the matrix \eqref{KA3:X0}.


{\noindent \it Step 4: If $K(u)$ is a symmetric generalized cross matrix \eqref{K:S1:A} but not a generalized involution matrix, then $K$ is equivalent to the matrix \eqref{KS}.}


We start by studying ratios between the matrix entries associated to the pairs $(i,\si(i))$ and $(k,k)$ satisfying $i<\si(i)$ and $k=\si(k)$. There are three cases we need to consider, $k<i<\si(i)$, $i<k<\si(i)$ and $i<\si(i)<k$. The corresponding configurations are represented below
\[
\btp
\fill[gray!20] (3,3) rectangle (6,6);
\draw[dotted] (0,7.5) node[left=1pt]{\color{black}\small$k$} -- (9,7.5);
\draw[dotted] (1.5,9) node[above=6pt, anchor=base]{\color{black}\small$k$} -- (1.5,0);
\draw[dotted] (0,6) node[left=1pt]{\color{black}\small$i$} -- (9,6) ;
\draw[dotted] (3,9) node[above=6pt, anchor=base]{\color{black}\small$i$} -- (3,0);
\draw[dotted] (0,3) node[left=1pt]{\color{black}\small$\si(i)$} -- (9,3);
\draw[dotted] (6,9) node[above=6pt, anchor=base]{\color{black}\small$\si(i)$} -- (6,0);
\draw (0,0) -- (9,0) -- (9,9) -- (0,9) -- (0,0);
\filldraw[fill=black] (1.5,7.5) circle (.25);
\filldraw[fill=black] (3,3) circle (.25) (3,6) circle (.25) (6,3) circle (.25) (6,6) circle (.25) ;
\etp
\qq\qu
\btp
\fill[gray!20] (3,3) rectangle (6,6);
\draw[dotted] (0,4.8) node[left=1pt]{\color{black}\small$k$} -- (9,4.8);
\draw[dotted] (4.2,9) node[above=6pt, anchor=base]{\color{black}\small$k\;$} -- (4.2,0);
\draw[dotted] (0,6) node[left=1pt]{\color{black}\small$i$} -- (9,6) ;
\draw[dotted] (3,9) node[above=6pt, anchor=base]{\color{black}\small$i$} -- (3,0);
\draw[dotted] (0,3) node[left=1pt]{\color{black}\small$\si(i)$} -- (9,3);
\draw[dotted] (6,9) node[above=6pt, anchor=base]{\color{black}\small$\;\si(i)$} -- (6,0);
\draw (0,0) -- (9,0) -- (9,9) -- (0,9) -- (0,0);
\filldraw[fill=black] (4.2,4.8) circle (.25);
\filldraw[fill=black] (3,3) circle (.25) (3,6) circle (.25) (6,3) circle (.25) (6,6) circle (.25) ;
\etp
\qq\qu
\btp
\fill[gray!20] (3,3) rectangle (6,6);
\draw[dotted] (0,1.5) node[left=1pt]{\color{black}\small$k$} -- (9,1.5);
\draw[dotted] (7.5,9) node[above=6pt, anchor=base]{\color{black}\small$\;k$} -- (7.5,0);
\draw[dotted] (0,6) node[left=1pt]{\color{black}\small$i$} -- (9,6) ;
\draw[dotted] (3,9) node[above=6pt, anchor=base]{\color{black}\small$i$} -- (3,0);
\draw[dotted] (0,3) node[left=1pt]{\color{black}\small$\si(i)$} -- (9,3);
\draw[dotted] (6,9) node[above=6pt, anchor=base]{\color{black}\small$\si(i)\;$} -- (6,0);
\draw (0,0) -- (9,0) -- (9,9) -- (0,9) -- (0,0);
\filldraw[fill=black] (7.5,1.5) circle (.25);
\filldraw[fill=black] (3,3) circle (.25) (3,6) circle (.25) (6,3) circle (.25) (6,6) circle (.25) ;
\etp
\] 
Substituting $i\to\si(i)$ and $j\to k$ in \eqref{K:22} gives
\aln{
a(\tfrac{u}{v}) \Big(  b_{ki}(u v)\,\key_{\si(i)i}(u)\,\key_{i\si(i)}(v) + b_{k\si(i)}(u v)\,\key_{\si(i)\si(i)}(u)\, \key_{\si(i)\si(i)}(v) - b_{\si(i)k}(u v)\,\key_{kk}(u)\, \key_{kk}(v) \Big)
 & \\
 + a(u v) \Big( b_{\si(i)k}(\tfrac{u}{v})\,\key_{\si(i)\si(i)}(v)\,\key_{kk}(u) - b_{k\si(i)}(\tfrac{u}{v})\,\key_{\si(i)\si(i)}(u)\,\key_{kk}(v) \Big) &= 0. 
}
Assume that $i<k<\si(i)$. Using \eqref{ab} we obtain 
\aln{
& (1-u v)\,\Big(u\,\key_{kk}(v)\,\key_{\si(i)\si(i)}(u) - v\,\key_{kk}(u)\,\key_{\si(i)\si(i)}(v)\Big) \\ & \qu + (u-v) \Big(u v\,\key_{\si(i)\si(i)}(u)\,\key_{\si(i)\si(i)}(v) + \key_{\si(i)i}(u)\,\key_{i\si(i)}(v) - \key_{kk}(u)\,\key_{kk}(v) \Big)  = 0.
}
It remains to use \eqref{K:S2:B1} and \eqref{K:S2:C} giving
\[
\frac{uv\,(u-v)(1-u^2) (1-v^2)}{(1-\al_{i\si(i)}u) (1-\al_{i\si(i)}v)}\Big(\ga_{i\si(i)k}(\ga_{i\si(i)k}+1)-\beta_{i\si(i)}^2\Big) \key_{\si(i)\si(i)}(u)\,\key_{\si(i)\si(i)}(v) = 0.
\]
The equality above is only true if $\beta_{i\si(i)}^2=\ga_{i\si(i)k}(\ga_{i\si(i)k}+1)$ with $\ga_{i\si(i)k}\in\C$. We have thus arrived at the following result, when $i<k<\si(i)$,
\ali{
\frac{\key_{kk}(u)}{\key_{\si(i)\si(i)}(u)} &= \bigg( 1 + \frac{u-u^{-1}}{\al_{i\si(i)} - u^{-1}}\,\ga_{i\si(i)k}\bigg) , \qq && \beta_{i\si(i)}^2=\ga_{i\si(i)k}(\ga_{i\si(i)k}+1) . 
\label{K:iki}
\intertext{Now assume that $k<i<\si(i)$ instead. Repeating the same steps as before we find} 
\frac{\key_{kk}(u)}{\key_{\si(i)\si(i)}(u)} &= \bigg( 1 + \frac{u-u^{-1}}{\al_{i\si(i)}-u^{-1}}\,(u \ga_{i\si(i)k}-1)\bigg), \qq &&\beta_{i\si(i)}^2=\ga_{i\si(i)k}(\ga_{i\si(i)k}-\al_{i\si(i)}). \label{K:kii}
\intertext{Finally, assume that $i<\si(i)<k$. This time we find} 
\frac{\key_{kk}(u)}{\key_{\si(i)\si(i)}(u)} &= \bigg( 1 + \frac{u-u^{-1}}{\al_{i\si(i)} - u^{-1}}\,u^{-1}\ga_{i\si(i)k}\bigg) , \qq && \beta_{i\si(i)}^2=\ga_{i\si(i)k}(\ga_{i\si(i)k}+\al_{i\si(i)}). \label{K:iik}
}

We now focus on the ratios between the matrix entries associated to the pairs $(i,\si(i))$ and $(j,\si(j))$ satisfying $i<\si(i)$ and $j<\si(j)$. Note that \eqref{K:ij-kl} holds, i.e., the pairs $(i,\si(i))$ and $(j,\si(j))$ must be non-crossing. Moreover, there can be at most two simultaneously non-crossing split pairs. Hence there are two cases we need to consider, $i<\si(i)<j<\si(j)$ and $i<j<\si(j)<\si(i)$. The configurations in question are represented below
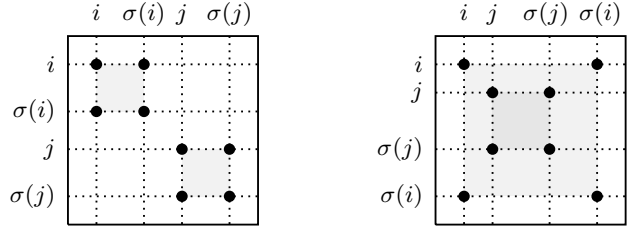
\[
\btp
\fill[gray!10] (1.5,6) rectangle (4,8.5);
\fill[gray!10] (6,4) rectangle (8.5,1.5);
\draw[dotted] (0,8.5) node[left=1pt]{\color{black}\small$i$} -- (10,8.5) ;
\draw[dotted] (1.5,10) node[above=6pt, anchor=base]{\color{black}\small$i$} -- (1.5,0);
\draw[dotted] (0,6) node[left=1pt]{\color{black}\small$\si(i)$} -- (10,6) ;
\draw[dotted] (4,10) node[above=6pt, anchor=base]{\color{black}\small$\si(i)$} -- (4,0);
\draw[dotted] (0,4) node[left=1pt]{\color{black}\small$j$} -- (10,4);
\draw[dotted] (6,10) node[above=6pt, anchor=base]{\color{black}\small$j$} -- (6,0);
\draw[dotted] (0,1.5) node[left=1pt]{\color{black}\small$\si(j)$} -- (10,1.5);
\draw[dotted] (8.5,10) node[above=6pt, anchor=base]{\color{black}\small$\si(j)$} -- (8.5,0);
\draw (0,0) -- (10,0) -- (10,10) -- (0,10) -- (0,0);
\filldraw[fill=black] (4,8.5) circle (.25);
\filldraw[fill=black] (1.5,6) circle (.25);
\filldraw[fill=black] (4,6) circle (.25);
\filldraw[fill=black] (1.5,8.5) circle (.25);
\filldraw[fill=black] (8.5,4) circle (.25);
\filldraw[fill=black] (6,1.5) circle (.25);
\filldraw[fill=black] (8.5,1.5) circle (.25);
\filldraw[fill=black] (6,4) circle (.25);
\etp
\qq\qq
\btp
\fill[gray!10] (1.5,8.5) rectangle (8.5,1.5);
\fill[gray!20] (3,4) rectangle (6,7);
\draw[dotted] (0,8.5) node[left=1pt]{\color{black}\small$i$} -- (10,8.5);
\draw[dotted] (1.5,10) node[above=6pt, anchor=base]{\color{black}\small$i$} -- (1.5,0);
\draw[dotted] (0,7) node[left=1pt]{\color{black}\small$j$} -- (10,7) ;
\draw[dotted] (3,10) node[above=6pt, anchor=base]{\color{black}\small$j$} -- (3,0);
\draw[dotted] (0,4) node[left=1pt]{\color{black}\small$\si(j)$} -- (10,4);
\draw[dotted] (6,10) node[above=6pt, anchor=base]{\color{black}\small$\si(j)\;$} -- (6,0);
\draw[dotted] (0,1.5) node[left=1pt]{\color{black}\small$\si(i)$} -- (10,1.5);
\draw[dotted] (8.5,10) node[above=6pt, anchor=base]{\color{black}\small$\;\si(i)$} -- (8.5,0);
\draw (0,0) -- (10,0) -- (10,10) -- (0,10) -- (0,0);
\filldraw[fill=black] (1.5,1.5) circle (.25);
\filldraw[fill=black] (8.5,8.5) circle (.25);
\filldraw[fill=black] (1.5,8.5) circle (.25);
\filldraw[fill=black] (8.5,1.5) circle (.25);
\filldraw[fill=black] (3,4) circle (.25);
\filldraw[fill=black] (6,7) circle (.25);
\filldraw[fill=black] (3,7) circle (.25);
\filldraw[fill=black] (6,4) circle (.25);
\etp
\]
Since $K(u)$ is a generalized cross matrix, the filled nodes represent the only nonzero entries in the dotted lines and columns.

Let $i,j$ be such that $i<\si(i)<j<\si(j)$. The relation \eqref{K:S2:C} implies that, for $\ga_{j\si(j)\si(i)}\in\C$,
\[
\frac{\key_{\si(i)\si(i)}(u)}{\key_{\si(j)\si(j)}(u)} = 1 + \frac{(u-u^{-1})(u\,\ga_{j\si(j)\si(i)}-1)}{\al_{j\si(j)} - u^{-1}} .
\]
By \eqref{K:ij-kl} we have $\key_{i\si(i)}(u)/\key_{j\si(j)}(u)\in \C u$. Then, using \eqref{K:S2:B1}, we deduce that, for $\eta\in\C^\times$,
\[
\frac{\key_{\si(i)\si(i)}(u)}{\key_{\si(j)\si(j)}(u)} = \frac{\al_{i\si(i)} - u^{-1}}{\al_{j\si(j)} - u^{-1}}\,u\,\eta . 
\]
By equating the above expressions we find 
\equ{
\ga_{j\si(j)\si(i)}=0,\qq \eta+\al_{j\si(j)}=0, \qq \al_{i\si(i)}\al_{j\si(j)}=1. \label{K:S3:gna}
}
Hence, upon combining with \eqref{K:S2:B2}, we have, when $i<\si(i)<j<\si(j)$,
\equ{
\frac{\key_{\si(i)\si(i)}(u)}{\key_{\si(j)\si(j)}(u)} = \frac{\al_{j\si(j)} - u }{\al_{j\si(j)} - u^{-1}} , \qq \frac{\key_{ii}(u)}{\key_{jj}(u)} = \frac{1 - u\,\al_{j\si(j)}}{1 - u^{-1} \al_{j\si(j)}}, \label{K:S3:diag1}
}
and so
\equ{
\frac{\key_{ii}(u)}{\key_{\si(j)\si(j)}(u)} = u^2 , \qq 
\frac{\key_{jj}(u)}{\key_{\si(i)\si(i)}(u)} = 1 . \label{K:S3:diag2}
}

Next substitute $i\to\si(i)$ and $j\to\si(j)$ in \eqref{K:22}. Provided $i<\si(i)<j<\si(j)$, this gives
\aln{
& a(\tfrac{u}{v}) \Big( b_{\si(j)i}(u v)\, \key_{\si(i)i}(u)\, \key_{i\si(i)}(v) + b_{\si(j)\si(i)}(u v)\,\key_{\si(i)\si(i)}(u)\, \key_{\si(i)\si(i)}(v) \\ & \qq\qu - b_{\si(i)j}(u v)\,\key_{j\si(j)}(u) \, \key_{\si(j)j}(v) -  b_{\si(i)\si(j)}(u v)\,\key_{\si(j)\si(j)}(u)\, \key_{\si(j)\si(j)}(v) \Big) 
\\
& + a(u v) \Big(b_{\si(i)\si(j)}(\tfrac{u}{v})\,\key_{\si(i)\si(i)}(v)\, \key_{\si(j)\si(j)}(u) - b_{\si(j)\si(i)}(\tfrac{u}{v})\,\key_{\si(i)\si(i)}(u)\,\key_{\si(j)\si(j)}(v) \Big) = 0.
}
Using \eqref{ab} and dividing by the common scalar factor we rewrite the equality above as 
\aln{
& (u-v) \Big( \key_{\si(i)i}(u)\,\key_{i\si(i)}(v) + \key_{\si(i)\si(i)}(u)\,\key_{\si(i)\si(i)}(v) \Big) \\ & \qu - u v\,(u-v)\,\Big(\key_{j\si(j)}(u)\,\key_{\si(j)j}(v) + \key_{\si(j)\si(j)}(u)\,\key_{\si(j)\si(j)}(v) \Big) \\ &\qu - (1-u v)\,\Big( u\,\key_{\si(i)\si(i)}(v)\,\key_{\si(j)\si(j)}(u) - v\,\key_{\si(i)\si(i)}(u)\,\key_{\si(j)\si(j)}(v) \Big) = 0.
}
Finally, using \eqref{K:S2:B1}, \eqref{K:22}, \eqref{K:S3:gna} and \eqref{K:S3:diag1} we obtain
\[
\frac{ uv\,(u-v) (1-u^2) (1-v^2) }{(1-\al_{j\si(j)}u) (1-\al_{j\si(j)}v)}\,\Big(\beta _{i\si(i)}^2 \al_{j\si(j)}^2-\beta_{j\si(j)}^2\Big)\, \key_{\si(j)\si(j)}(u)\,\key_{\si(j)\si(j)}(v) = 0.
\]
Hence the above equality is only true if $\beta_{i\si(i)}^2 = \al_{j\si(j)}^{-2} \beta_{j\si(j)}^2$. We have thus arrived at the following result, when $i<\si(i)<j<\si(j)$,
\equ{
\frac{\key_{i\si(i)}(u)}{\key_{\si(j)\si(j)}(u)} = \pm \frac{u-u^{-1}}{\al_{j\si(j)} - u^{-1} } \, u\,\beta_{j\si(j)} . \label{K:iijj}
}

Let $i,j$ be such that $i<j<\si(j)<\si(i)$. The relation \eqref{K:S2:C} now implies that, for $\ga_{i\si(i)\si(j)}\in\C$,
\[
\frac{\key_{\si(j)\si(j)}(u)}{ \key_{\si(i)\si(i)}(u)} = 1 + \frac{u-u^{-1}}{\al_{i\si(i)} - u^{-1}}\,\ga_{i\si(i)\si(j)}.
\]
By \eqref{K:ij-kl} we have $\key_{i\si(i)}(u)/\key_{j\si(j)}(u)\in \C$. Then, using \eqref{K:S2:B1} and \eqref{K:S2:B2}, we deduce that that, for $\eta\in\C^\times$, 
\[
\frac{\key_{\si(j)\si(j)}(u)}{\key_{\si(i)\si(i)}(u)} = \frac{\al_{j\si(j)} - u^{-1}}{\al_{i\si(i)} - u^{-1}}\,\eta , 
\]
yielding $\ga_{i\si(i)\si(j)}=0$, $\eta=1$ and $\al_{i\si(i)}=\al_{j\si(j)}$. Upon combining with \eqref{K:S2:B2} we obtain
\equ{
\frac{\key_{ii}(u)}{\key_{jj}(u)} = \frac{\key_{\si(i)\si(i)}(u)}{\key_{\si(j)\si(j)}(u)} = 1 . \label{K:S3:diag3}
}
Finally, by repeating the same steps as before, i.e., by combining the above results with \eqref{K:S2:B1} and \eqref{K:22}, we find $\beta_{i\si(i)}^2 = \beta_{j\si(j)}^2$. Therefore, when $i<j<\si(j)<\si(i)$,
\equ{
\frac{\key_{i\si(i)}(u)}{\key_{j\si(j)}(u)} = \pm1 , \label{K:ijji}
}


The last ingredient that we need is a constraint on the choice of the pairs $(k,\si(k))$ with $k=\si(k)$. Let $i,j,k$ be such that $i<k<j<\si(j)<\si(i)$ and $k=\si(k)$. We already know that in this case $\al_{i\si(i)}=\al_{j\si(j)}$ and $\beta_{i\si(i)}^2= \beta_{j\si(j)}^2$. Thus, upon combining \eqref{K:S3:diag3}, \eqref{K:ijji}, \eqref{K:iki} and \eqref{K:kii}, we find
\[
1 + \frac{u-u^{-1} }{\al_{i\si(i)} - u^{-1}} \,\ga_{i\si(i)k} = 1 +  \frac{u-u^{-1}}{\al_{i\si(i)} - u^{-1}}\,(1+u \ga_{j\si(j)k}) ,
\]
where $\beta_{i\si(i)}^2=\ga_{i\si(i)k}(\ga_{i\si(i)k}+1)=\ga_{j\si(j)k}(\ga_{j\si(j)k}-\al_{j\si(j)})$. The equality above is only true if $\ga_{i\si(i)k}=-1$ and $\ga_{j\si(j)k}=0$. But then $\beta_{i\si(i)}=0$ yielding a contradiction. Now let $i,j,k$ be such that $i<j<\si(j)<k<\si(i)$ and $k=\si(k)$ instead. Using similar arguments as above, with \eqref{K:iik} instead of \eqref{K:kii}, we again obtain a contradiction. In other words, the following two configurations
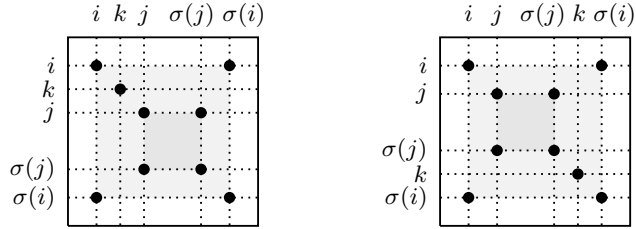
\[
\btp
\fill[gray!10] (1.5,8.5) rectangle (8.5,1.5);
\fill[gray!20] (4,3) rectangle (7,6);
\draw[dotted] (0,8.5) node[left=1pt]{\color{black}\small$i$} -- (10,8.5);
\draw[dotted] (1.5,10) node[above=6pt, anchor=base]{\color{black}\small$i$} -- (1.5,0);
\draw[dotted] (0,7.25) node[left=1pt]{\color{black}\small$k$} -- (10,7.25);
\draw[dotted] (2.75,10) node[above=6pt, anchor=base]{\color{black}\small$k$} -- (2.75,0);
\draw[dotted] (0,6) node[left=1pt]{\color{black}\small$j$} -- (10,6) ;
\draw[dotted] (4,10) node[above=6pt, anchor=base]{\color{black}\small$j$} -- (4,0);
\draw[dotted] (0,3) node[left=1pt]{\color{black}\small$\si(j)$} -- (10,3);
\draw[dotted] (7,10) node[above=6pt, anchor=base]{\color{black}\small$\!\si(j)\;\;$} -- (7,0);
\draw[dotted] (0,1.5) node[left=1pt]{\color{black}\small$\si(i)$} -- (10,1.5);
\draw[dotted] (8.5,10) node[above=6pt, anchor=base]{\color{black}\small$\;\;\;\si(i)\!\!$} -- (8.5,0);
\draw (0,0) -- (10,0) -- (10,10) -- (0,10) -- (0,0);
\filldraw[fill=black] (2.75,7.25) circle (.25);
\filldraw[fill=black] (1.5,1.5) circle (.25) (8.5,8.5) circle (.25) (1.5,8.5) circle (.25) (8.5,1.5) circle (.25);
\filldraw[fill=black] (4,3) circle (.25) (7,3) circle (.25) (4,6) circle (.25) (7,6) circle (.25);
\etp
\qq\qq
\btp
\fill[gray!10] (1.5,8.5) rectangle (8.5,1.5);
\fill[gray!20] (3,4) rectangle (6,7);
\draw[dotted] (0,8.5) node[left=1pt]{\color{black}\small$i$} -- (10,8.5);
\draw[dotted] (1.5,10) node[above=6pt, anchor=base]{\color{black}\small$i$} -- (1.5,0);
\draw[dotted] (0,7) node[left=1pt]{\color{black}\small$j$} -- (10,7) ;
\draw[dotted] (3,10) node[above=6pt, anchor=base]{\color{black}\small$j$} -- (3,0);
\draw[dotted] (0,4) node[left=1pt]{\color{black}\small$\si(j)$} -- (10,4);
\draw[dotted] (6,10) node[above=6pt, anchor=base]{\color{black}\small$\!\!\si(j)\;\;$} -- (6,0);
\draw[dotted] (0,1.5) node[left=1pt]{\color{black}\small$\si(i)$} -- (10,1.5);
\draw[dotted] (8.5,10) node[above=6pt, anchor=base]{\color{black}\small$\;\;\si(i)\!\!\!$} -- (8.5,0);
\draw[dotted] (0,2.75) node[left=1pt]{\color{black}\small$k$} -- (10,2.75);
\draw[dotted] (7.25,10) node[above=6pt, anchor=base]{\color{black}\small$\,k$} -- (7.25,0);
\draw (0,0) -- (10,0) -- (10,10) -- (0,10) -- (0,0);
\filldraw[fill=black] (7.25,2.75) circle (.25);
\filldraw[fill=black] (1.5,1.5) circle (.25) (8.5,8.5) circle (.25) (1.5,8.5) circle (.25) (8.5,1.5) circle (.25);
\filldraw[fill=black] (3,4) circle (.25) (6,7) circle (.25) (3,7) circle (.25) (6,4) circle (.25);
\etp
\]
are {\it not} allowed. Recall that the filled nodes represent the only nonzero entries in the dotted lines and columns. This means that all simultaneously non-crossing non-split pairs $(i_r,\si(i_r))$ satisfying $i_r<\si(i_r)$ with $r=1,\dots,n$, where $n$ is the number of such pairs, form the sequence \eqref{K:seq-ii}. Moreover, by \eqref{K:S3:diag3} and \eqref{K:ijji}, $\al_{i_1\si(i_1)}=\dots=\al_{i_n\si(i_n)}$ and so
\equ{
\key_{i_1i_1}(u)=\dots=\key_{i_ni_n}(u) , \qq
\key_{\si(i_1)\si(i_1)}(u)=\dots=\key_{\si(i_1)\si(i_n)}(u) . \label{K:diags1}
}
By \eqref{K:S2:A1} we have $\beta^2_{i_1\si(i_1)}=\dots=\beta^2_{i_n\si(i_n)}$. Finally, by combining \eqref{K:S2:B2} with \eqref{K:iki}, \eqref{K:kii} and \eqref{K:iik}, we deduce that, for all $i,i',j,j',k,k'$ fixed under $\si$ and satisfying $i,i'< i_1$,\; $i_n< j,j'<\si(i_n)$ and $\si(i_1)< k,k'$ the relations \eqref{K:diags2} hold.


The above results imply that $\si$ must be given by $\Sigma^{\rm sym'}_N$ or $\Sigma^{\rm sym''}_N$. The nonzero entries of $K(u)$ can thus be represented as follows (recall a similar result in Step 3): 
\[
\btp
\fill[gray!10] (2,12) rectangle (11,3);
\fill[gray!20] (4.5,9.5) rectangle (8.5,5.5);
\draw[dotted] (-0.5,12) node[left=1pt]{\color{black}\small$\ell\!+\!1$} -- (14.5,12);
\draw[dotted] (2,14.5) node[above=6pt, anchor=base]{\color{black}\small$\ell\!+\!1$} -- (2,-0.5);
\draw[dotted] (-0.5,9.5) node[left=1pt]{\color{black}\small$r$} -- (14.5,9.5);
\draw[dotted] (4.5,14.5) node[above=6pt, anchor=base]{\color{black}\small$r$} -- (4.5,-0.5);
\draw[dotted] (-0.5,5.5) node[left=1pt]{\color{black}\small$t\!-\!r$} -- (14.5,5.5);
\draw[dotted] (8.5,14.5) node[above=6pt, anchor=base]{\color{black}\small$\!t\!-\!r\;$} -- (8.5,0);
\draw[dotted] (-0.5,3) node[left=1pt]{\color{black}\small$t\!-\!\ell\!+\!1$} -- (14.5,3);
\draw[dotted] (11,14.5) node[above=6pt, anchor=base]{\color{black}\small$\;\;\;t\!-\!\ell\!+\!1\!\!\!\!\!$} -- (11,-0.5);
\draw (-0.5,-0.5) -- (14.5,-0.5) -- (14.5,14.5) -- (-0.5,14.5) -- (-0.5,-0.5);
\draw[gray!50,line width=4.5pt,line cap=round] (0,14) -- (14,0); 
\draw[gray!50,line width=4.5pt,line cap=round] (5,9) -- (8,6);
\draw[gray!50,line width=4.5pt,line cap=round] (2,3) -- (4.5,5.5);
\draw[gray!50,line width=4.5pt,line cap=round] (8.5,9.5) -- (11,12);
\filldraw[fill=black] 
(0,14) circle (.25) (2,12) circle (.25)
(4.5,9.5) circle (.25) (8.5,5.5) circle (.25) (11,3) circle (.25) (14,0) circle (.25)
(2,3) circle (.25) (11,12) circle (.25) 
(4.5,5.5) circle (.25) (8.5,9.5) circle (.25);
\etp
\qq\qu
\btp
\fill[gray!10] (0,14) rectangle (8,6);
\fill[gray!20] (2.5,11.5) rectangle (5.5,8.5);
\fill[gray!10] (14,0) rectangle (8.5,5.5);
\fill[gray!20] (12.5,1.5) rectangle (10,4);
\draw[dotted] (-0.5,14) node[left=1pt]{\color{black}\small$1$} -- (14.5,14);
\draw[dotted] (0,14.5) node[above=6pt, anchor=base]{\color{black}\small$1$} -- (0,-0.5);
\draw[dotted] (-0.5,14-2.5) node[left=1pt]{\color{black}\small$\ell$} -- (14.5,14-2.5);
\draw[dotted] (2.5,14.5) node[above=6pt, anchor=base]{\color{black}\small$\ell$} -- (2.5,-0.5);
\draw[dotted] (-0.5,14-5.5) node[left=1pt]{\color{black}\small$m\!-\!\ell\!+\!1$} -- (14.5,14-5.5);
\draw[dotted] (5.5,14.5) 
-- (5.5,-0.5);
\draw (-0.5,14-7.65) node[left=1pt]{\color{black}\small$m$};
\draw[dotted] (-0.5,14-8) -- (14.5,14-8);
\draw[dotted] (8,14.5) node[above=6pt, anchor=base]{\color{black}\small$m$} -- (8,-0.5);
\draw[dotted] (-0.5,14-8.75) node[left=1pt]{\color{black}\small$m\!+\!1$};
\draw[dotted] (-0.5,14-8.5) -- (14.5,14-8.5);
\draw[dotted] (8.5,14.5) 
-- (8.5,-0.5);
\draw[dotted] (-0.5,14-10) node[left=1pt]{\color{black}\small$r$} -- (14.5,14-10);
\draw[dotted] (10,14.5) node[above=6pt, anchor=base]{\color{black}\small$r$} -- (10,-0.5);
\draw[dotted] (-0.5,14-12.5) node[left=1pt]{\color{black}\small$N\!+\!m\!-\!r\!+\!1$} -- (14.5,14-12.5);
\draw[dotted] (12.5,14.5) 
-- (12.5,-0.5);
\draw[dotted] (-0.5,0) node[left=1pt]{\color{black}\small$N$} -- (14.5,0);
\draw[dotted] (14,14.5) node[above=6pt, anchor=base]{\color{black}\small$N$} -- (14,-0.5);
\draw (-0.5,-0.5) -- (14.5,-0.5) -- (14.5,14.5) -- (-0.5,14.5) -- (-0.5,-0.5);
%
\draw[gray!50,line width=4.5pt,line cap=round] (0,6) -- (2.5,8.5);
\draw[gray!50,line width=4.5pt,line cap=round] (5.5,11.5) -- (8,14) ;
\draw[gray!50,line width=4.5pt,line cap=round] (0,14) -- (14,0);
\draw[gray!50,line width=4.5pt,line cap=round] (8.5,0) -- (10,1.5);
\draw[gray!50,line width=4.5pt,line cap=round] (12.5,4) -- (14,5.5) ;
\filldraw[fill=black] 
(0,6) circle (.25) (2.5,8.5) circle (.25)
(5.5,11.5) circle (.25) (8,14) circle (.25)
(2.5,11.5) circle (.25) (5.5,8.5) circle (.25)
(0,14) circle (.25) (8,14-8) circle (.25)
(8.5,0) circle (.25) (10,1.5) circle (.25)
(12.5,4) circle (.25) (14,5.5) circle (.25)
(10,4) circle (.25) (12.5,1.5) circle (.25)
(8.5,5.5) circle (.25) (14,0) circle (.25);
\etp
\]

Let $(\ell,r,t,\si)\in\Sigma^{\rm sym'}_N$. By \eqref{K:S3:diag3}, \eqref{K:ijji}, \eqref{K:diags1} and \eqref{K:diags2} we have
\gan{
\key_{11}(u) = \dots = \key_{\ell\ell}(u) =  u^2\,\key_{t-l+2,t-l+2}(u) = \dots = u^2\,\key_{NN}(u) , \\
\key_{\ell+1,\ell+1}(u) = \dots = \key_{rr}(u) , \qq  \key_{t-r,r-r}(u) = \dots = \key_{t-\ell+1,t-\ell+1}(u) ,
\\
\key_{r+1,r+1}(u) = \dots = \key_{t-r-1,t-r-1}(u), \qq \key_{\ell+1,\si(\ell+1)}^2(u) = \dots = \key_{r\si(r)}^2(u) .
}
Without loss of generality, we set $\key_{r+1,r+1}(u)=1$ and 
\equ{
\al_{r\si(r)}=\frac{\mu-\mu^{-1}}{\la-\la^{-1}} , \qq \ga_{r\si(r),r+1} = \frac{\la^{-1}}{\la-\la^{-1}}, \qq \la,\mu\in\C^\times. \label{ab-lamu}
}
Using \eqref{K:S2:B1}, \eqref{K:S2:B2} and \eqref{K:iki} we find
\eqg{
\key_{rr}(u) = 1 + \frac{\la(u-u^{-1})}{(\la^{-1}\mu^{-1}+u^{-1})(\la-\mu\,u)}, \qq 
\key_{\si(r)\si(r)}(u) = 1 + \frac{\la^{-1}(u-u^{-1})}{(\la^{-1}\mu^{-1}+u^{-1})(\la-\mu\,u)} , \\
\key_{r\si(r)}(u) = \pm\frac{u-u^{-1}}{(\la^{-1}\mu^{-1}+u^{-1})(\la-\mu\,u)} , \qq \beta_{r\si(r)} = \mp\frac{1}{\la-\la^{-1}} . \label{K:rr}
}
Finally, using \eqref{K:kii}, \eqref{K:iik} and \eqref{K:diags2}, we find that 
\ali{
\key_{11}(u) &= 1 + \frac{u-u^{-1}}{\la^{-1} \mu^{-1} + u^{-1}}, \qq& \key_{NN}(u) &= 1 - \frac{u-u^{-1}}{(\la^{-1} \mu^{-1} + u^{-1})\la \mu u} . \label{K:11NN}
}
Hence, up to an overall scalar factor, 
\aln{ 
K(u) &= I + \frac{u-u^{-1}}{\la^{-1} \mu^{-1} + u^{-1}} \Bigg( \sum_{1\le i \le \ell} E_{ii} -  \frac{1}{\la\mu u} \sum_{t<i\le N} E_{ii} \\ & \hspace{4cm} + \frac{1}{\la-\mu\,u} \sum_{\ell< i\le r} \Big( \la E_{ii} + \la^{-1} E_{\si(i)\si(i)} \mp E_{i\si(i)} \mp E_{i\si(i)}\Big)\Bigg)
}
By Lemma \ref{L:[K]}, it is equivalent to \eqref{KS}.


Now let $(\ell,m,r,\si)\in\Sigma^{\rm sym''}_N$. By \eqref{K:S3:diag2}, \eqref{K:iijj}, \eqref{K:S3:diag3}, \eqref{K:ijji}, \eqref{K:diags1} and \eqref{K:diags2} we have   
\gan{
\key_{11}(u) = \dots = \key_{\ell\ell}(u) = u^2\,\key_{N-m-r+1}(u) = \dots = u^2\,\key_{NN}(u) , 
\\
\key_{m-\ell+1,m-\ell+1}(u) = \dots = \key_{mm}(u) = \key_{m+1,m+1}(u) = \dots = \key_{rr}(u),
\\
\key_{\ell+1,\ell+1}(u) = \dots = \key_{m-\ell,m-\ell}(u), \qq \key_{r+1,r+1}(u) = \dots = \key_{N+m-r,N+m-r}(u),
\\
\key_{1\si(1)}^2(u) = \dots = \key_{\ell\si(\ell)}^2(u) = u^2\,\key_{m+1,\si(m+1)}^2(u) = \dots = u^2\,\key_{r\si(r)}^2(u) .
}
As in the previous case, we set $\key_{r+1,r+1}(u)=1$ and choose the parametrization \eqref{ab-lamu}. Then \eqref{K:rr} also holds. Moreover, $\key_{\ell+1,\ell+1}(u)$ coincides with $\key_{11}(u)$ given by \eqref{K:11NN}. Hence, up to an overall scalar factor,
\aln{
K(u) = I + \frac{u-u^{-1}}{\la^{-1}\mu^{-1}+u^{-1}} &\Bigg( \sum_{1\le i \le \ell} E_{ii} + \sum_{\ell<i\le m-\ell} E_{ii} \el & \qu + \frac{1}{\la-\mu u} \sum_{1\le i \le \ell}\Big( \mu^{-1} u E_{ii} + \la E_{\si(i)\si(i)} \mp u E_{i\si(i)} \mp u E_{\si(i)i} \Big) \el & \qu + \frac{1}{\la-\mu u} \sum_{m< i\le r} \Big( \la E_{ii} + \la^{-1} E_{\si(i)\si(i)} \mp E_{i\si(i)} \mp E_{i\si(i)}\Big)\Bigg).
}
By Lemma 3.4, it is equivalent to \eqref{KS}.


{\noindent \it Step 5: If $K(u)$ is a non-symmetric generalized cross matrix, then $K$ is equivalent to \eqref{KT}.}

Let $K(u)$ be given by \eqref{K:S1:A}. 
Assume that $K(u)$ is a non-symmetric matrix, i.e., there exists $k$ such that $k<\si(k)$ and $\key_{k\si(k)}(u)\,\key_{\si(k)k}(u)=0$. 
Then for all $i$ such that $i<\si(i)$ we have $\key_{i\si(i)}(u)\,\key_{\si(i)i}(u)=0$, i.e., a non-symmetric $K(u)$ is triangular. We will show this by contradiction. 
(Note that if $k$ is the only index such that $k<\si(k)$, for example when $N=3$, then there is nothing to prove.)
Assume that there exists $i$ such that $i<\si(i)$ and $\key_{i\si(i)}(u)\,\key_{\si(i)i}(u)\neq 0$. 
There are four cases we need to consider: 
\equ{
k<\si(k)<i<\si(i) , \qu i<\si(i)<k<\si(k), \qu i<k<\si(k)<\si(i) , \qu k<i<\si(i)<\si(k). \label{K:S5:1}
}
Consider the first case above. The equality \eqref{K:22} in this case becomes, upon substituting $i\to\si(i)$ and $j\to k$,
\aln{
a(\tfrac{u}{v}) \Big( b_{ki}(u v)\,\key_{\si(i)i}(u)\,\key_{i\si(i)}(v) + b_{k\si(i)}(u v)\,\key_{\si(i)\si(i)}(u)\,\key_{\si(i)\si(i)}(v) - b_{\si(i)k}(u v)\,\key_{kk}(u)\,\key_{kk}(v) \Big) \qu& \\
+\ a(u v) \Big(b_{\si(i)k}(\tfrac uv)\,\key_{\si(i)\si(i)}(v)\,\key_{kk}(u) - b_{k\si(i)}(\tfrac uv)\,\key_{\si(i)\si(i)}(u)\,\key_{kk}(v) \Big) & = 0 .
}
By \eqref{K:S3:diag2}, which holds for $i,k$ described above, we have $\key_{kk}(u)=u^2 \key_{\si(i)\si(i)}(u)$ giving
\[
a(\tfrac{u}{v})\, b_{ki}(u v)\,\key_{\si(i)i}(u)\,\key_{i\si(i)}(v) + A(u,v)\, \key_{\si(i)\si(i)}(u)\, \key_{\si(i)\si(i)}(v) = 0 ,
\]
where
\aln{
A(u,v) &= a(u v) \left(u^2\, b_{\si(i)k}(\tfrac{u}{v})-v^2\, b_{k\si(i)}(\tfrac{u}{v})\right)+a(\tfrac{u}{v}) \left(b_{k\si(i)}(u v)-u^2 v^2\,b_{\si(i)k}(u v)\right) \\
& = \frac{1-uv}{q-q^{-1}uv} \left(\frac{(q^2-1)u^2v}{q^2v-u} - \frac{(q^2-1)uv^2}{q^2v-u} \right) +  \frac{v-u}{qv-q^{-1}u} \left(\frac{(q^2-1)uv}{q^2-uv} - \frac{(q^2-1)u^2v^2}{q^2-uv} \right) \\
& = \frac{\left(q^2-1\right) u v (1-u v)(u-v)}{\left(q-q^{-1}u v\right)\left(q^2 v-u\right) } - \frac{\left(q^2-1\right) u v (v-u) (u v-1)}{\left(q v-q^{-1}u\right) \left(q^2-u v\right)} = 0.
}
Hence $a(\tfrac{u}{v})\, b_{ki}(u v)\,\key_{\si(i)i}(u)\,\key_{i\si(i)}(v)=0$ yielding a contradiction. Proceeding in a similar way we find that each of the remaining three cases in \eqref{K:S5:1} also yields a contradiction. 
Therefore a non-symmetric $K(u)$ is triangular. 

Next, note that \eqref{K:S2:B2} holds for all diagonal entries of a triangular $K(u)$. Hence the diagonal part is the same as in \eqref{KD}, i.e.,
\equ{
\sum_{1\le i \le \ell} u\,E_{ii} + \sum_{\ell< i\le r} \frac{\al-u}{\al\,u-1}\, E_{ii} + \sum_{r< i\le N} u^{-1} E_{ii} \label{K:S5:2}
}
with $\ell$, $r$ satisfying $0\le \ell < r \le N$. 

We also note that for any $i,j$ such that $i<\si(i)$ and $j<\si(j)$ the pairs $(i,\si(i))$ and $(j,\si(j))$ are non-crossing, and there are at most two simultaneously non-crossing split pairs. This follows by the same arguments as before, i.e., from \eqref{K:ij-kl}. Moreover, \eqref{K:S3:diag1} and \eqref{K:S3:diag3} also hold in the split and non-split cases, respectively. Recall that in the split case we have $\al_{i\si(i)}\al_{j\si(j)}=1$, while in the non-split case $\al_{i\si(i)}=\al_{j\si(j)}$. 

Finally, we note that constraints involving parameters $\beta_{i\si(i)}$ in (\ref{K:iki}-\ref{K:iik}) are not present in the triangular case, i.e., repeating the same analysis as we did in Step 4 does not yield any constraints on $\beta_{i\si(i)}$. In particular, if $i,j,k$ are such that $i<j<\si(j)<\si(i)$ and $k=\si(k)$, the the configurations
\[
i<k<j<\si(j)<\si(i) , \qq i<j<\si(j)<k<\si(i)
\]
are also allowed. (Recall that in the symmetric case only $k<i<j<\si(j)<\si(i)$, $i<j<k<\si(j)<\si(i)$ and $i<j<\si(j)<\si(i)<k$ were allowed.) This implies that all simultaneously non-crossing non-split pairs $(i_r,\si(i_r))$ satisfying $i_r<\si(i_r)$ for $r=1,\dots,n$, where $n$ is the number of such pairs, form a sequence 
\[
(i_{a_1}, \si(i_{a_1})), \; (i_{a_2}, \si(i_{a_2})) ,\;\dots,\; (i_{a_n}, \si(i_{a_n}))
\]
such that 
\[
i_{a_1}<i_{a_2}<\dots < i_{a_n} , \qq \si(i_{a_1})>\si(i_{a_2})>\dots > \si(i_{a_n}).
\]
Furthermore, if $i,j,k$ are such that $i<\si(i)<j<\si(j)$ and $k=\si(k)$, then the following three configurations are also allowed:
\[
k<i<\si(i)<j<\si(j), \qq i<\si(i)<k<j<\si(j), \qq i<\si(i)<j<\si(j)<k .
\]
(Recall that in the symmetric case only $i<k<\si(i)<j<\si(j)$ and $i<\si(i)<j<k<\si(j)$ were allowed.) This implies that there exist $\ell$, $m$, $r$ satisfying $0\le \ell < r \le N$ and $\ell \le m \le r$ such that for all $i,j$ satisfying $i\ne \si(i)$, $j\ne \si(j)$, 
\aln{
& 0<\si(j)\le\si(i)\le \ell && \qu\text{if}\qu \ell< i \le j \le m, \\ 
& r<\si(j)\le\si(i)\le N && \qu\text{if}\qu m < i \le j \le r .
}
In other words, $(\ell,m,r,\si)\in\Sigma^{\rm tri'}_N$.
Finally, by combining \eqref{K:S2:B1}, \eqref{K:S2:C} and \eqref{K:S5:2}, we find that, for $i\ne \si(i)$,
\[
\key_{i\si(i)}(u) = (\del_{i\le m}\,u + \del_{i>m})\frac{u-u^{-1}}{\al\,u-1}\,\beta_{i\si(i)}  
\]
with $\beta_{i\si(i)}\in\C$. Hence
\aln{
K(u) &= \sum_{1\le i \le \ell} u\,E_{ii} + \sum_{\ell< i\le r} \frac{\al-u}{\al\,u-1}\, E_{ii} + \sum_{r< i\le N} u^{-1} E_{ii} \\
& \qu + \sum_{\ell\le i \le r} (\del_{i\le m}\,u + \del_{i>m}) \frac{u-u^{-1}}{\al\,u-1}\,\Big(c_{i\si(i)} E_{i\si(i)} + c_{\si(i)i} E_{\si(i)i} \Big)
}
with $c_{i\si(i)}\in\C$ satisfying $c_{ii}=0$ and $c_{i\si(i)}c_{\si(i)i}=0$. By Lemma \ref{L:[K]}, it is equivalent to \eqref{KT}.
\end{proof}

Recall that in the proof of Theorem \ref{T:K} we used an affinization procedure to show that $K(u)$ is a solution of the reflection equation\eqref{RE}.
The equivalence relation for solutions of the constant reflection equation (i.e., an analogue of Lemma \ref{L:K-symm}) is given by a conjugation with a diagonal matrix and a multiplication by an arbitrary nonzero scalar. Taking this into account, the Type 1 solutions stated in \cite[Thm.~2]{Mu} correspond to the matrix $G$ in \eqref{KQ}, while the Type 2 solutions correspond to $Q$ in \eqref{KT:aff}.%
\footnote{Note that solutions found in \cite{Mu} are $w$-transposed with respect to the ones considered in the present paper. This is because of a difference in the definition of the constant $R$-matrix $R_q$.}
This fact combined with Theorem \ref{T:K} leads to the following important corollary.

\begin{crl} \label{C:Aff}
Every invertible solution of the reflection equation \eqref{RE} can be obtained by the affinization procedure \eqref{KS:aff} or \eqref{KT:aff} of a symmetric (Type 1) or triangular (Type 2) solution of the constant reflection equation and an application of Lemma~\ref{L:K-symm}, respectively. 
\end{crl}

\begin{rmk} [{{\it Non-invertible solutions}}]
Consider the triangular solution \eqref{KT}. Multiply by $\al$ and substitute $\veps_{i\si(i)}\to \al^{-1}\veps_{i\si(i)}$, $\veps_{\si(i)i}\to \al^{-1}\veps_{\si(i)i}$. Taking the limit $\al\to 0$ gives a set of non-invertible solutions of \eqref{RE}
\[
\sum_{m<i\le N} \Big( \veps_{i\si(i)} E_{i\si(i)} + \veps_{\si(i)i} E_{\si(i)i} \Big)
\]
with $(m,\si)\in\Sigma^{\rm tri}_N$ and $\veps_{i\si(i)},\veps_{\si(i)i}$ satisfying the same conditions as those in \eqref{KT}. 
A numerical low-rank analysis of \eqref{RE} suggests that these are indeed all non-invertible solutions, inequivalent in the sense of Lemma \ref{L:K-symm}. 
A similar limit of \eqref{KS} does not give any new solutions. \hfill \rmkend
\end{rmk}


\subsection{Twisted reflection equation}

We now turn to classification of invertible solutions of the twisted reflection equation \eqref{CtRE}.

\begin{thrm} \label{T:CK}
Let $N\ge 3$. Then $K \in\Rat(\C^N)^\times$ is a solution of the twisted reflection equation \eqref{CtRE} if and only if it is equivalent to one of the solutions given by the following matrices: 
\gat{
\sum_{1\le i \le N} E_{i\tsp\bi} + \sum_{1\le i<j\le N} \frac{1+q}{\tq + q\tsp u}\Big( (-q)^{\frac{j-i}2} u\tsp E_{i\tsp\bj} + (-q)^{\frac{i-j}2} \tq\tsp E_{j\tsp\bi} \Big) , \label{CK:qOns}
\\
\sum_{1\le i \le N} E_{i\tsp\bi} , \label{CK:A1}
\\
\sum_{1\le i \le N/2} (E_{2i-1,\overline{2i}} + E_{2i,\overline{2i}+1} ) , \label{CK:A2}
\\
\sum_{1\le i \le N/2} ( u\tsp E_{i,\bi-N/2} + E_{i+N/2,\bi} ) . \label{CK:A4}
}
The last two cases are only allowed if $N$ is even. 
\end{thrm}

We first state a technical lemma which will help us with proving Theorem \ref{T:CK}.

\begin{lemma} \label{L:[CK]} 
Let $K$ be given by \eqref{CK:qOns}, \eqref{CK:A1}, \eqref{CK:A2} or \eqref{CK:A4}. 
Then the solutions of \eqref{CtRE} equivalent to $K$ are given by
\gat{
g(u) \Bigg( \sum_{1\le i \le N} c_{\bi}^2E_{i\tsp\bi} + \sum_{1\le i<j\le N} c_{\bi}\,c_{\bj}\;\frac{1+q}{\tq\pm q\tsp u} \Big( \pm(-q)^{\frac{j-i}2} u\tsp E_{i\tsp\bj} + (-q)^{\frac{i-j}2} \tq\tsp E_{j\tsp\bi} \Big) \Bigg) \qu\text{for}\qu \eqref{CK:qOns}\!\! \label{CK:[qOns]}
\\
g(u) \sum_{i} c_{\bi} E_{i\,\bi} \qq\text{for}\qq \eqref{CK:A1}  \label{CK:[A1]}
\\
g(u)  \sum_{1\le i \le N/2} c_{i}( E_{2i-1,\overline{2i}} + E_{2i,\overline{2i}+1} ) \qq\text{or}  \label{CK:[A2a]} \\
g(u) \Big( c_1(u\,E_{11} + u^{-1} E_{NN}) + \sum_{1< i \le N/2} c_{i}( E_{2i-1,\overline{2i}+2} + E_{2i-2,\overline{2i}+1} ) \Big) \qq\text{for}\qq \eqref{CK:A2} \label{CK:[A2b]} 
\\
g(u) \sum_{1\le i \le N/2} c_i( \pm u\,E_{i,\bi-N/2} + E_{i+N/2,\bi} )  \qq\text{for}\qq \eqref{CK:A4} \label{CK:[A4]}
}
for some $g \in\Rat^\times$ and $c_i\in\C^\times$. 
\end{lemma}

\begin{proof}
The proof is analogous to that of Lemma \ref{L:[K]}, i.e., the statements above follow by tedious computations using Lemma~\ref{L:K-symm}. In particular, for $1\le k<N$,
\aln{
g(u)\,D^w \big((Z^\rho(\tfrac{\eta}{u})^w)\big)^k K(\pm u) (Z^\rho(\eta u))^k D &= K'(u) , \qq c_i = \eta^{\del_{i\le k}} d_i , 
}
where $K(u)$ is the matrix \eqref{CK:qOns}, \eqref{CK:A1} or \eqref{CK:A4} and $K'(u)$ is the matrix \eqref{CK:[qOns]}, \eqref{CK:[A1]} or \eqref{CK:[A4]}, respectively. 
The coefficients $c_i$ are given by $d_i \eta^{\del_{i\le k}}$, $d_i \eta^{\del_{i\le k}}$, $d_i^2 \eta^{2\del_{i\le k}}$ and $d_{N/2-i+1} d_{\bi} \eta^{\del_{N/2-i< k}}$, accordingly. 
In case of \eqref{CK:A2} we have instead
\[
g(u)\,D^w \big((Z^\rho(\tfrac{\eta}{u})^w)\big)^k K(\pm u) (Z^\rho(\eta u))^k D = 
\begin{cases} 
K'(u)  & \text{if $k$ is even,} \\
K''(u) & \text{if $k$ is odd,}
\end{cases}
\]
where $K(u)$ is \eqref{CK:A2}, $K'(u)$ is \eqref{CK:[A2a]} with $c_i = \eta^{2\del_{\overline{2i}\le k}} d_{\overline{2i}} d_{\overline{2i}+1}$ and $K''(u)$ is \eqref{CK:[A2b]} with $c_1 = \eta\,d_1 d_N$ and $c_i = \eta^{2\del_{\overline{2i}+2\le k}} d_{\overline{2i}+1} d_{\overline{2i}+2}$. 
Lastly, all matrices \eqrefs{CK:[qOns]}{CK:[A4]} are invariant under the transformation given by Lemma \ref{L:K-symm} (iv).
\end{proof}

\begin{proof}[Proof of Theorem \ref{T:CK}] 

($\Longrightarrow$) 
We start by showing that matrices \eqrefs{CK:qOns}{CK:A4} are solutions to \eqref{CtRE}. 
The matrices \eqrefs{CK:qOns}{CK:A2} in the form given by $\wt{K}(u)$, cf.\ \eqref{CK}, are known to be solutions of \eqref{tRE}.
For \eqref{CK:qOns} this was shown in \cite{Gan}. 
For \eqref{CK:A1} this is accounted for by the fact that R-matrices $R(\tfrac uv)$ and $R^{t_1}(\tfrac{1}{uv})$ commute. 
For \eqref{CK:A2} this was shown in \cite{MRS} (in the proof of Theorem 3.10). 
For \eqref{CK:A4} this was stated without proof in \cite{RV}. 
Hence we only need to show that \eqref{CK:A4} is indeed a solution of \eqref{CtRE}. 
By Lemmas \ref{L:CRE} and \ref{L:K-symm}  it is sufficient to show that the matrix
\[
\wt{K}(u) = \sum_{1\le i \le N/2} (u E_{i+N/2,i} + E_{i,i+N/2})
\]
is a solution of \eqref{tRE}. Let $T$ be the unique element of $\GL(\C^{N} \ot \C^{N})$ that sends the ordered basis 
\[
(e_1 \ot e_1 , \dots , e_1 \ot e_{N} , e_2 \ot e_1 , \dots , e_2 \ot e_{N} ,  \dots , e_{N} \ot e_1, \dots , e_{N} \ot e_{N} )
\] 
to the ordered basis 
\begin{align*}
& (e_1 \ot e_1 , \dots , e_1 \ot e_{N/2} , e_2 \ot e_1 , \dots , e_2 \ot e_{N/2},  \dots , e_{N} \ot e_1, \dots , e_{N} \ot e_{{N/2}},  \\
& \qq e_1 \ot e_{{N/2}+1} , \dots , e_1 \ot e_{N} , e_2 \ot e_{{N/2}+1} , \dots , e_2 \ot e_{N},  \dots , e_{N} \ot e_{{N/2}+1}, \dots , e_{N} \ot e_{N}).
\end{align*}
We conjugate the matrices $R(u)$, $R^{t_1}(u)$, $\wt{K}_1(u)$ and $\wt{K}_2(v)$ with the matrix $T$.  
This gives 
\gan{ 
\setlength{\arraycolsep}{2pt}
\renewcommand{\arraystretch}{0.9}
T R(\tfrac uv)\,T^{-1} = \left(\begin{array}{cccc} R(\tfrac uv) & \\ & (1-\tfrac uv)\,f_q(\tfrac uv)\,I & (q-q^{-1})\,\tfrac uv\,f_q(\tfrac uv)\,P\\[.25em] & (q-q^{-1})\,f_q(\tfrac uv)\,P & (1-\tfrac uv)\,f_q(\tfrac uv)\,I \\ &&&  R(u) \end{array}\right) , 
\\[.25em]
\setlength{\arraycolsep}{0pt}
\renewcommand{\arraystretch}{0.9}
TR^{t_1}(\tfrac1{uv})\,T^{-1} = \left(\begin{array}{cccc} R^{t_1}(\tfrac1{uv}) & & & (q-q^{-1})\, f_q(\tfrac1{uv})\,P^{t_1} \\ & (1-\tfrac1{uv})\,f_q(\tfrac1{uv})\,I \\ &  & (1-\tfrac1{uv})\,f_q(\tfrac1{uv})\,I \\ (q-q^{-1})\,\tfrac1{uv}\, f_q(\tfrac1{uv})\,P^{t_1} &&& R^{t_1}(\tfrac1{uv}) \end{array}\right) 
}
and
\gan{ 
\setlength{\arraycolsep}{3pt}
\renewcommand{\arraystretch}{0.9}
T \wt{K}_1(u)\, T^{-1} = \left(\begin{array}{cccc} & & I \\ & && I \\ \!u\,I \\ & \!u\,I  \end{array}\right), \qq
T \wt{K}_2(v)\, T^{-1} = \left(\begin{array}{cccc} & I \\ \!v\,I \\ &&& I \\ && \!v\,I  \end{array}\right) ,
}
where the operators inside the matrices are each acting on a copy of $\C^{N/2} \ot \C^{N/2}$. It remains to conjugate both sides of \eqref{tRE} with $T$ and use the expressions above together with the identities $P P^{t_1} = P^{t_1} P = P^{t_1}$ and $P\,R^{t_1}(u)\,P = R^{t_1}(u)$. Then a straightforward computation shows that both sides of \eqref{tRE} agree. Hence \eqref{CK:A4} is indeed a solution of~\eqref{CtRE}.

($\Longleftarrow$) 
Conversely, suppose that $K(u) = \sum_{ij} \key_{ij}(u)\,e_{ij}$ with $\key_{ij} \in \Rat$. We write the twisted reflection equation \eqref{CtRE} in the braided form,
\equ{ 
\check{R}(\tfrac{u}{v})\, K_2(u)\, \check{R}^{\vee}(uv) \,K_2(v) = K_2(v)\, \check{R}^{\vee}(uv)\, K_2(u)\, \check{R}(\tfrac{u}{v}) \label{CtRE1} ,
}
where $\check{R}^{\vee}(u):=P R^{\vee}(u)$.
Recall \eqref{RC(u)-aff} and \eqref{RCq}. Observe that
\gan{
\check{R}^{\vee}(uv) (e_k \ot e_l) = c_{kl}(uv)\, e_l \ot e_k + \del_{k\bl} \sum_{1\le r\le N} d_{kr}(uv)\, e_{\bar r} \ot e_r , \\
(e^*_k \ot e^*_l)\,\check{R}^{\vee}(uv) = c_{kl}(uv)\, e^*_l \ot e^*_k + \del_{k\bl} \sum_{1\le r\le N} d_{kr}(uv)\, e^*_{\bar r} \ot e^*_r , 
}
where 
\eqg{
c_{kl}(uv) = q^{\del_{k\tsp\bl}} f_q\big(\tfrac{\tq^{\,2}}{uv}\big) + q^{-\del_{k\tsp\bl}} f_{q^{-1}}\big(\tfrac{uv}{\tq^{\,2}}\big), 
\\
d_{kr}(uv) = (q-q^{-1})\,(-q)^{r-k} \,\Big( \del_{k<r}\, f_q\big(\tfrac{\tq^{\,2}}{uv}\big) - \del_{k>r}\, f_{q^{-1}}\big(\tfrac{uv}{\tq^{\,2}}\big) \Big) . \label{cd}
}
In particular, $c_{k\tsp\bk}(uv)=1$ and $d_{kk}(uv)=0$. Repeating the same steps as in the proof of Theorem~\ref{T:K} we find
\aln{
& (e^*_i \ot e^*_j)\, \check{R}(\tfrac{u}{v})\, K_2(u)\, \check{R}^{\vee}(uv)\, K_2(v) \,(e_k \ot e_l) 
\\
& \qq = \sum_{1\le r,s\le N} \Big( a_{ij}(\tfrac uv)\,\key_{ir}(u)\, e^*_j \ot e^*_r + b_{ij}(\tfrac uv)\, \key_{jr}(u)\, e^*_i \ot e^*_r \Big) \el[-.25em] & \hspace{2.75cm} \times \Big( c_{ks}(uv)\, \key_{sl}(v)\, e_s \ot e_k + \del_{k\bar s}\sum_{1\le t\le N} d_{kt}(uv)\, \key_{sl}(v)\, e_{\bar t} \ot e_t \Big) 
\\
& \qq = a_{ij}(\tfrac uv)\Big( d_{k\bj}(uv)\, \key_{i\tsp\bj}(u)\,  \key_{\bk l}(v) + c_{kj}(uv)\, \key_{ik}(u)\,\key_{jl}(v) \Big) \el
& \qq\qu  + b_{ij}(\tfrac uv)\Big( d_{k\bi}(uv)\, \key_{j\tsp\bi}(u)\, \key_{\bk l}(v) + c_{ki}(uv)\, \key_{jk}(u)\,  \key_{il}(v) \Big)
\intertext{and}
& (e^*_i \ot e^*_j)\, K_2(v)\, \check{R}^{\vee}(uv)\, K_2(u)\, \check{R}(\tfrac{u}{v}) \,(e_k \ot e_l) 
\\
& \qq= \sum_{1\le r,s\le N} \Big( \key_{jr}(v)\, c_{ir}(uv)\, e^*_r \ot e^*_i + \del_{i\bar r} \sum_{1\le t\le N} \key_{jr}(v)\, d_{it}(uv)\, e^*_{\bar t} \ot e^*_t \Big) \el[-.25em] & \hspace{2.75cm} \times \Big( \key_{sk}(u)\, a_{kl}(\tfrac uv)\, e_l \ot e_s + \key_{sl}(u)\,b_{kl}(\tfrac uv)\, e_k \ot e_s \Big) 
\\
& \qq = a_{kl}(\tfrac uv)\Big( d_{i\bl}(uv)\,\key_{j\tsp\bi}(v)\,\key_{\bl k}(u) + c_{il}(uv)\,\key_{jl}(v)\,\key_{ik}(u) \Big) \\
& \qq\qu + b_{kl}(\tfrac uv)\Big( d_{i\bk}(uv)\,\key_{j\tsp\bi}(v)\,\key_{\bk l}(u) + c_{ik}(uv)\,\key_{jk}(v)\,\key_{il}(u) \Big). 
}
Taking the difference of the expressions above we obtain
\eqa{
\label{CK1}
& \Big( a_{ij}(\tfrac uv)\,c_{kj}(uv) - a_{kl}(\tfrac uv)\,c_{il}(uv)\Big) \, \key_{ik}(u)\,\key_{jl}(v)\\
& + \Big( a_{ij}(\tfrac uv)\, d_{k\bj}(uv)\, \key_{i\tsp\bj}(u) + b_{ij}(\tfrac uv)\, d_{k\bi}(uv)\, \key_{j\tsp\bi}(u)\Big)\, \key_{\bk l}(v) \\
& - \Big( a_{kl}(\tfrac uv)\, d_{i\bl}(uv)\,\key_{\bl k}(u) + b_{kl}(\tfrac uv)\, d_{i\bk}(uv)\,\key_{\bk l}(u) \Big) \,\key_{j\tsp\bi}(v) \\
& + b_{ij}(\tfrac uv)\, c_{ki}(uv)\, \key_{jk}(u)\,  \key_{il}(v) - b_{kl}(\tfrac uv)\, c_{ik}(uv)\,\key_{jk}(v)\,\key_{il}(u) = 0.
}

Next, in four steps, we prove that any invertible solution $K(u)$ of \eqref{CtRE1} is equivalent to \eqref{CK:qOns}, \eqref{CK:A1}, \eqref{CK:A2} or \eqref{CK:A4}. For any $i\ne j$ we will write $a_{ij}(u)=a(u)$ and for any $i\ne \bj$ we will write $c_{ij}(u)=c(u)$.


{\noindent \it Step 1: Without loss of generality, for any $i$ we have
\equ{
\key_{i\tsp\bi}(u) = 1 \tx{or} \key_{i\tsp\bi}(u) = 0. \label{CK:S1:1}
}
For any $i,j$ such that $i<\bj$ we have $\key_{ij}(u)=\key_{\bj\tsp\bi}(u)=0$ or both are nonzero and
\equ{
\key_{\bj\tsp\bi}(u) = -(-q)^{i-\bj}\,\tq\,u^{-1}\,\frac{\tq\,c+q\,u}{\tq + q\,u\,c}\,\key_{ij}(u) \label{CK:S1:2}
}
with $c\in\C$. Moreover, when $\key_{i\tsp\bi}(u)=1$, we have
\equ{
\key_{ij}(u)=(1+q) \frac{c_{ij}}{q+\tq\,u^{-1}}, \qq \key_{\bj\tsp\bi}(u)= (-q)^{i-\bj}\,\tq\,u^{-1}\key_{ij}(u), \label{CK:S1:3}
}
or, if $i>\bj$ instead, 
\equ{
\key_{\bj\tsp\bi}(u) = (1+q) \frac{c_{\bj\tsp\bi}}{q+\tq\,u^{-1}}, \qq \key_{ij}(u)= (-q)^{\bj-i}\,\tq\,u^{-1}\key_{\bj\tsp\bi}(u), \label{CK:S1:4}
}
where $c_{ij},c_{\bj\tsp\bi}\in\C^\times$.
}


Setting $i=j$ and $k=l$ in \eqref{CK1} we obtain
\[
d_{k\bi}(u v)\, \key_{i\tsp\bi}(u)\,\key_{\bk k}(v) - d_{i\bk}(u v)\,\key_{i\tsp\bi}(v)\,\key_{\bk k}(u) = 0 ,
\]
which, by \eqref{cd}, is equivalent to $\key_{i\tsp\bi}(u) = c_{ik} \key_{k\bk}(u)$ for some $c_{ik} \in \C$. Then, by Lemma \ref{L:K-symm} (iii), we may assume \eqref{CK:S1:1}. 


Next, setting $\bi = k$ and $\bj = l$ in \eqref{CK1} we obtain
\[
\Big(a(\tfrac{u}{v})\,d_{\bi\tsp\bj}(u v)\,\key_{i\tsp\bj}(u) + b_{ij}(\tfrac{u}{v}) \,\key_{j\tsp\bi}(u) \Big) \key_{i\tsp\bj}(v) - \Big(a(\tfrac{u}{v})\,d_{ij}(u v)\,\key_{j\tsp\bi}(u) + b_{\bi\tsp\bj}(\tfrac{u}{v})\,\key_{i\tsp\bj}(u) \Big)\,\key_{j\tsp\bi}(v) = 0. 
\]
Let $\key_{i\tsp\bj}(u)=0$. Then the equality above gives $a(\tfrac uv)\,d_{ij}(uv)\,\key_{j\tsp\bi}(u)\,\key_{j\tsp\bi}(v)=0$. Hence $\key_{i\tsp\bj}(u)=0$ implies $\key_{j\tsp\bi}(u)=0$. Now suppose that both $\key_{i\tsp\bj}(u)$ and $\key_{j\tsp\bi}(u)$ are nonzero instead. Assuming further that $i<j$ the equality above gives
\aln{
& (-q)^{i+j} (q^2 u v-\tq^2) \Big(u\,\key_{j\tsp\bi}(u)\,\key_{i\tsp\bj}(v) - v\,\key_{i\tsp\bj}(u)\,\key_{j\tsp\bi}(v) \Big) \\ 
& \qq + (u-v)\Big( (-q)^{2 i+1} \tq^2\,\key_{i\tsp\bj}(u)\,\key_{i\tsp\bj}(v) - (-q)^{2 j+1} u v \,\key_{j\tsp\bi}(u)\,\key_{j\tsp\bi}(v)\Big) = 0 ,
}
which can be solved using separation of variables. In particular, upon subsituting $j\to \bj$, we obtain the required relation \eqref{CK:S1:2}.


Finally, setting $i=j=\bk$ in \eqref{CK1} and assuming $\key_{i\tsp\bi}(u)=1$ (by \eqref{CK:S1:1}), we find
\equ{
\big(a(\tfrac{u}{v})\,c(u v)-1\big) \,\key_{il}(v) + b_{\bi l}(\tfrac{u}{v})\,\key_{il}(u) + a(\tfrac{u}{v})\,d_{i\bl}(u v)\, \key_{\bl\bi}(u) = 0. \label{CK:S1:6}
}
Using \eqref{ab}, \eqref{cd} and setting $v=\tq\,q^{-1}$ we obtain
\[
(\del_{i<\bl}\,u + \del_{i>\bl}\,q\,\tq)\,\key_{\bl \bi}(u) - (-q)^{i-\bl} \,( \del_{i<\bl}\,q\,u + \del_{i>\bl}\,\tq)\,\key_{il}(u) = 0 
\]
which we rewrite as
\equ{
\key_{\bj\tsp\bi}(u) = (-q)^{i-\bj}\, (\del_{i<\bj}\, \tq \,u^{-1} + \del_{i>\bj}\, \tq^{-1} u )\, \key_{ij}(u) . \label{CK:S1:7}
}
(This corresponds to the $c=1$ case in \eqref{CK:S1:2}.) We now substitute \eqref{CK:S1:7} back to \eqref{CK:S1:6} giving
\aln{
& q (-q)^{\bl-i}  (u-v) (\tq^2\,\del_{i>\bl}+u v\,\del_{i<\bl})\, \key_{\bl\bi}(u)  \\ & \qq + (\tq^2-q^2 u v)\,(u\,\del_{i>\bl} + v\,\del_{i<\bl})\,\key_{il}(u) - u\,(\tq^2-q^2 v^2)\,\key_{il}(v) = 0.
}
Substituting $l\to j$ and separating variables we find \eqref{CK:S1:3} (when $i<\bj$) and \eqref{CK:S1:4} (when $i>\bj$), as required.


{\noindent \it Step 2: Suppose that $\key_{ij}(u)\ne 0$ for all $i,j$. Then $K$ is equivalent to \eqref{CK:qOns}.
}


By \eqref{CK:S1:1} we can assume that $\key_{i\tsp\bi}(u)=1$ for all $i$. 
Then \eqref{CK:S1:3} allows us to find all the remaining matrix entries up to the unknown constants $c_{ij}$ with $i+j\le N$. 
Set $i=l$ and $j=\bk$ in \eqref{CK1}. Assuming $i\ne \bi$ we obtain
\aln{
& a(\tfrac{u}{v}) \Big((c(u v)-1)\, \key_{ji}(v)\, \key_{i \bj}(u)+d_{i\,\bi}(u v)\, \key_{\bi\tsp\bj}(u)\, \key_{j\,\bi}(v) \Big) + b_{\bj i}(\tfrac{u}{v})\, d_{ij}(u v) \,\key_{ji}(u) \,\key_{j\bi}(v)
\\
& \qu + c(u v) \Big( b_{\bj i}(\tfrac{u}{v}) \,\key_{ii}(u)\,\key_{j\tsp\bj}(v) - b_{ij}(\tfrac{u}{v}) \, \key_{ii}(v)\,\key_{j\tsp\bj}(u) \Big) - b_{ij}(\tfrac{u}{v})\, d_{\bj\bi}(u v)\, \key_{ji}(v)\, \key_{j\bi}(u) = 0.
}
Assuming further that $i<j$ and $i+j\le N$ and using \eqref{ab} and \eqref{cd} we find
\aln{
\frac{(-q)^{-2 i}}{(u-q^2 v) (\tq^2-q^2 u v)} & \bigg( (-q)^{2 i+1} (q^2-1) (\tq^2-u v) (u\,\key_{ii}(v) - v \,\key_{ii}(u)) \\
& \qu +  (-q)^{i+j} (q^2-1)^2\, u v \,( u\,\key_{ji}(v)\,\key_{j\,\bi}(u) - v\,\key_{ji}(u)\,\key_{j\,\bi}(v) ) \\
& \qu + (q-1)(u-v) \Big( (-q)^{2 i+1} (\tq^2+q\,u v)\,\key_{ji}(v)\,\key_{i\,\bj}(u) \\ & \qq\qq\qq\qq\qq - \tq^2 q^2 (q+1)\,u v\,\key_{\bi\,\bj}(u)\,\key_{j\,\bi}(v)\Big) \bigg) = 0 .
}
We now need to substitute \eqref{CK:S1:3}. 
The resulting equality is 
\[
\frac{q^2 (q+1) (q^2-1) u v (u-v) (u v-\tq^2) (c_{ii} - c_{i\bj} \,c_{ji} )}{(q u+\tq) (q v+\tq) (u-q^2 v) (\tq^2-q^2 u v)} = 0.
\]
This is only true if $c_{ii} = c_{i\tsp\bj}\,c_{ji}$. Now set $i=\bl = k$ in \eqref{CK1}. 
Assuming $i\ne \bi$ we obtain
\aln{
& a(\tfrac{u}{v}) \Big( (c(u v)-1)\, \key_{ii}(u)\, \key_{j\bi}(v) + d_{i\tsp\bj}(u v)\,\key_{\bi\tsp\bi}(v)\, \key_{i\tsp\bj}(u) \Big) + b_{ij}(\tfrac{u}{v})\, d_{i\tsp\bi}(u v)\,\key_{\bi\tsp\bi}(v) \,\key_{j\bi}(u) \\
& \qu + c(u v) \Big(b_{ij}(\tfrac{u}{v})\,\key_{ji}(u)\, \key_{i\tsp\bi}(v) - b_{i\tsp\bi}(\tfrac{u}{v})\,\key_{ji}(v)\, \key_{i\tsp\bi}(u) \Big) - b_{i\tsp\bi}(\tfrac{u}{v})\, d_{i\tsp\bi}(u v)\,\key_{\bi\tsp\bi}(u) \, \key_{j\bi}(v) = 0.
}
Assuming further that $i<j$ and $i+j\le N$ and using \eqref{ab} and \eqref{cd} we find
\aln{
\frac{1}{(u-q^2 v) (\tq^2-q^2 u v)} & \bigg( q\,(q^2-1)\,u\,(\tq^2-u v)\,(\key_{ji}(v) - \key_{ji}(u)) \\ & \qu + (-q)^{\overline{2i}} (q^2-1)^2 u^2 v \,(\key_{\bi\tsp\bi}(v)\, \key_{j\tsp\bi}(u) - \key_{\bi\tsp\bi}(u)\,\key_{j\tsp\bi}(v)) \\ & \qu + q\,(q-1) (u-v) \Big((\tq^2+q u v)\,\key_{ii}(u)\,\key_{j\tsp\bi}(v) \\ 
& \qq\qq\qq\qq\qq - (-q)^{\bi-j} (q+1)\,u v\,\key_{\bi\tsp\bi}(v)\,\key_{i\tsp\bj}(u) \Big) \bigg) = 0 .
}
Substituting \eqref{CK:S1:3} yields
\[
\frac{\tq\,q\,(q-1)\,(q+1)^2 u\,(u-v) (\tq^2-u v) ((-q)^{i-j} c_{ii}\,c_{i\tsp\bj}-c_{ji})}{(\tq+q u) (\tq+q v) (u-q^2 v) (\tq^2-q^2 u v)} = 0
\]
and so $c_{ji} = (-q)^{i-j} c_{ii}\, c_{i\tsp\bj}$. Upon combing with the identity $c_{ii} = c_{i\tsp\bj}\,c_{ji}$ we find $c_{ij}^2 = (-q)^{\bj-i}$.
We will use Lemma \ref{L:K-symm} (iii) to narrow down the choice of the sign for $c_{ij}$. 
Conjugation by the symmetry $Z =\sum_{1\le i \le N} z_i\,E_{ii}$ maps each $\key_{ij}(u)$ to $z_{\bi} z_j \key_{ij}(u)$. 
Assume that $z_i \in \{ \pm 1 \}$. 
Then all the anti-diagonal entries $\key_{i\,\bi}(u)$ are mapped to themselves, hence the assumption that $\key_{i\,\bi}(u)=1$ remains valid.
Choose $z_{N}=1$ so that each $\key_{1i}(u)$ is mapped to $z_i\key_{1i}(u)$. 
This allows us to choose the square root $c_{1j}=(-q)^{\frac{\bj-1}2}$ for $1\le j<N$. 

Next, set $k=l$ in \eqref{CK1}. 
Assuming $i\ne k \ne j$ and $i\ne j$ we find
\aln{
& c(u v) \Big((a(\tfrac{u}{v})-1)\,\key_{ik}(u)\, \key_{jk}(v) + b_{ij}(\tfrac{u}{v})\,\key_{ik}(v)\,\key_{jk}(u) \Big) \\
& \qq + a(\tfrac{u}{v})\,d_{k\bj}(u v)\,\key_{i\bj}(u) + b_{ij}(\tfrac{u}{v})\,d_{k\bi}(u v)\,\key_{j\bi}(u) - d_{i\bk}(u v)\,\key_{j\bi}(v) = 0. 
}
Assume further that $i<k<j$, $i+k\le N$ and $i+j>N+1$, so that $2\le k< N$. Then, using \eqref{ab} and \eqref{cd}, the equality above becomes 
\aln{
&\frac{1}{(u-q^2v)(\tq^2-q^2 u v)} \Bigg(  \frac{(u+q v)\,\key_{ik}(u)\,\key_{jk}(v) - (q+1)\,u\, \key_{ik}(v)\, \key_{jk}(u)}{q (q-1)(\tq^2-u v)}  \\
& \qq + \frac{ (-q)^{i-j+1}\tq^2 (u-v)\,\key_{i\bj}(u) + u v \left((u-q^2 v)\,\key_{j\bi}(v) + (q^2-1)\,u\,\key_{j\bi}(u)\right) }{(-q)^{i-\bk}(q^2-1)}\Bigg) = 0 .
}
Upon substituting \eqref{CK:S1:3} we obtain the relation $c_{i\bj} = c_{ik}\,c_{\bk\bj}$.
Setting $i=1$ and using the expression for $c_{1j}$ we find $c_{\bk\bj}=(-q)^{\frac{j-\bk}{2}}$, or equivalently $c_{ji}=(-q)^{\frac{\bi-j}{2}}$ for $1\le i<j<N$. Now set $j=N$ in $c_{i\bj} = c_{ik}\,c_{\bk\bj}$ and use the expression for $c_{j1}$. This gives $c_{ik}=(-q)^{\frac{\bk-i}{2}}$. It remains to find the expressions for coefficients $c_{ii} = c_{i\bj}\,c_{ji}= (-q)^{\frac{\bi-i}2}$. By combining the above results we find that $c_{ij}=(-q)^{\frac{\bj-i}2}$ for $i+j\le N$. Substituting this back to \eqref{CK:S1:3} and using \eqref{CK:S1:1}, namely $\key_{i\,\bi}(u)=1$, we obtain the solution \eqref{CK:qOns}, as required.


{\noindent \it Step 3: Suppose there exists $i,j$ such that $i\ne \bj$, $\key_{ij}(u)= 0$ and $\key_{i\tsp\bi}(u)=1$ or $\key_{\bj\tsp j}(u)=1$. Then $K(u)$ is given by \eqref{CK:A1}.
}


We know from Step 1 that $\key_{ij}(u)= 0$ implies $\key_{\bj\tsp\bi}(u)= 0$. With this condition in mind, i.e., when $\key_{i\tsp\bj}(u)=\key_{j\tsp\bi}(u)=0$ or $\key_{\bl k}(u)=\key_{\bk l}(u)=0$, \eqref{CK1} gives
\aln{
& b_{ij}(\tfrac uv)\,\key_{il}(v)\,\key_{jk}(u) - b_{kl}(\tfrac uv)\,\key_{il}(u)\,\key_{jk}(v) \\
& \qq = (u\,\del_{k<l}+v\,\del_{k>l})\,\key_{il}(u)\,\key_{jk}(v) - (u\,\del_{i<j} + v\,\del_{i>j})\,\key_{il}(v)\,\key_{jk}(u) = 0.
}
Note that the equality above also holds if $l=\bi$ or $k=\bj$. (The case when $l=\bi$ and $k=\bj$ holds trivially by \eqref{CK:S1:1}.) We first focus on the case when $\key_{i\tsp\bj}(u)=\key_{j\tsp\bi}(u)=0$ and $i<j$. By separating the variables we find
\equ{
\frac{\key_{il}(u)}{\key_{jk}(u)} \in (\del_{k<l} + u\,\del_{k>l})\,\C . \label{CK:S3:1a}
}
which can be represented by 
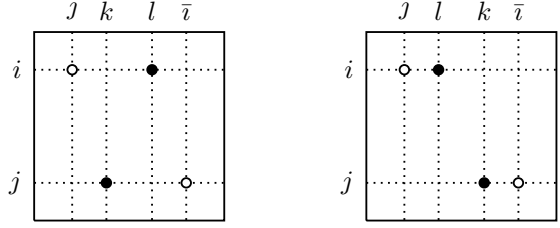
\[
\btp
\draw (0,0) -- (10,0) -- (10,10) -- (0,10) -- (0,0);
\draw[dotted] (0,8) node[left=1pt]{$i$} -- (10,8) ;
\draw[dotted] (8,10) node[above=1pt]{$\bi$} -- (8,0);
\draw[dotted] (0,2) node[left=1pt]{$j$} -- (10,2) ;
\draw[dotted] (2,10) node[above=1pt]{$\bj$} -- (2,0);
\draw[dotted] (3.8,10) node[above=1pt]{$k$} -- (3.8,0);
\draw[dotted] (6.2,10) node[above=1pt]{$l$} -- (6.2,0) ;
\filldraw[fill=white] (2,8) circle (.25);
\filldraw[fill=white] (8,2) circle (.25);
\filldraw[fill=black] (6.2,8) circle (.25);
\filldraw[fill=black] (3.8,2) circle (.25);
\etp
\qq\qq
\btp
\draw (0,0) -- (10,0) -- (10,10) -- (0,10) -- (0,0);
\draw[dotted] (0,8) node[left=1pt]{$i$} -- (10,8) ;
\draw[dotted] (8,10) node[above=1pt]{$\bi$} -- (8,0);
\draw[dotted] (0,2) node[left=1pt]{$j$} -- (10,2) ;
\draw[dotted] (2,10) node[above=1pt]{$\bj$} -- (2,0);
\draw[dotted] (3.8,10) node[above=1pt]{$l$} -- (3.8,0);
\draw[dotted] (6.2,10) node[above=1pt]{$k$} -- (6.2,0) ;
\filldraw[fill=white] (2,8) circle (.25);
\filldraw[fill=white] (8,2) circle (.25);
\filldraw[fill=black] (6.2,2) circle (.25);
\filldraw[fill=black] (3.8,8) circle (.25);
\etp
\]
Here the empty circles represent matrix entries that are equal to zero, i.e., $\key_{i\bj}(u)=\key_{j\bi}(u)=0$. 
The filled circles represent the nonzero entries in question, $\key_{il}(u)$ and $\key_{jk}(u)$. 
In the cases $l=\bi$ or $k=\bj$ the corresponding vertical dotted lines should be identified.
 
In the case when $\key_{\bl k}(u)=\key_{\bk l}(u)=0$ and $k<l$ we find instead
\equ{
\frac{\key_{il}(u)}{\key_{jk}(u)} \in (\del_{i<j} + u^{-1}\del_{i>j} )\,\C, \label{CK:S3:1b}
}
which we represent by
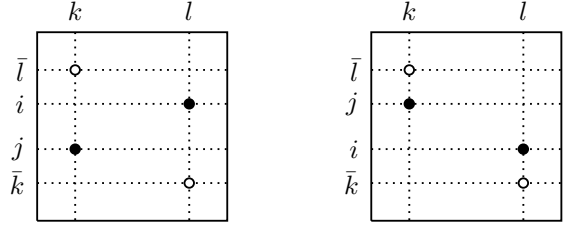
\[
\btp
\draw (0,0) -- (10,0) -- (10,10) -- (0,10) -- (0,0);
\draw[dotted] (0,8) node[left=1pt]{$\bl$} -- (10,8) ;
\draw[dotted] (8,10) node[above=1pt]{$l$} -- (8,0);
\draw[dotted] (0,2) node[left=1pt]{$\bk$} -- (10,2) ;
\draw[dotted] (2,10) node[above=1pt]{$k$} -- (2,0);
\draw[dotted] (0,6.2) node[left=1pt]{$i$} -- (10,6.2) ;
\draw[dotted] (0,3.8) node[left=1pt]{$j$} -- (10,3.8);
\filldraw[fill=white] (2,8) circle (.25);
\filldraw[fill=white] (8,2) circle (.25);
\filldraw[fill=black] (8,6.2) circle (.25);
\filldraw[fill=black] (2,3.8) circle (.25);
\etp
\qq\qq
\btp
\draw (0,0) -- (10,0) -- (10,10) -- (0,10) -- (0,0);
\draw[dotted] (0,8) node[left=1pt]{$\bl$} -- (10,8) ;
\draw[dotted] (8,10) node[above=1pt]{$l$} -- (8,0);
\draw[dotted] (0,2) node[left=1pt]{$\bk$} -- (10,2) ;
\draw[dotted] (2,10) node[above=1pt]{$k$} -- (2,0);
\draw[dotted] (0,6.2) node[left=1pt]{$j$} -- (10,6.2) ;
\draw[dotted] (0,3.8) node[left=1pt]{$i$} -- (10,3.8);
\filldraw[fill=white] (2,8) circle (.25);
\filldraw[fill=white] (8,2) circle (.25);
\filldraw[fill=black] (2,6.2) circle (.25);
\filldraw[fill=black] (8,3.8) circle (.25);
\etp
\]
Similarly as before, when $\bl=i$ or $\bk=j$ the corresponding horizontal dotted lines should be identified.


Next, setting $i=j$ and $\key_{ik}(u)=0$ in \eqref{CK1} and assuming $\key_{i\tsp\bi}(u)=1$ we obtain
\[
a(\tfrac{u}{v})\,d_{i\bl}(u v)\,\key_{\bl k}(u) + b_{kl}(\tfrac{u}{v})\,d_{i\bk}(u v)\,\key_{\bk l}(u) - d_{k\bi}(u v)\,\key_{\bk l}(v) = 0.
\]
Assume further that $i<\bk$. Then the equality above gives
\ali{
& (-q)^{k+1} (u-v) \big(\tq^2 \del_{i>\bl} + u v\, \del_{i<\bl}\big) \,\key_{\bl k}(u) \el 
& \qq + (-q)^l \,u v \Big((q^2-1)\,(u\,\del_{k<l} + v\,\del_{k>l})\,\key_{\bk l}(u) + (u-q^2 v)\,\key_{\bk l}(v)\Big) = 0 . \label{CK:S3:2}
}
There are three cases we need to consider: $l<k$, $k<l<\bi$ and $\bi<l$. They are shown diagrammatically below
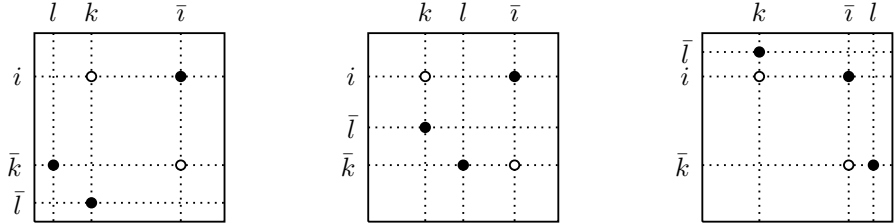
\[
\btp
\draw (0,0) -- (10,0) -- (10,10) -- (0,10) -- (0,0);
\draw[dotted] (0,7.7) node[left=1pt]{$i$} -- (10,7.7) ;
\draw[dotted] (7.7,10) node[above=1pt]{$\bi$} -- (7.7,0);
\filldraw[fill=black] (7.7,7.7) circle (.25);
\draw[dotted] (0,3) node[left=1pt]{$\bk$} -- (10,3) ;
\draw[dotted] (3,10) node[above=1pt]{$k$} -- (3,0);
\filldraw[fill=white] (3,7.7) circle (.25);
\filldraw[fill=white] (7.7,3) circle (.25);
\draw[dotted] (0,1) node[left=1pt]{$\bl$} -- (10,1) ;
\draw[dotted] (1,10) node[above=1pt]{$l$} -- (1,0);
\filldraw[fill=black] (1,3) circle (.25);
\filldraw[fill=black] (3,1) circle (.25);
\etp
\qq\qq
\btp
\draw (0,0) -- (10,0) -- (10,10) -- (0,10) -- (0,0);
\draw[dotted] (0,7.7) node[left=1pt]{$i$} -- (10,7.7) ;
\draw[dotted] (7.7,10) node[above=1pt]{$\bi$} -- (7.7,0);
\filldraw[fill=black] (7.7,7.7) circle (.25);
\draw[dotted] (0,3) node[left=1pt]{$\bk$} -- (10,3) ;
\draw[dotted] (3,10) node[above=1pt]{$k$} -- (3,0);
\filldraw[fill=white] (3,7.7) circle (.25);
\filldraw[fill=white] (7.7,3) circle (.25);
\draw[dotted] (0,5) node[left=1pt]{$\bl$} -- (10,5) ;
\draw[dotted] (5,10) node[above=1pt]{$l$} -- (5,0);
\filldraw[fill=black] (5,3) circle (.25);
\filldraw[fill=black] (3,5) circle (.25);
\etp
\qq\qq
\btp
\draw (0,0) -- (10,0) -- (10,10) -- (0,10) -- (0,0);
\draw[dotted] (0,7.7) node[left=1pt]{$i$} -- (10,7.7) ;
\draw[dotted] (7.7,10) node[above=1pt]{$\bi$} -- (7.7,0);
\filldraw[fill=black] (7.7,7.7) circle (.25);
\draw[dotted] (0,3) node[left=1pt]{$\bk$} -- (10,3) ;
\draw[dotted] (3,10) node[above=1pt]{$k$} -- (3,0);
\filldraw[fill=white] (3,7.7) circle (.25);
\filldraw[fill=white] (7.7,3) circle (.25);
\draw[dotted] (0,9) node[left=1pt]{$\bl$} -- (10,9) ;
\draw[dotted] (9,10) node[above=1pt]{$l$} -- (9,0);
\filldraw[fill=black] (9,3) circle (.25);
\filldraw[fill=black] (3,9) circle (.25);
\etp
\]
Relation \eqref{CK:S3:1a} with $l=\bi$ implies that in the first two cases we have $\key_{i\tsp\bi}(u)/\key_{\bk l}(u)\in\C$. Since $\key_{i\tsp\bi}(u)=1$, it follows that $\key_{\bk l}(u)\in\C$. In a similar way, relation \eqref{CK:S3:1b} with $l=\bi$ implies that $\key_{i\tsp\bi}(u)/\key_{\bl k}(u)\in\C$ and so $\key_{\bl k}(u)\in\C$. In the third case \eqref{CK:S3:1a} implies $\key_{\bk l}(u)\in\C u^{-1}$, while \eqref{CK:S3:1b} gives $\key_{\bl k}(u)\in\C u$. Applying these arguments to \eqref{CK:S3:2} we find 
\[
\key_{\bk l}(u) = \Big( {-}(-q)^{k-l+1}\del_{l<k} + (-q)^{k-l-1} \del_{k<l<\bi} + (-q)^{k+\bl} u^2\del_{\bi<l} \Big)\key_{\bl k}(u) .
\]
On the other hand, \eqref{CK:S1:2} gives
\[
\key_{\bk l}(u) = \Big( (-q)^{k-l-1}\del_{l<k} + (-q)^{k-l+1} \del_{k<l<\bi} + (-q)^{k+\bl-2} u^2\del_{\bi<l} \Big)\key_{\bl k}(u) 
\]
yielding a contradiction. Assuming (above \eqref{CK:S3:2}) that $i>\bk$ instead and repeating the same steps as before we obtain a similar contradiction. Hence, if $\key_{ik}=0$ and $\key_{i\tsp\bi}(u)=1$, then the only nonzero entry in $\bk$-th line and $k$-th column of $K(u)$ is $\key_{\bk k}(u)=1$. We indicate this result by
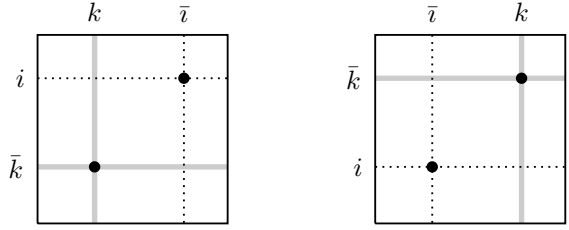
\[
\btp
\draw[line width=2,gray!40] (0,3) node[left=1pt]{\color{black}$\bk$} -- (10,3) ;
\draw[line width=2,gray!40] (3,10) node[above=1pt]{\color{black}$k$} -- (3,0);
\draw (0,0) -- (10,0) -- (10,10) -- (0,10) -- (0,0);
\draw[dotted] (0,7.7) node[left=1pt]{$i$} -- (10,7.7) ;
\draw[dotted] (7.7,10) node[above=1pt]{$\bi$} -- (7.7,0);
\filldraw[fill=black] (7.7,7.7) circle (.25);
\filldraw[fill=black] (3,3) circle (.25);
\etp
\qq\qq
\btp
\draw[line width=2,gray!40] (0,7.7) node[left=1pt]{\color{black}$\bk$} -- (10,7.7) ;
\draw[line width=2,gray!40] (7.7,10) node[above=1pt]{\color{black}$k$} -- (7.7,0);
\draw (0,0) -- (10,0) -- (10,10) -- (0,10) -- (0,0);
\draw[dotted] (0,3) node[left=1pt]{$i$} -- (10,3) ;
\draw[dotted] (3,10) node[above=1pt]{$\bi$} -- (3,0);
\filldraw[fill=black] (3,3) circle (.25);
\filldraw[fill=black] (7.7,7.7) circle (.25);
\etp
\]
when $i<\bk$ and $i>\bk$, respectively. Here the grey lines indicate that the matrix entries in that line or column are all equal to zero except $\key_{\bk k}(u)=1$. By combining both cases and repeating the same arguments line-by-line (and consequently column-by-column) and using the fact that each line and each column of $K(u)$ must have at least one nonzero element we conclulde that, if there exists $i,j$ such that $i\ne \bj$, $\key_{ij}(u)=0$ and $\key_{i\tsp\bi}(u)=1$ or $\key_{\bi\tsp i}(u)=1$, then $K(u)$ is the constant antidiagonal matrix \eqref{CK:A1}.


{\noindent \it Step 4: Suppose there exists $i$ such that $\key_{i\tsp\bi}(u)=0$. Then $K(u)$ is given by \eqref{CK:A2} or \eqref{CK:A4}.
}


Set $i=j$ in \eqref{CK1} and assume that $\key_{i\tsp\bi}(u)=0$. This gives
\[
\left(a(\tfrac{u}{v})-1\right) \key_{ik}(u)\,\key_{il}(v) + b_{kl}(\tfrac{u}{v})\, \key_{ik}(v)\, \key_{il}(u) = 0.
\] 
Upon setting $v=-q^{-1}u$ this yields $\key_{ik}(-q^{-1}u)\,\key_{il}(u)=0$. Now set $k=l$ in \eqref{CK1} and assume that $\key_{\bk k}(u)=0$. We now have
\[
\left(a(\tfrac{u}{v})-1\right) \key_{ik}(u)\,\key_{jk}(v) + b_{ij}(\tfrac{u}{v})\,\key_{ik}(v)\,\key_{jk}(u) = 0.
\]
Upon setting $v=-q^{-1}u$ this yields $\key_{ik}(-q^{-1}u)\,\key_{jk}(u)=0$. Therefore, if $\key_{i\tsp\bi}(u)=0$, then there exists only one nonzero entry in the $i$-th line and in the $\bi$-th column of $K(u)$. Arguments used in Step~3 further imply that $\key_{i\tsp\bi}(u)=0$ for all $i$. Hence $K(u)$ is a generalized permutation matrix with nonzero entries distributed symmetrically with respect to the main antidiagonal (this follows from Step 1) and $N$ is an even positive integer greater than 2.  


We will use \eqref{CK:S3:1a}, \eqref{CK:S3:1b} and \eqref{CK:S1:2} to determine the distribution of the nonzero entries of $K(u)$. Let $i,j,k,l$ be all different and let $\key_{i\tsp\bj}(u)$, $\key_{j\tsp\bi}(u)$, $\key_{k\bl}(u)$ and $\key_{l\bk}(u)$ be all nonzero. According to \eqref{CK:S3:1a} and \eqref{CK:S3:1b} there can be three inequivalent configuration, namely when $i<j<k<l$, $i<k<j<l$ and $i<k<l<j$. We represent these configurations as follows:
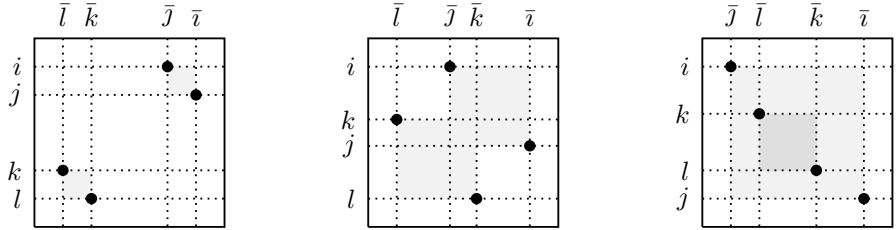
\[
\btp
\fill[gray!10] (7,8.5) rectangle (8.5,7);
\fill[gray!10] (1.5,3) rectangle (3,1.5);
\draw[dotted] (0,8.5) node[left=1pt]{\color{black}$i$} -- (10,8.5) ;
\draw[dotted] (8.5,10) node[above=1pt]{\color{black}$\bi$} -- (8.5,0);
\draw[dotted] (0,7) node[left=1pt]{\color{black}$j$} -- (10,7) ;
\draw[dotted] (7,10) node[above=1pt]{\color{black}$\bj$} -- (7,0);
\draw[dotted] (0,3) node[left=1pt]{\color{black}$k$} -- (10,3);
\draw[dotted] (3,10) node[above=1pt]{\color{black}$\bk$} -- (3,0);
\draw[dotted] (0,1.5) node[left=1pt]{\color{black}$l$} -- (10,1.5);
\draw[dotted] (1.5,10) node[above=1pt]{\color{black}$\bl$} -- (1.5,0);
\draw (0,0) -- (10,0) -- (10,10) -- (0,10) -- (0,0);
\filldraw[fill=black] (7,8.5) circle (.25);
\filldraw[fill=black] (8.5,7) circle (.25);
\filldraw[fill=black] (1.5,3) circle (.25);
\filldraw[fill=black] (3,1.5) circle (.25);
\etp
\qq\qq
\btp
\fill[gray!10] (4.3,8.5) rectangle (8.5,4.3);
\fill[gray!10] (1.5,5.7) rectangle (5.7,1.5);
\draw[dotted] (0,8.5) node[left=1pt]{\color{black}$i$} -- (10,8.5) ;
\draw[dotted] (8.5,10) node[above=1pt]{\color{black}$\bi$} -- (8.5,0);
\draw[dotted] (0,4.3) node[left=1pt]{\color{black}$j$} -- (10,4.3) ;
\draw[dotted] (4.3,10) node[above=1pt]{\color{black}$\bj$} -- (4.3,0);
\draw[dotted] (0,5.7) node[left=1pt]{\color{black}$k$} -- (10,5.7);
\draw[dotted] (5.7,10) node[above=1pt]{\color{black}$\bk$} -- (5.7,0);
\draw[dotted] (0,1.5) node[left=1pt]{\color{black}$l$} -- (10,1.5);
\draw[dotted] (1.5,10) node[above=1pt]{\color{black}$\bl$} -- (1.5,0);
\draw (0,0) -- (10,0) -- (10,10) -- (0,10) -- (0,0);
\filldraw[fill=black] (4.3,8.5) circle (.25);
\filldraw[fill=black] (8.5,4.3) circle (.25);
\filldraw[fill=black] (1.5,5.7) circle (.25);
\filldraw[fill=black] (5.7,1.5) circle (.25);
\etp
\qq\qq
\btp
\fill[gray!10] (1.5,8.5) rectangle (8.5,1.5);
\fill[gray!25] (6,3) rectangle (3,6);
\draw[dotted] (0,8.5) node[left=1pt]{\color{black}$i$} -- (10,8.5) ;
\draw[dotted] (8.5,10) node[above=1pt]{\color{black}$\bi$} -- (8.5,0);
\draw[dotted] (0,6) node[left=1pt]{\color{black}$k$} -- (10,6) ;
\draw[dotted] (6,10) node[above=1pt]{\color{black}$\bk$} -- (6,0);
\draw[dotted] (0,3) node[left=1pt]{\color{black}$l$} -- (10,3);
\draw[dotted] (3,10) node[above=1pt]{\color{black}$\bl$} -- (3,0);
\draw[dotted] (0,1.5) node[left=1pt]{\color{black}$j$} -- (10,1.5);
\draw[dotted] (1.5,10) node[above=1pt]{\color{black}$\bj$} -- (1.5,0);
\draw (0,0) -- (10,0) -- (10,10) -- (0,10) -- (0,0);
\filldraw[fill=black] (1.5,8.5) circle (.25);
\filldraw[fill=black] (8.5,1.5) circle (.25);
\filldraw[fill=black] (3,6) circle (.25);
\filldraw[fill=black] (6,3) circle (.25);
\etp
\]
Note that the only nonzero matrix entries in each dotted line and each dotted column are the ones indicated by the filled circles. The grey shadings are used to emphasize the differences between the configurations.

We start by focusing on the first case, when $i<j<k<l$. Relation \eqref{CK:S3:1a} implies that $\key_{i\tsp\bj}(u)/\key_{l\bk}(u)\in\C$, while \eqref{CK:S3:1b} gives $\key_{j\tsp\bi}(u)/\key_{k\bl}(u)\in\C$. By requiring these to be compatible with \eqref{CK:S1:2} we deduce that 
\equ{
\key_{j\tsp\bi}(u)=(-q)^{i-j+1}\key_{i\tsp\bj}(u), \qq \key_{l\bk}(u)=(-q)^{k-l+1}\key_{k\bl}(u).  \label{CK:S4:1a}
} 
Hence, without loss of generality, we may assume that 
\equ{
\key_{i\tsp\bj}(u), \key_{j\tsp\bi}(u), \key_{k\bl}(u), \key_{l\bk}(u) \in\C. \label{CK:S4:1}
}
In the second case, when $i<k<j<l$, by similar arguments as above we deduce that, without loss of generality,
\equ{
\key_{i\tsp\bj}(u), \key_{k\bl}(u) \in \C u, \qu \key_{j\tsp\bi}(u), \key_{l\bk}(u) \in\C, \label{CK:S4:2}
}
and in particular, 
\equ{
\key_{j\tsp\bi}(u)=\pm(-q)^{i-j}\tq u^{-1}\key_{i\tsp\bj}(u) , \qq
\key_{l\bk}(u)=\pm(-q)^{k-l}\tq u^{-1}\key_{k\bl}(u). \label{CK:S4:2a}
}
Finally, in the third case, when $i<k<l<j$, we find that, without loss of generality,
\equ{
\key_{i\tsp\bj}(u) \in \C u, \qu \key_{k\bl}(u) , \key_{l\bk}(u) \in\C, \qu \key_{j\tsp\bi}(u) \in \C u^{-1},  \label{CK:S4:3}
}
and in particular, 
\equ{
\key_{j\tsp\bi}(u)=(-q)^{i-j-1}\tq^2 u^{-2}\key_{i\tsp\bj}(u) , \qq
\key_{l\bk}(u)=(-q)^{k-l+1}\key_{k\bl}(u) . \label{CK:S4:3a}
}

If $N=4$ then the relations above are enough determine $K$.
If $N\ge 6$ we need additional arguments.
We return to the first case, when $i<j<k<l$. Let $r,s$ be such that $\key_{r\bar s}(u)$ and $\key_{s\tsp\bar r}(u)$ are also nonzero. Then \eqref{CK:S3:1a}, \eqref{CK:S3:1b} and \eqref{CK:S1:2} imply that there are two inequivalent configurations, when $i<j<r<s<k<l$ with $\key_{r\bar s}(u),\key_{s\tsp\bar r}(u)\in\C$, and when $r<i<j<k<l<s$ with $\key_{r\bar s}(u)\in \C u$ and $\key_{s\tsp\bar r}(u)\in \C u^{-1}$. These cases correspond to the following two diagrams, respectively,
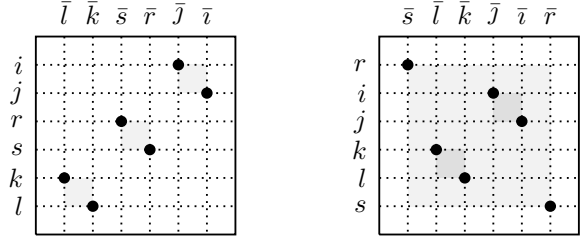
\[
\btp
\fill[gray!10] (9,7.5) rectangle (7.5,9);
\fill[gray!10] (4.5,6) rectangle (6,4.5);
\fill[gray!10] (1.5,3) rectangle (3,1.5);
\draw[dotted] (0,9) node[left=1pt]{\color{black}$i$} -- (10.5,9) ;
\draw[dotted] (9,10.5) node[above=1pt]{\color{black}$\bi$} -- (9,0);
\draw[dotted] (0,7.5) node[left=1pt]{\color{black}$j$} -- (10.5,7.5) ;
\draw[dotted] (7.5,10.5) node[above=1pt]{\color{black}$\bj$} -- (7.5,0);
\draw[dotted] (0,6) node[left=1pt]{\color{black}$r$} -- (10.5,6);
\draw[dotted] (6,10.5) node[above=1pt]{\color{black}$\bar r$} -- (6,0);
\draw[dotted] (0,4.5) node[left=1pt]{\color{black}$s$} -- (10.5,4.5);
\draw[dotted] (4.5,10.5) node[above=1pt]{\color{black}$\bar s$} -- (4.5,0);
\draw[dotted] (0,3) node[left=1pt]{\color{black}$k$} -- (10.5,3);
\draw[dotted] (3,10.5) node[above=1pt]{\color{black}$\bk$} -- (3,0);
\draw[dotted] (0,1.5) node[left=1pt]{\color{black}$l$} -- (10.5,1.5);
\draw[dotted] (1.5,10.5) node[above=1pt]{\color{black}$\bl$} -- (1.5,0);
\draw (0,0) -- (10.5,0) -- (10.5,10.5) -- (0,10.5) -- (0,0);
\filldraw[fill=black] (7.5,9) circle (.25);
\filldraw[fill=black] (9,7.5) circle (.25);
\filldraw[fill=black] (4.5,6) circle (.25);
\filldraw[fill=black] (6,4.5) circle (.25);
\filldraw[fill=black] (1.5,3) circle (.25);
\filldraw[fill=black] (3,1.5) circle (.25);
\etp
\qq\qq
\btp
\fill[gray!10] (9,1.5) rectangle (1.5,9);
\fill[gray!25] (7.5,6) rectangle (6,7.5);
\fill[gray!25] (4.5,3) rectangle (3,4.5);
\draw[dotted] (0,7.5) node[left=1pt]{\color{black}$i$} -- (10.5,7.5) ;
\draw[dotted] (7.5,10.5) node[above=1pt]{\color{black}$\bi$} -- (7.5,0);
\draw[dotted] (0,6) node[left=1pt]{\color{black}$j$} -- (10.5,6) ;
\draw[dotted] (6,10.5) node[above=1pt]{\color{black}$\bj$} -- (6,0);
\draw[dotted] (0,9) node[left=1pt]{\color{black}$r$} -- (10.5,9);
\draw[dotted] (9,10.5) node[above=1pt]{\color{black}$\bar r$} -- (9,0);
\draw[dotted] (0,1.5) node[left=1pt]{\color{black}$s$} -- (10.5,1.5);
\draw[dotted] (1.5,10.5) node[above=1pt]{\color{black}$\bar s$} -- (1.5,0);
\draw[dotted] (0,4.5) node[left=1pt]{\color{black}$k$} -- (10.5,4.5);
\draw[dotted] (4.5,10.5) node[above=1pt]{\color{black}$\bk$} -- (4.5,0);
\draw[dotted] (0,3) node[left=1pt]{\color{black}$l$} -- (10.5,3);
\draw[dotted] (3,10.5) node[above=1pt]{\color{black}$\bl$} -- (3,0);
\draw (0,0) -- (10.5,0) -- (10.5,10.5) -- (0,10.5) -- (0,0);
\filldraw[fill=black] (6,7.5) circle (.25);
\filldraw[fill=black] (7.5,6) circle (.25);
\filldraw[fill=black] (1.5,9) circle (.25);
\filldraw[fill=black] (9,1.5) circle (.25);
\filldraw[fill=black] (3,4.5) circle (.25);
\filldraw[fill=black] (4.5,3) circle (.25);
\etp
\]
All the remaining cases are either equivalent to the two above or are not allowed. For example, the case $i<j<k<l<r<s$, after renaming $k\leftrightarrow r$ and $l\leftrightarrow s$, gives $i<j<r<s<k<l$. As a second example consider the case, when $i<r<j<s<k<l$. Such a configuration is not allowed since then $\key_{i\tsp\bj}(u), \key_{r\bar s}(u)\in\C u$ and $\key_{i\tsp\bj}(u),\key_{r\bar s}(u)\in\C$ yielding a contradiction.

We now focus on the second case, when $i<k<j<l$. As before, we assume that $r,s$ are such that $\key_{r\bar s}(u)$ and $\key_{s\tsp\bar r}(u)$ are nonzero. Then there is only one configuration we need to consider, when $i<r<k<j<s<l$, yielding $\key_{r\bar s}(u)\in \C u$ and $\key_{s\tsp\bar r}(u)\in \C$. (All the remaining cases are either equivalent to this one or are not allowed.) Diagrammatically this configuration is as follows,
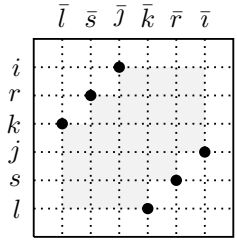
\[
\btp
\fill[gray!10] (4.5,9) rectangle (9,4.5);
\fill[gray!10] (3,7.5) rectangle (7.5,3);
\fill[gray!10] (1.5,6) rectangle (6,1.5);
\draw[dotted] (0,9) node[left=1pt]{\color{black}$i$} -- (10.5,9) ;
\draw[dotted] (9,10.5) node[above=1pt]{\color{black}$\bi$} -- (9,0);
\draw[dotted] (0,4.5) node[left=1pt]{\color{black}$j$} -- (10.5,4.5) ;
\draw[dotted] (4.5,10.5) node[above=1pt]{\color{black}$\bj$} -- (4.5,0);
\draw[dotted] (0,7.5) node[left=1pt]{\color{black}$r$} -- (10.5,7.5);
\draw[dotted] (7.5,10.5) node[above=1pt]{\color{black}$\bar r$} -- (7.5,0);
\draw[dotted] (0,3) node[left=1pt]{\color{black}$s$} -- (10.5,3);
\draw[dotted] (3,10.5) node[above=1pt]{\color{black}$\bar s$} -- (3,0);
\draw[dotted] (0,6) node[left=1pt]{\color{black}$k$} -- (10.5,6);
\draw[dotted] (6,10.5) node[above=1pt]{\color{black}$\bk$} -- (6,0);
\draw[dotted] (0,1.5) node[left=1pt]{\color{black}$l$} -- (10.5,1.5);
\draw[dotted] (1.5,10.5) node[above=1pt]{\color{black}$\bl$} -- (1.5,0);
\draw (0,0) -- (10.5,0) -- (10.5,10.5) -- (0,10.5) -- (0,0);
\filldraw[fill=black] (4.5,9) circle (.25);
\filldraw[fill=black] (9,4.5) circle (.25);
\filldraw[fill=black] (3,7.5) circle (.25);
\filldraw[fill=black] (7.5,3) circle (.25);
\filldraw[fill=black] (1.5,6) circle (.25);
\filldraw[fill=black] (6,1.5) circle (.25);
\etp
\]

It remains to consider the third case, when $i<k<l<j$. Once again, we assume that $r,s$ are such that $\key_{r\bar s}(u)$ and $\key_{s\tsp\bar r}(u)$ are nonzero. The only allowed configuration is $i<r<s<k<l<j$ with $\key_{r\bar s}(u),\key_{r\bar s}(u)\in\C$. (Again, all the remaining cases are either equivalent to this one or are not allowed.) Note that this configuration was already discussed above, i.e., it is equivalent to the case, when $r<i<j<k<l<s$. Also note that in each case above the allowed configurations are of even dimension.

The considerations above together with the requirement for $K(u)$ to be invertible imply that if there exists $i$ such that $\key_{i\,\bi}(u)=0$ then $K(u)$ is given by 
\[
\sum_{1\le i \le N/2} \Big(\key_{2i-1,\overline{2i}}(u)\,E_{2i-1,\overline{2i}} + \key_{2i,\overline{2i}+1}(u)\,E_{2i,\overline{2i}+1} \Big) 
\]
with $\key_{2i-1,\overline{2i}}(u),\key_{2i,\overline{2i}+1}(u)\in \C$, or by
\[
\key_{11}(u)\,\,E_{11} + \key_{NN}(u)\, E_{NN} + \sum_{1\le i < N/2} \Big( \key_{2i,\overline{2i}-1}(u)\, E_{2i,\overline{2i}-1} + \key_{2i,\overline{2i}+1}(u)\, E_{2i,\overline{2i}+1} \Big) 
\]
with $\key_{11}(u)\in \C u$, $\key_{NN}(u) \in \C u^{-1}$ and $\key_{2i-1,\overline{2i}}(u),\key_{2i,\overline{2i}+1}(u)\in \C$, or by
\[
\sum_{1\le i \le N/2} \Big( \key_{i,\bi-N/2}(u)\,E_{i,\bi-N/2} + \key_{i+N/2,\bi}(u)\,E_{i+N/2,\bi} \Big) 
\]
with $\key_{i,\bi-N/2}(u)\in \C u$ and $\key_{i+N/2,\bi}(u)\in \C$. The ratios between the entries symmetric with respect to the main antidiagonal follow from \eqref{CK:S4:1a}, \eqref{CK:S4:2a} and \eqref{CK:S4:3a} yielding matrices \eqref{CK:[A2a]}, \eqref{CK:[A2b]} and \eqref{CK:[A4]}, respectively. By Lemma \ref{L:[CK]} this completes the proof.
\end{proof}


\begin{rmk} [{{\it Affinization procedure for \eqref{CtRE}}}] \label{R:Aff}
Upon dividing by $(1 + \tq^{-1}q\,u)$ the solution \eqref{CK:qOns} may be written as
\equ{
(1+\tq^{-1} q\,u)J + L + \tq^{-1} u\,C^{-1} L^t C^t , \label{CK-aff}
}
where
\equ{
J = \sum_{1\le i \le N} E_{i\tsp\bi} , \qq L = \sum_{1\le i < j \le N} (1+q) (-q)^{\frac{i-j}2} E_{j\tsp\bi}. \label{JD}
}
Similarly, the solution \eqref{CK:A4} may be written as
\equ{
G + \tq^{-1} u\, C^{-1}G^t C^t \qq G = \sum_{1\le i \le N/2} E_{i+N/2,\bi}. \label{KIV-aff}
}
The above identities may be viewed as analogues of the affinization procedure for the twisted reflection equation \eqref{CtRE}. The matrices $J+L$ and $G$ defined above and the matrices \eqref{CK:A1} and \eqref{CK:A2} are solutions of the constant braided ($C$-conjugated) twisted reflection equation
\equ{
\check{R}_q G_2 (\check{R}_q^{\vee})^{-1} G_2 = G_2 (\check{R}_q^{\vee})^{-1} G_2 \check{R}_q . \label{CtRE-const}
}
This can be seen by taking the $u,v\to 0$ limit of \eqref{CtRE}. 
More generally, if a matrix $G\in\End(\C^N)$ is a solution of \eqref{CtRE-const}, then it also a solution of \eqref{CtRE} if it has nonzero entries both above and below the main antidiagonal. 
Otherwise, if $G$ has no nonzero entries above the main antidiagonal, then the matrix \eqref{CK-aff}, where $J$ is the antidiagonal part of $G$ and $L:=G-J$, is a solution of \eqref{CtRE}. 
The matrices \eqref{CK:A1} and \eqref{CK:A2} and the matrix $J+L$ given by \eqref{JD} are the only invertible solutions of \eqref{CtRE-const} subject to an equivalence relation $G \sim D^w G D$, where $D$ is a \cid matrix, cf.\ Lemma \ref{L:K-symm} (iii). 
A classification of non-invertible solutions of \eqref{CtRE-const} will be presented elsewhere. 
The non-invertible solutions of \eqref{CtRE-const}, with the exception of the matrix $G$ given in \eqref{KIV-aff}, yield non-invertible solutions of \eqref{CtRE}.  \hfill\rmkend
\end{rmk}

\begin{rmk}[{{\it Unitarity\,}}]
A solution $K$ of \eqref{RE} or \eqref{CtRE} is called \emph{unitary} if $K(u)\,K(u^{-1})=I$ for generic $u$. 
Solutions defined by \eqrefs{KS}{KT} and \eqrefs{CK:A1}{CK:A4} are unitary. 
Upon multiplication by $(-q)^{1/2}(q+\tq\,u)(1-\tq\,u)^{-1}$ the solution \eqref{CK:qOns} also becomes unitary.  \hfill \rmkend
\end{rmk}

\begin{rmk}[{{\it Regularity\,}}] 
A solution $K$ of \eqref{RE} or \eqref{CtRE} is called \emph{regular} if $K(1)$ or $K(-1)$ equals~$I$.
Regularity is a requirement for defining quantum spin chain Hamiltonians with reflecting boundary conditions as derivatives of transfer matrices in the usual way, see \cite{Sk} for the untwisted case. 
All solutions of \eqref{RE} are regular unless $\la^2 = \mu^2 = 1$ in the solution given by \eqrefs{KS1}{KS2}. 
In the latter case $K(u)$ is a generalized involution matrix, cf.~Remark \ref{R:spec}.
None of the solutions of \eqref{CtRE} are regular, except the very special case of \eqref{CK:[A2b]} with $N=4$. 
For solutions of \eqref{CtRE} there exists a special construction to define a quantum spin chain Hamiltonian, see \cite{Do}. To our best knowledge this is the only known construction of Hamiltonians associated to non-regular solutions of \eqref{CtRE}. 
\hfill \rmkend
\end{rmk}


\end{document}